\documentclass{article}

\usepackage{microtype}
\usepackage{graphicx}
\usepackage{subcaption}
\usepackage{booktabs} 
\usepackage{float}
\usepackage{array}

\usepackage{makecell}
\usepackage[table]{xcolor}
\definecolor{ESLOrange}{RGB}{255,127,0}
\definecolor{ICMLBlue}{RGB}{0,0,128}

\usepackage[
  backref=page,
  colorlinks=true,
  linkcolor=ICMLBlue,
  citecolor=ICMLBlue
]{hyperref}

\usepackage{mdframed}
\mdfsetup{
  linewidth=0.8pt,
  roundcorner=4pt,
}
\usepackage{witharrows}

\usepackage{multirow}
\usepackage{booktabs} 

\usepackage[preprint]{icml2026}

\usepackage{amsmath}
\usepackage{amssymb}
\usepackage{mathtools}
\usepackage{amsthm}

\usepackage[capitalize,noabbrev]{cleveref}

\theoremstyle{plain}

\newtheorem{theorem}{Theorem}

\newtheorem{fact}[theorem]{Fact}

\newtheorem{lemma}[theorem]{Lemma}
\newtheorem{corollary}[theorem]{Corollary}
\newtheorem{definition}[theorem]{Definition}
\newtheorem{assumption}[theorem]{Assumption}
\newtheorem{remark}[theorem]{Remark}

\usepackage[textsize=tiny]{todonotes}
\usepackage{subcaption} 

\icmltitlerunning{Faster $k$-means Seeding Under The Manifold Hypothesis}
\usepackage{mathrsfs}

\newcommand{\opt}{\operatorname{opt}}

\newcommand{\cost}{\operatorname{cost}}
\newcommand{\Insert}{\textsc{Insert}}
\newcommand{\Query}{\textsc{Query}}
\newcommand{\eps}{\varepsilon}
\newcommand{\E}{\mathbf{E}}
\renewcommand{\Pr}{\mathbf{Pr}}
\newcommand{\N}{\mathbf{N}}
\newcommand{\R}{\mathbb{R}}
\newcommand{\U}{\mathcal{U}}
\renewcommand{\H}{\mathcal{H}}
\newcommand{\M}{\mathcal{M}}

\newcommand{\X}{\mathcal{X}}
\newcommand{\Y}{\mathcal{Y}}

\newcommand{\dvol}{\operatorname{dvol}}
\newcommand{\vol}{\operatorname{vol}}
\newcommand{\diam}{\operatorname{diam}}
\newcommand{\defeq}{\overset{\mathrm{def}}{=}}

\renewcommand{\L}{\mathscr{L}}

\newcommand{\F}{\mathtt{F}}
\renewcommand{\rho}{\varrho}
\newcommand{\afkmc}{\operatorname{AFKMC}^2}
\newcommand{\rejsample}{\operatorname{RejectionSampling}}

\newcommand{\pronecoreset}{\operatorname{PRONECoreset}}
\newcommand{\fastkmeans}{\operatorname{FastKMeans}}

\newcommand{\qkmeans}{\operatorname{Qkmeans}}
\renewcommand{\phi}{\varphi}
\newcommand{\C}{\mathcal{C}}
\newcommand{\rownumber}{\stepcounter{rownum}\arabic{rownum}}

\begin{document}

\twocolumn[
  \icmltitle{Fast $k$-means Seeding Under The Manifold Hypothesis }



  \icmlsetsymbol{equal}{*}

  \begin{icmlauthorlist}
    \icmlauthor{Poojan Shah}{yyy}
    \icmlauthor{Shashwat Agrawal}{yyy}
    \icmlauthor{Ragesh Jaiswal}{yyy}

  \end{icmlauthorlist}

  \icmlaffiliation{yyy}{Department of Computer Science and Engineering, Indian Institute of Technology Delhi}

  \icmlcorrespondingauthor{Poojan Shah}{shahpoojan2004@gmail.com}

  \icmlkeywords{Machine Learning, ICML}

  \vskip 0.3in
]



\printAffiliationsAndNotice{}  

\begin{abstract}
    We study \emph{beyond worst case analysis} for the $k$-means problem where the goal is to model typical instances of $k$-means arising in practice. Existing theoretical approaches provide guarantees under certain assumptions on the optimal solutions to $k$-means, making them difficult to validate in practice. We propose the \emph{manifold hypothesis}, where data obtained in ambient dimension $D$ concentrates around a low dimensional manifold of intrinsic dimension $d \ll D$, as a reasonable assumption to model real world clustering instances. We identify key geometric properties of datasets which have theoretically predictable scaling laws depending on the \emph{quantization exponent} $\eps = 2/d$ using techniques from optimum quantization theory. We show how to exploit these regularities to design a fast seeding method called $\qkmeans$\footnote{The $\operatorname{Q}$ stands for Quantization, as the analysis of our algorithm was motivated by results from optimum quantization theory.} which provides $O(\rho^{-2} \log k)$ approximate solutions to the 
    $k$-means problem in time $O(nD) + \widetilde{O}(\eps ^{1+\rho}\rho^{-1}k^{1+\gamma})$; where the exponent $\gamma = \eps + \rho$ for an input parameter $\rho < 1$. This allows us to obtain new runtime - quality tradeoffs. We perform a large scale empirical study across various domains to validate our theoretical predictions and algorithm performance to bridge theory and practice for beyond worst case data clustering.  
\end{abstract}

\section{Introduction}
Data clustering is a fundamental task in machine learning and data science. The goal is to partition a dataset into subsets so that points within the same subset (or cluster) are more similar to each other than to points in other clusters. Such partitions help uncover latent structure and compress data into a small number of interpretable groups. Among the many approaches to clustering, the $k$-means formulation is perhaps the most widely used. In this formulation, we are given a  dataset $\X = \{x_1,\dots,x_n\} \subset \R^D$, and our task is to find a set of $k$ centers $C = \{c_1,\dots,c_k\} \subset \R^D$ which minimize the total within-cluster variance:
\[
  \cost(\X,C) \defeq \sum_{x \in \X} \min_{c \in C} \|x-c\|^2.
\]
These centers induce a partition $\X = \bigcup_{j=1}^k C_j$, where $C_j$ contains the points in the dataset $\X$ closest to $c_j$. An optimal set $C_k^*$ achieves cost $\opt_k(\X) = \cost(\X,C_k^*)$. Optimizing the $k$-means cost then serves as a proxy for uncovering some natural clustering of the dataset. This formulation has become a key component of algorithm design in several fields such as vector quantization and compression \citep{Linde1980}, approximate nearest neighbor search \cite{Jegou2011,Johnson2017} along with numerous others.\\ 

Given the immense importance of $k$-means in modern data processing, it is natural that the problem has been thoroughly investigated from both theoretical and empirical perspectives. The $k$-means problem is known to be $\mathsf{NP}$-hard \citep{dasgupta_08,mahajan_09}, so it is unlikely to have a polynomial time algorithm. Much research has been done on designing polynomial time {\em approximation schemes}\footnote{These were developed in a long line of research by \cite{kumar_10,jaiswal_14,jaiswal_15,cohen-addad_18,friggstad_19,cohen-addad_19,bhattacharya_20}} for the $k$-means problem. 
However, the algorithm most popularly used in practice\footnote{Ranked as the second most popular data mining algorithm in \citep{wu2008top10}.}  to solve $k$-means instances is a heuristic known as the $k$-means algorithm ({\it not to be confused with the $k$-means problem}).
This heuristic, also known as Lloyd's iterations \citep{lloyd1982least}, iteratively improves the solution in several rounds. 
The heuristic starts with an arbitrarily chosen set of $k$ centers. In every iteration, it (i) partitions the points based on the nearest center and (ii) updates the center set to the centroids of the $k$ partitions.
In the classical computational model, it is easy to see that every Lloyd iteration costs $O(nkD)$ time.
This hill-climbing approach may get stuck at a local minimum or take a huge amount of time to converge, and hence, it does not give provable guarantees on the quality of the final solution or the running time \citep{dasgupta_03,har-peled_sadri_05,arthur_vassilvitskii_06a,arthur_vassilvitskii_06b}.
In practice, Lloyd's iterations are usually preceded by the $k$-means++ algorithm \cite{arthur2007kmeanspp}, a fast sampling-based approach for picking the initial $k$ centers that also gives an approximation guarantee.
So, Lloyd's iterations, preceded by the $k$-means++ algorithm, give the best of both worlds, theory and practice. Hence, it is unsurprising that a lot of work has been done on these two algorithms. 
The work ranges from efficiency improvements in specific settings to implementations in distributed and parallel models.\\ 

\textbf{Beyond worst case analysis.} A large fraction of theoretical algorithms for $k$-means are designed to be used as \emph{off the shelf} algorithms in a \emph{data oblivious} manner. This means that they assume that the input dataset $\X$ can be any arbitrary set of points in $\R^D$ and provide worst case guarantees. The pervasive issue is that algorithms designed to resist pathological instances, i.e. inputs that maximize resource usage-often yield complexity bounds far removed from their typical empirical performance. The realization that real datasets rarely resemble these theoretical worst-case examples provides the core motivation for moving beyond pessimistic, data-oblivious constraints \citep{roughgarden2019beyond,balcan2021datadriven,agarwal2015approx}. This \emph{data driven} paradigm acts as a necessary bridge between theoretical analysis and practical applicability and so has gained attention in the algorithm design community. \\ 

For $k$-means in particular, this comprises of several approaches. The performance of Lloyd-type methods and $k$-means++  under separability conditions was analyzed by \cite{ostrovsky2013effectiveness,jaiswal2012separable}. An instance $(\X,k)$ of the $k$-means problem is said to be $\epsilon$-separated if $\frac{\opt_k(\X)}{\opt_{k-1}(\X)} \leq \epsilon $. This definition is motivated by the usual practice of using the \emph{elbow method} to determine a good value of $k$ to be used for clustering $\X$. Another line of research studies  the notion of \emph{approximation stability} in \citep{balcan2009approximate,blum2021approx,agarwal2015approx}. Approximation stability is the idea that any clustering whose cost is close to optimal must be similar to the ground truth clustering. An instance $(\X,k)$ is called $(1+\epsilon,\alpha)$-approximation stable with respect to the cost function $\cost$ and a ground truth target clustering $C_T$ if every $1+\epsilon$ approximate solution to $\cost$ is atmost $\alpha$ far from $C_T$ in the sense of fraction of points misclustered.  Other proposals include distributional stability \cite{bilu2010stability} , spectral stability  \cite{kumar2010clustering} and perturbation resilience  \cite{10.1109/FOCS.2010.36}.  \cite{cohen-addad2017local} studied the structural consequences of these notions. A common theme across these proposals is assuming some properties of the optimal $k$-means clustering, which makes an empirical study of these notions quite difficult. Indeed, it is a non trivial problem to figure out how to check if a given dataset instance follows these assumptions. Although these notions have been studied theoretically, not much work has been done in studying whether they hold (even approximately) true for real world datasets, in part due to the fact that these assumptions bestow constraints on the optimal cost, which are computationally hard to verify. \\

\begin{mdframed}

  \textbf{\textit{\underline{Motivating Questions}}}
  \emph{The above discussion naturally leads to the following questions:}
\begin{enumerate}
   \item \emph{Can we identify structural properties of real-world datasets that are computationally verifiable, are well motivated in theory and empirically hold true for diverse modern datasets ? }
    \item \emph{If so, can we exploit the existence of such structural properties to design better algorithms for $k$-means clustering?}
\end{enumerate}

\end{mdframed}

\textbf{Contributions.} In order to tackle these questions, we present the following contributions, which are discussed in detail in the next section : 
\begin{enumerate}
    \item We propose the manifold hypothesis as a theoretically grounded, empirically tested and computationally verifiable assumption for real world clustering instances. 
    \item Under this framework, we identify geometrical parameters which have predictable scaling law behaviour.
    \item We use this structure to design a fast seeding procedure through analyzing a \emph{perturbed} version of $k$-means++ sampling
    \item We perform a large scale empirical study to validate our assumptions and the performance of our seeding algorithm.
\end{enumerate}

\subsection{Our Contributions}

We propose the idea that the \emph{manifold hypothesis} is a good candidate to derive such structural properties. The manifold hypothesis states that although real-world data often reside in high-dimensional ambient spaces, their probability mass tends to concentrate near low-dimensional, smooth submanifolds \( \mathcal{M} \subset \mathbb{R}^D \) with intrinsic dimension \( d \ll D \). This hypothesis forms the conceptual bedrock for explaining the success of modern deep-learning systems \citep{bengio2013representation} and has subsequently attracted thorough theoretical and experimental investigation  \citep{fefferman2016testing,brown2022union,bordt2023manifold} and hence is widely accepted and established. We formally state this:

\begin{assumption}\label{ass} ({Manifold hypothesis assumption})
  The dataset $\X = \{x_1,\dots,x_n\} \subset \R^D$ is generated by drawing $n$ independent and identically distributed samples from a density function $f : \M \to \R_{+}$ whose support $\operatorname{supp}(f)$ is a measurable subset of a smooth, compact  $d$-dimensional submanifold $\M \subset \R^D$.
\end{assumption}

From the point of view of beyond–worst-case analysis, the manifold hypothesis offers a different type of structural regularity: instead of assuming properties about the \emph{optimal solution}, it assumes geometric regularity of the \emph{input distribution}. This makes it attractive for algorithm design, as the assumption is intrinsic to the data rather than to the algorithm’s output. Moreover, it is compatible with the increasing body of empirical evidence that high-dimensional machine learning data has low intrinsic dimension \citep{tempczyk2022lidl,horvat2022flows}. Finally, from a practical standpoint, the manifold hypothesis is more aligned with modern applications than separation- or stability-based assumptions. While the latter are tailored to idealized clustering scenarios, the manifold hypothesis is already the implicit structural model underlying essentially all deep learning architectures used in vision, NLP, physics, and generative tasks \citep{bengio2013representation,cresswell2024survey} What does this imply for the $k$-means problem ? A fundamental question is to study how the optimal $k$ means cost decreases with the number of clusters $k$. This question was explored for Euclidean spaces in an early work of \cite{zador1982asymptotic}, who showed asymptotic laws for the optimal \emph{quantization cost} given by $\Delta_k(f) = \inf_{|C| = k} \Delta(f,C)$ where $\Delta(f,C) = \int_{\M} \min_{c \in C}\|x-c\|^2 f(x)dx$. Note that this can be seen as a continuous analog to the $k$-means problem. This led to the development of the field now known as \emph{optimum quantization theory}, for a detailed exposition of which we point the readers to   \citep{gray1998quantization,graf2000foundations}. Zador's results were generalized to manifolds by \cite{gruber2004optimum}, which are of interest to us : 


  \begin{fact} \label{thm:scaling-laws-fact} \citep{gruber2004optimum}  Suppose $f$ is a probability density according to Assumption~\ref{ass} and $S_k = \{s_1,\dots,s_k\} \subset \R^D$ be an optimal $k$-quantizer of $f$ i.e., $S_k \in \arg \min _{|C| = k} \Delta(f,C)$.
  
  Define the parameters $\beta_k(f)$ and $\eta_k(f)$ as follows : 
\vspace{-2pt}
  $$ \beta_k(f) \defeq \frac{\Delta_1(f)}{\Delta_k(f)} \qquad \eta_k(f) \defeq \frac{\max_{i \neq j \in [k]} \|s_i-s_j\|}{\min_{i \neq j \in [k]} \|s_i-s_j\|} $$

  Then, the following holds : 

  $$ \beta_k(f) \in  O_f(1)\cdot k^{\eps} \qquad \eta_k(f) \in O_f(1)\cdot k^{\eps/2}$$

  where $\eps = 2/d$ is referred to as the quantization exponent and $O_f(\cdot)$ hides constants depending only on $f$. 
  
  \end{fact}

\
    We mention that such scaling laws are prevelant in several research areas. Notable examples include the scaling law between word frequency and word rank in large text documents, referred to as Zipf's law \citep{zipf1949human} and the power law distribution of node degrees in large scale networks \citep{barabasi1999emergence}. These are frequently used to guide algorithm design and algorithm analysis. We now define these parameters for finite datasets : 

\begin{definition} \label{def:parameters}
    For a dataset $\X = \{x_1,\dots,x_n\}$ and a set of centers $C = \{c_1,\cdots,c_k\}$ we define the following: 
    \[\beta(\X,C) := \frac{\cost(\X,\mu)}{\cost(\X,C)} \quad \beta_k(\X) := \frac{\opt_1(\X)}{\opt_k(\X)} \] 
    \begin{align*} 
    \eta(\X,C) &:= \frac{\max_{i\neq j} \|c_i-c_j\|}{\min_{i \neq j}\|c_i-c_j\|} \\  
    \eta(\X) &:= \frac{\max_{i\neq j} \|x_i-x_j\|}{\min_{i \neq j}\|x_i-x_j\|}
    \end{align*}

    Here, $\mu = \mu(\X)$ is the centroid of the dataset $\X$. 
\end{definition}

    Note that we always have $\beta(\X,C) \leq \beta_k(\X)$ and $\eta(\X,C) \leq \eta(\X)$ for all $C = \{c_1,\dots,c_k\}$. We show the following finite sample bounds : 

\begin{theorem} (Scaling laws) \label{thm:scaling-laws}
    Suppose $\X$ is sampled according to Assumption~\ref{ass} then the following hold with probability atleast $1 - \frac{1}{\operatorname{poly}(n)}$: 
    \begin{align*}
        \beta_k(\X) &\in \left(1+O_f\left(\sqrt{\frac{kD \log n}{n}}\right)\right) \cdot k^{\eps} \\ 
        \eta(\X) &\in O_f(1) \cdot n^{3\eps/2}
    \end{align*}
\end{theorem}

Next, we design a fast rejection-sampling based seeding scheme to accelerate $k$-means++ with novel runtime-approximation tradeoffs. Our algorithms are designed in a way which allows us to exploit the scaling laws. This allows us to improve upon existing approaches aimed at speeding up $k$-means++ under the manifold hypothesis. We formally state our main results as follows : 
  \begin{theorem} \label{thm:tradeoffs} (Main Theorem)
  Let $\X = \{x_1,\dots,x_n\} \subset \R^D$ be a dataset. Then there exists an algorithm which additionally takes as input parameters $\rho < 1, m \in \mathbb{N}$ and ouputs a set $C$ of $k$ centers such that 
  $$\E[\cost(\X,C)] \in  O(\rho^{-2})(\ln k \cdot \opt_k(\X) + k^{\frac{-\rho m}{\beta}} \cdot \opt_1(\X))$$ in  time $O(nD) + \widetilde{O}(m \log k  (k \log \eta)^{1+\rho})$ where $\beta = \beta(\X,C)$ and $\eta = \eta(\X,C)$, and $\tilde{O}$ hides terms of $\log k$ and $\log n$. 
  \end{theorem}

The above theorem holds true for \emph{all} datasets without any assumptions and hence the bounds are in terms of the data dependent parameters $\beta$ and $\eta$. We can see that for fixed $\eta,\beta$, the runtime consists of an $O(nD)$ term, which is required just for reading the input $\X$ and an additional term which is \emph{almost} linear in $k$. The approximation guarantee consists of the standard $O(\log k)$ guarantee along with an additive,scale invariant and exponentially decaying term depending on the variance of the dataset. Due to the strong decaying of the additive term, it can be removed as follows :

  \begin{corollary} \label{thm:rejection-sampling} (Rejection sampling for $k$-means++) Let $\X = \{x_1,\dots,x_n\} \subset \R^D$ be a dataset. Then there exists an algorithm which additionally takes as input a parameter $\rho < 1$ and ouputs a set $C$ of $k$ centers such that $$\E[\cost(\X,C)] \in O(\rho^{-2}\log k) \opt_k(\X)$$ in expected time $$O(nD) + O(\rho^{-1}\beta (k \log \eta)^{1+\rho})$$ where $\beta = \beta(\X,C)$ and $\eta = \eta(\X,C)$.  
  
  \end{corollary}


 Finally, we can combine the scaling laws from Theorem~\ref{thm:scaling-laws} along with Theorem~\ref{thm:tradeoffs} to obtain the following result for $k$-means clustering under the manifold hypothesis: 
  \begin{corollary} \label{thm:assumption} (Rejection sampling under manifold hypothesis) Let $\X = \{x_1,\dots,x_n\} \subset \R^D$ be a dataset generated according to Assumption~\ref{ass}. Then there exists an algorithm which additionally takes as input a parameter $\rho < 1$ and ouputs a set $C$ of $k$ centers such that $$\E[\cost(\X,C)] \in O(\rho^{-2}\log k) \opt_k(\X)$$ in expected time $$O(nD) + \widetilde{O}(\eps^{1+\rho} \rho^{-1}k^{1+ \gamma} )$$ where $\gamma = \eps + \rho$, and $\widetilde{O}(\cdot)$ hides $\log k, \log n$ factors. 
  
  \end{corollary}

\paragraph{Empirical results.}We perform extensive empirical analysis to test whether our theoretical assumptions are reasonable for practical datasets. Our study captures cross-domain datasets covering image data, image and text embeddings, as well as tabular data captured from sensors and scientific experiments. Our assumptions can be easily checked by computing the cost for various values of $k$, then evaluating the geometric properties of the resulting clusters. We find that the scaling laws  strongly manifest in these datasets. We also analyze the effect of intrinsic dimension on these laws using independent ID estimates. Finally, we observe that our seeding method gives about an order of magnitude speedup for $k$-means++ seeding on high dimensional datasets when compared to prior work.
\subsection{Comparison with prior work}

\paragraph{Speeding up k-means++.} Let us discuss the methods proposed to accelerate $k$-means++ so as to position our results in the literature. Directly implementing $k$-means++ provides an $O(\log k)$ approximation guarantee in $O(nkD)$ time, which quickly becomes prohibitively slow even for moderate sized datasets. The problem is that $k$-means++ is inherently a sequential procedure where the new center chosen depends on the previously chosen centers. $k$-means++ has also been extended to distributed settings \citep{bahmani2012scalable,bachem2017distributed,bachem2018lightweight}, but our results are for clustering on a single node.
\paragraph{MCMC methods.} In the line of work \cite{bachem2016approximate,bachem2016good}, the authors propose to approximate the $k$-means++ distribution using an appropriately designed markov chain, whose limiting distribution is equal to the $D^2$ distribution. Both of these algorithms have runtime of $\Omega(k^2)$, thus making them slow for use-cases when a dataset is required to be partitioned into thousands of clusters \footnote{This was empirically demonstrated in Section 6 of \cite{cohenaddad2020fast}. }. In comparison, $\qkmeans$ does not suffer from this $\Omega(k^2)$ dependence, but has $O(k^{1+\gamma})$ dependence where the exponent $\gamma := \eps + \rho$. As we shall see, $\eps \sim 0.1-0.2$ in real large scale datasets.  Moreover, the MCMC methods often have an additive term which decays as $O(1/m)$ where $m$ is the markov chain length, while we obtain an exponential $O(k^{-\rho m /\beta})$ decay in our results.
\paragraph{ANNS techniques. }  In \cite{cohenaddad2020fast}, the authors provide a speedup for $k$-means++ using a variety of algorithmic techniques which include hierarchical tree embeddings, approximate nearest neighbor search and rejection sampling. However, constructing this embedding brings up dependence on the aspect ratio of the dataset $\Delta$, which in general may be unbounded in terms of the input size and requires special pre-processing. They then use rejection sampling to convert samples from the $D^2$ distribution on the tree metric to a sample in the euclidean metric, which introduces an additional factor of $O(D^2)$.  In comparison, our algorithm does not require such metric embeddings and instead depends on simply sampling from an $\ell_2$-norm distribution which can be computed in a lightweight pre-processing step. Moreover, we have better dependence on the ambient dimension $O(D)$ vs $O(D^3)$ and prove improved approximation guarantee $O(\rho^{-2})$ vs $O(\rho^{-3})$. We also mention several works which try to speed up $k$-means++ or Lloyd's iterations by maintaining some nearest neighbor data structure \cite{pelleg1999accelerating,kanungo2002efficient,elkan2003using,hamerly2010making,wang2012fast,ding2015yinyang,bottesch2016speeding,newling2016fast,curtin2017dualtree}, but do not provide any theoretical guarantee on the solution quality or runtime for their proposed methods.

\paragraph{Dimensionality reduction. } In \cite{charikar2023simple}, the authors propose to project the dataset onto a randomly sampled unit vector and then perform $k$-means++ in $1$D in $O(n \log n)$ time. They show that when this 1D solution is lifted back to $D$ dimensions, it provides an $O(k^4 \log k)$ approximation guarantee. Although this is extremely fast, they report clustering costs $10$-$100\times$ larger than those provided by $k$-means++. They also propose a coreset-based method which theoretically runs in $O(nD+n\log n) + O(k^5D)$ time to provide $O(\log k)$ approximation. In comparison, our method has a much lower dependence on $k$. There is a line of research (see \cite{makarychev2019jl} and the references therein) on dimensionality reduction for $k$-means which show that the ambient dimension can be reduced to $O(\log k)$ while preserving the $k$-means cost. 

\paragraph{Coresets. } Coresets are a way to compress the dataset by computing a weighted subset of the data. We refer the reader to \cite{bachem2017practical,feldman2020introduction} for becoming familiar with coresets, which have now become ubiquitous tools for $k$-means clustering. This long line of research culminated in \cite{draganov2024settle}, where a coreset construction in near-linear time $\tilde{O}(nD \log \log \Delta)$ of size $O(k/\eps^4)$ is proposed which can then be used in conjunction with standard $k$-means++ to obtain an $O(nD \log \log \Delta + k^2D)$ algorithm having $O(\log k)$ guarantee. We conclude by providing a summary of this discussion in Table~\ref{tab:compare-results}. We mention that our results provide new theoretical tradeoffs as well as insights into the structure of input data for $k$-means clustering, along with providing a method to exploit them.

\begin{table*}[t]
\centering
\footnotesize
\setlength{\tabcolsep}{6pt}
\renewcommand{\arraystretch}{1.2}

\caption{
Comparison of time complexity and approximation guarantees for accelerating
$k$-means++ seeding.
For our method, we empirically show $\beta \sim k^{\eps}$ and
$\log \eta \sim \eps \log k$, where $\eps$ is the \emph{quantization exponent}. In the table, $\Delta$ is the aspect ratio of the dataset and $\beta,\eta$ are the parameters from Theorem~\ref{thm:scaling-laws}. Note that the relation $\eta \leq \Delta$ always holds true. Moreover, under the manifold hypothesis the overall exponent $\gamma$ is defined as $\gamma = \eps + \rho$. 
}
\label{tab:compare-results}

\begin{tabular}{@{}p{4.2cm}p{5.4cm}p{5.6cm}@{}}
\toprule
\textsc{Approach} &
\textsc{Time Complexity} &
\textsc{Approximation Guarantee} \\
\midrule

\multicolumn{3}{l}{\emph{Prior work}} \\
\midrule

\citep{bachem2016good} &
$O(nD) + O\!\left(\eps^{-1} k^2 D \log k\,\eps^{-1}\right)$ &
$O(\log k)\,\opt_k + \eps\,\opt_1$ \\

\citep{bachem2018lightweight} &
$O(nD) + O\!\left(\eps^{-2} k^2 D^2 \log k\right)$ &
$O(\log k)\,\opt_k + \eps\,\opt_1$ \\

\citep{cohenaddad2020fast} &
\makecell[l]{
$O\!\left(n(D+\log n)\log (\Delta\, D)\right)$ \\
$+\,O\!\left(\rho^{-1} k D^3 \log \Delta (n\log \Delta)^{O(\rho)}\right)$
} &
$O\!\left(\rho^{-3}\log k\right)\opt_k$ \\

\citep{charikar2023simple} &
$O(nD) + O(n\log n)$ &
$O\!\left(k^4\log k\right)\opt_k$ \\

\citep{charikar2023simple} &
\makecell[l]{
$O(nD) + O(n\log n)$ \\
$+\,O(\eps^{-2}k^5 D \log^2 k)$
} &
$O(\log k)\,\opt_k$ \\

\citep{draganov2024settle} &
$O(nD\log\log \eta) + O(\eps^{-4}k^2 D)$ &
$O(\log k)\,\opt_k$ \\

\citep{shah2025quantum} &
$O(nD) + O(\Delta^2k^2 D)$ &
$O(\log k)\,\opt_k$ \\

\midrule

\multicolumn{3}{l}{\emph{Our results}} \\
\midrule

Theorem~\ref{thm:tradeoffs} &
$O(nD) + O\!\left(mD\log k (k\log \eta)^{1+\rho}\right)$ &
$O(\rho^{-2})\!\left(\log k\,\opt_k
+ k^{-\Omega(m/\beta)}\opt_1\right)$ \\

Corollary~\ref{thm:rejection-sampling} &
$O(nD) + O\!\left( \rho^{-1}\beta D (k\log \eta)^{1+\rho}\right)$ &
$O\!\left(\rho^{-2}\log k\right)\opt_k$ \\

Corollary~\ref{thm:assumption} &
$O(nD) + O\!\left(\eps^{1+\rho}\rho^{-1}(k\log n)^{1+\gamma} D\right)$ &
$O\!\left(\rho^{-2}\log k\right)\opt_k$ \\

\bottomrule
\end{tabular}
\end{table*}



\section{Seeding Algorithm}

\begin{algorithm}[t]
\caption{\textsc{PreProcess}$(\X)$}
\label{proc:preproc}
\begin{enumerate}
    \item For $\eps_{\operatorname{JL}} \in (0,1/4)$, apply the JL-like dimensionality reduction from \cite{makarychev2019jl} in $O(nD)$ time to obtain $D \in O(\log(k/\eps_{\operatorname{JL}})/\eps_{\operatorname{JL}}^2)$
    \item Compute the mean $\mu(\X) := \frac{1}{n}\sum_ix_i$ of the dataset 
    \item Center each $x_i \gets x_i -\mu(\X)$, $1\leq i\leq n$
    \item Compute $\|\X\|_{\operatorname{F}}^2 := \sum_i \|x_i\|^2$
    \item Compute $\kappa(x_i) \gets \|x_i\|^2 / \|\X\|_{\operatorname{F}}^2$, $1 \leq i \leq n$
\end{enumerate}
\end{algorithm}

\begin{algorithm}[t]
\caption{\textsc{Sample}$(\L, C, m)$}
\label{proc:sample}
\begin{enumerate}
    \item Initialize ${iter} \gets 0$ and sample $s \sim \textsc{Uniform}(\X)$
    \item \textbf{while} ${iter} < m \ln k$ \textbf{do}
    \begin{enumerate}
        \item ${iter} \gets {iter} + 1$
        \item Sample $x \sim \kappa(\cdot \mid C)$, where
        \[
            \kappa(x \mid C)
            :=
            \frac{\|x\|^2 + \|c_1\|^2}{\|\X\|_{\operatorname{F}}^2 + n\|c_1\|^2}
        \]
        \item Sample $R \sim \textsc{Uniform}(0,1)$
        \item Compute
        \[
            r \gets \frac{\cost(x, \Query(\L, x))}{2\rho^{-1}\bigl(\|x\|^2 + \|c_1\|^2\bigr)}
        \]
        \item \textbf{if} $R \le r$ \textbf{then}
        \begin{enumerate}
            \item $s \gets x$
            \item \textbf{break}
        \end{enumerate}
    \end{enumerate}
    \item \textbf{Return} $s$
\end{enumerate}
\end{algorithm}

\begin{algorithm}[t]
\caption{\textsc{QKMeans}$(\X, k, m)$}
\label{alg:qkmeans}
\begin{enumerate}
    \item \textsc{PreProcess}$(\X)$
    
    \item Initialize the data structure $\textsc{Init}(\L, \rho)$
    
    \item Sample $c_1 \sim \textsc{Uniform}(\X)$
    
    \item Set $C \gets \{c_1\}$, insert $c_1$ into $\L$, and set $C_1 \gets \{c_1\}$
    
    \item \textbf{for} $t = 2$ \textbf{to} $k$ \textbf{do}
    \begin{enumerate}
        \item Sample $c_t \gets \textsc{Sample}(\L, C, m)$
        \item Insert $c_t$ into $\L$
        \item Update $C_t \gets C_{t-1} \cup \{c_t\}$
    \end{enumerate}
    
    \item \textbf{Return} $C_k$
\end{enumerate}
\end{algorithm}

At a high level, our algorithm works by \emph{simulating} the $D^2$ sampling step of
$k$-means++ using rejection sampling. The idea is simple: instead of sampling directly
from the $D^2$ distribution—which is expensive to update after every new center—we
sample from a much easier distribution and correct the bias using a rejection step.
The efficiency of this correction is controlled by a single quantity: the max
R\'enyi divergence between the two distributions.

Suppose $\mu$ is a target distribution over a finite domain $\X$, and $\nu$ is a proposal
distribution from which sampling is easy. Rejection sampling allows us to draw exact
samples from $\mu$ using only samples from $\nu$, provided that $\mu$ is not too
``peaked'' relative to $\nu$.

\begin{definition} \label{def:max-renyi-div}
    Let $\mu$ and $\nu$ be probability distributions supported on $\X$.
    The max R\'enyi divergence of $\mu$ from $\nu$ is defined as
    \[
        D_\infty(\mu \| \nu)
        \defeq \sup_{x \in \X} \log \frac{\mu(x)}{\nu(x)}.
    \]
\end{definition}

Intuitively, $\exp(D_\infty(\mu\|\nu))$ measures the worst-case mismatch between $\mu$
and $\nu$, and directly determines how many proposal samples we need before one is
accepted.

\begin{algorithm}[t]
\caption{\textsc{RejectSample}$(\mu, \nu, M)$}
\label{alg:rejection}
\begin{enumerate}
    \item \textbf{for} $i = 1$ \textbf{to} $M \ln k$ \textbf{do}
    \begin{enumerate}
        \item Sample $x \sim \nu$
        \item Sample $R \sim \textsc{Uniform}(0,1)$
        \item \textbf{if} $R \le \frac{\mu(x)}{M \nu(x)}$ \textbf{then}
        \begin{enumerate}
            \item \textbf{return} $x$
        \end{enumerate}
    \end{enumerate}
    \item \textbf{Fail}
\end{enumerate}
\end{algorithm}

\begin{lemma}[Rejection sampling guarantee]
\label{lem:rejection}
Let $\mu$ and $\nu$ be distributions over a finite set $\X$ with
$D_\infty(\mu\|\nu) < \infty$, and set $M = \exp(D_\infty(\mu\|\nu))$.
Then:
\begin{enumerate}
    \item Conditioned on acceptance, Algorithm~\ref{alg:rejection} outputs an exact
    sample from $\mu$.
    \item Each iteration succeeds with probability at least $1/M$.
    \item Running the algorithm for $M \ln k$ iterations succeeds with probability
    at least $1 - 1/k$.
\end{enumerate}
\end{lemma}

\begin{proof}
In one iteration, the probability of accepting a point $x$ is
$\nu(x)\cdot \mu(x)/(M\nu(x)) = \mu(x)/M$.
Thus, conditioned on acceptance, the output distribution is exactly $\mu$.
The total acceptance probability is $1/M$, and repeating for $M \ln k$ iterations
fails with probability at most $(1-1/M)^{M\ln k} \le 1/k$.
\end{proof}

For performing $k$-means++ seeding, we are interested in sampling from the $D^2$ distribution given by $\pi(x|C) = \cost(x,C) / \cost(\X,C)$ for each $x \in \X$, which is our target. What remains to be done is to select a suitable proposal distribution, which we now define as follows : 

\begin{definition} \label{def:proposal}
    Let $c_1$ be the first center which is chosen uniformly at random from $\X$ by Algorithm~\ref{alg:qkmeans}. The distribution $\kappa(\cdot|C)$ is defined as follows : 
    $$\kappa(x|C) = \frac{\|x\|^2 + \|c_1\|^2}{\|\X\|_F^2 + n\|c_1\|^2}$$
\end{definition}

\begin{remark}
    The motivation behind selecting this particular distribution comes from a line of work in the theoretical computer science community called \emph{``dequantized quantum machine learning"} which often uses similar $\ell^2$ norm distributions over the data. More recently, this distribution was introduced in the context of clustering algorithms by \cite{shah2025quantum}. 
\end{remark}

Using the Cauchy-Schwarz inequality, we can provide an upper bound on the max-R\'enyi divergence as follows : 

\begin{lemma}
    Let $\pi(\cdot|C)$ be the target $D^2$ distribution and $\kappa(\cdot|C)$ be the proposal distribution defined in Definition~\ref{def:proposal}. Then the following holds : 
    $$\E[D_\infty\big(\pi(\cdot|C) \| \kappa(\cdot|C)\big)] \leq \log \left(4 \cdot \frac{\opt_1(\X)}{\opt_k(\X)}\right) $$
\end{lemma}

Notice that the expression inside the $\log(\cdot)$ is exactly the $\beta_k(\X)$ parameter which we introduced earlier. Under Assumption~\ref{ass} we can use Theorem~\ref{thm:scaling-laws} to obtain the following: 

\begin{lemma}
    Suppose the dataset $\X$ is generated according to Assumption~\ref{ass}, then the following holds : 
      $$\E[D_\infty\big(\pi(\cdot|C) \| \kappa(\cdot|C)\big)] \in O\left(\frac{\log k}{d}\right)$$
\end{lemma}

Thus we can perform $k$-means seeding by first sampling a uniformly random point as a centre, and subsequently calling Algorithm~\ref{alg:rejection} $k-1$ times with $\mu, \nu$ being the distributions $\pi(\cdot|C), \kappa(\cdot|C)$ respectively, where $C$ is the current set of centers. In order to employ the rejection sampling idea, the probability of accepting a sampled point $x\in\X$ needs to be proportional to $\cost(x,C)^2$. However, computing this quantity takes $\Theta(kD)$ time, which contributes $\Theta(k^2D)$ time on performing this sampling $k$ times. To mitigate this quadratic dependence on $k$, we apply another optimization technique; we use an approximate nearest neighbors (ANNS) data structure to approximate the distance between $x$ and the closest center. We use one such ANNS data structure developed by \citep{cohenaddad2020fast} which is based on locality sensitive hashing (details in Appendix~\ref{subsec:lsh}). In particular, for a parameter $\rho < 1$, we can query for a point in $C$ which is $1/\sqrt{\rho}$ times the minimum distance of $x$ with $C$ in time $O\left(D \log \eta (k \log \eta)^{\rho}\right)$. 

We take one more step. We are passing a parameter $m\in\mathbb{N}$ in our input that is used to give a threshold on the time complexity of our seeding algorithm. We use this parameter to control the number of iterations in Rejection Sampling; if we do not output a sample after $m\ln k$ iterations of sampling from $\kappa(\cdot|C)$, we terminate and instead output a sample from the uniform distribution on $\X$. With these add-ons, our idea for $D^2$ sampling via rejection sampling is given in Algorithm~\ref{proc:sample}. 

It is important to observe that both these add-ons to our original rejection sampling strategy of Algorithm~\ref{alg:rejection} distort the output distribution from $D^2$. As a consequence, we no longer have the $O(\log k)$ approximation guarantee readily available to us. We instead have an approximation guarantee which is dependent on $m$ and $\rho$ as a part of Theorem~\ref{thm:tradeoffs}. We give a detailed proof of the new approximation guarantee in Appendix~\ref{app:eps-delta-analysis}; to model the new distribution of cluster centers, we analyze an abstract version of $k$-means++ which we call $(\rho,\delta)$-$k$-means++, where instead of sampling from the standard $D^2$ distribution $\pi(\cdot|C)$, we instead have to sample from a perturbed distribution defined as follows . 

\begin{definition}[$(\rho,\delta)$ perturbed $D^2$ distribution]
    Let $\X = \{x_1,\dots,x_n\} \subset \R^d$ be a dataset and $S \subset X$ be a set of centers sampled from $X$. Suppose that the data structure $\L$ with parameter $\rho
    < 1$ is successful and each $c \in S$ is inserted into the data structure $\L$. Then for each $x \in \X$, we define the $(\rho,\delta)$ perturbed $D^2$ distribution as :
    \begin{equation*}
        \pi_{(\rho,\delta)}(x|S) = (1 - \delta)\cdot \pi^\L(x|S) + \delta \cdot \frac{1}{n}
    \end{equation*}
\end{definition}

Our techinical contribution is to analyze the seeding quality of our algorithm under this distribution. This requires a careful potential based analysis to control the effect of the parameters $(\rho,\delta)$ without introducing any large constant factors. We emphasize that this is a non trivial task since it is not known whether perturbing the $D^2$ distribution by a small amount translates to a similar guarantee in the solution cost or not \cite{grunau2023noisykmeanspp}. We leverage the specific structure of the perturbation to show the following. 

\begin{theorem}[Approximation guarantee for $(\rho,\delta)$-$k$-means++] \label{thm:approx guarantee-main}
Let $\X \subset \R^D$ be any dataset which is to be partitioned into $k$ clusters. Let $C$ be the set of centers returned by Algorithm~\ref{alg:rho-delta} for any $\delta \in (0,0.5)$ and $\rho < 1$. The expected cost $\E[\cost(\X,C)]$ equals:
\begin{align*}
O(\rho^{-2} \ln k)\opt_k(\X) + \delta  \cdot O(k + \rho^{-2} \ln k) \opt_1(\X)
\end{align*}
\end{theorem}

If we let $m=\infty$, we can analyze the expected time of our algorithm and the approximation guarantee only depends on $\rho$, which is given in Corollary~\ref{thm:rejection-sampling}. Finally, by plugging in the manifold hypothesis assumption, we obtain Corollary~\ref{thm:assumption}. 

\section{Empirical Study}
In this section, we empirically evaluate the extent to which the theoretical predictions developed in this paper are reflected in real-world data. The complete details are in Appendix~\ref{app:empirical}.

\textit{Compute environment. }All experiments were conducted on a dual-socket machine equipped with 2.20 GHz $\operatorname{Intel Xeon Gold 5220R}$ running $\operatorname{Linux}$ on $\operatorname{x86-64}$  architecture.

\textit{Datasets. } We include popular image datasets, both as raw pixels as well as image embeddings which are produced through joint vision-language models such as $\operatorname{CLIP-ViT-B/32}$ \cite{radford2021clip}; text datasets in the form of text embeddings generated by the $\operatorname{All-MiniLM-L6-v2}$ model \cite{wang2020minilm}; tabular datasets which record sensor data as well as data generated by scientific experiments. The details of the datasets used are presented in Table~\ref{tab:datasets}.

\begin{table}[t]
\centering
\small
\begin{tabular}{lcc}
\toprule
\textsc{Dataset} & $n$ & $d$ \\
\midrule
\textsc{MNIST}          & 60{,}000      & 784   \\
\textsc{FMNIST}         & 60{,}000      & 784   \\
\textsc{CIFAR10}        & 50{,}000      & 3{,}072 \\
\textsc{CIFAR100}       & 50{,}000      & 3{,}072 \\
\midrule
\textsc{MNIST (CLIP)}   & 60{,}000      & 512   \\
\textsc{FMNIST (CLIP)}  & 60{,}000      & 512   \\
\textsc{CIFAR10 (CLIP)} & 50{,}000      & 512   \\
\textsc{CIFAR100 (CLIP)}& 50{,}000      & 512   \\
\midrule
\textsc{Reddit}         & 100{,}000     & 384   \\
\textsc{StackExchange}  & 20{,}000      & 384   \\
\midrule
\textsc{SUSY}           & 5{,}000{,}000 & 18    \\
\textsc{HAR}            & 10{,}299      & 561   \\
\bottomrule
\end{tabular}
\caption{Datasets used in experiments with sample size $n$ and ambient dimension $d$.}
\label{tab:datasets}
\end{table}

\paragraph{Scaling Laws.} We test the hypothesis that the data-dependent parameters $\beta = \beta(\X,C)$ and $\eta = \eta(\X,C)$ grow polynomially with the number of clusters $k$. For each dataset, we compute $k$-means clusterings for
$k \in \{5,10,50,100,250,500,750,1000\}$ using $10$ independent runs of standard $k$-means++ followed by Lloyd's iterations. We fit linear regressions to $\log \beta$ and $\log \eta$ versus $\log k$ and report slopes, $R^2$, and $95\%$ confidence intervals. Figure~\ref{fig:scaling-side-by-side} shows that across datasets and modalities, both quantities exhibit reasonable linear behavior in log--log coordinates. It is important to note that although we are not able to provide bounds on $\eta$ in term of $k$, the experiments show that $\log \eta$ scales atmost linearly with $\log k$.The estimated quantization exponent $\widehat{\eps}$ is the slope of the best fit line for the $\log \beta$ vs $\log k$ plot.  

\begin{figure}
    \centering
    \begin{subfigure}{0.48\linewidth}
        \centering
        \includegraphics[width=\linewidth]{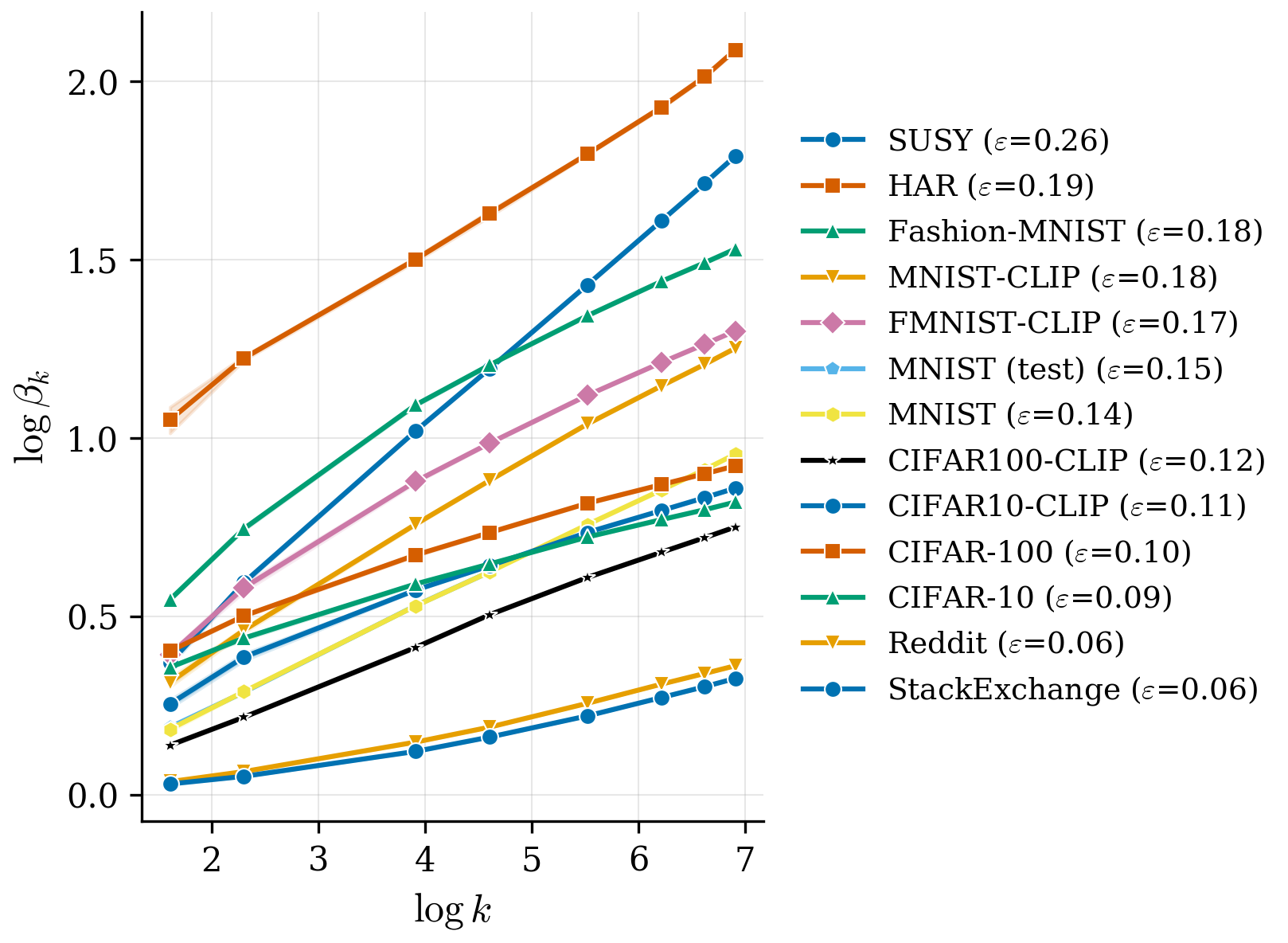}
        \caption{$\beta$ scaling}
        \label{fig:beta-scaling}
    \end{subfigure}
    \hfill
    \begin{subfigure}{0.48\linewidth}
        \centering
        \includegraphics[width=\linewidth]{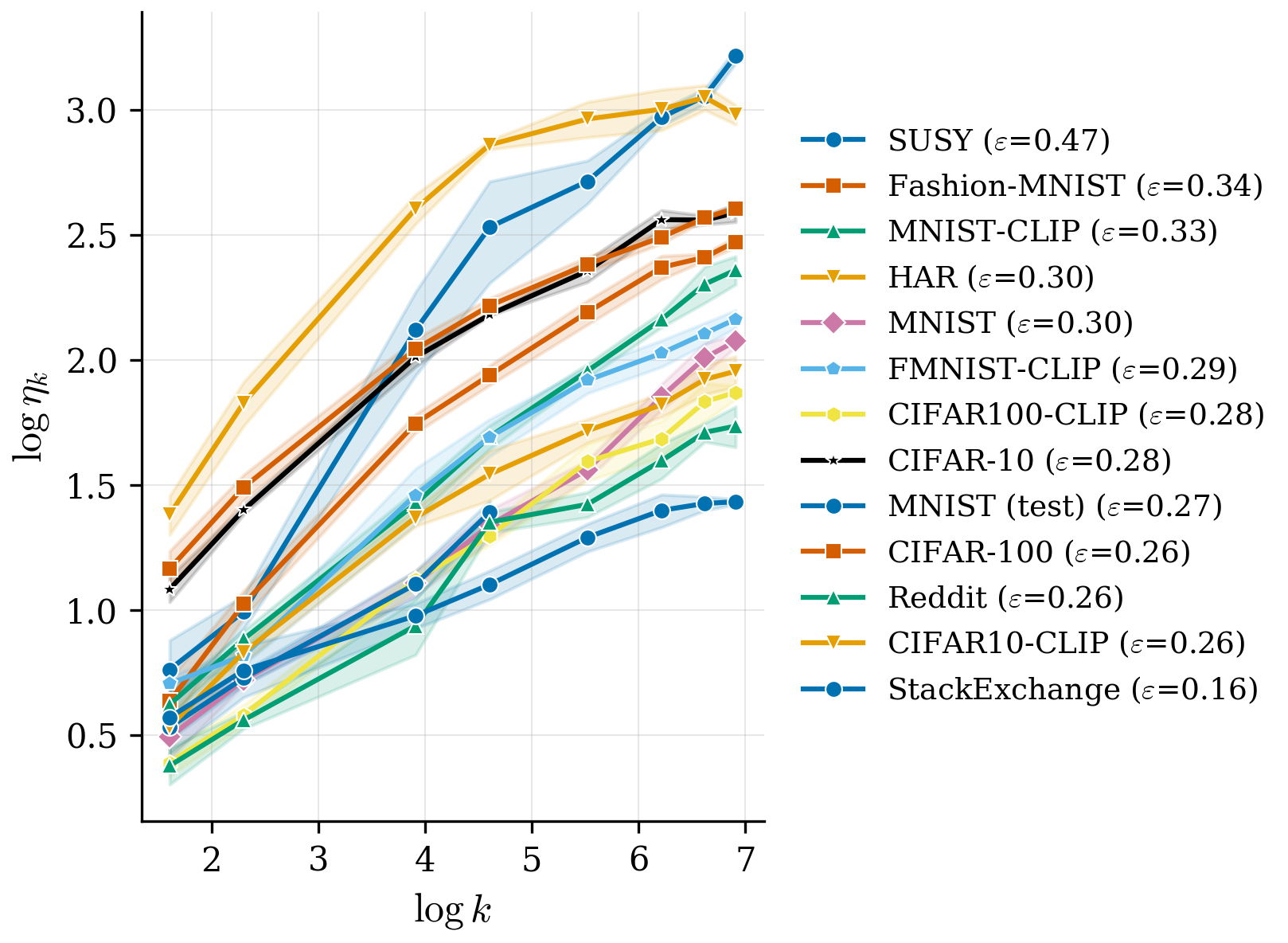}
        \caption{$\eta$ scaling}
        \label{fig:eta-scaling-sub}
    \end{subfigure}
    \caption{Scaling behavior across datasets}
    \label{fig:scaling-side-by-side}
\end{figure}

\paragraph{Intrinsic dimension.} We test the hypothesis that the empirically estimated quantization exponent $\widehat{\varepsilon}$ is inversely proportional to the intrinsic dimension. Independent intrinsic dimension estimates $\widehat{d}_{\operatorname{MLE}}$ are obtained by using the popular maximum likelihood method (MLE) \cite{levina2005maximum}. For a dataset $\X = \{x_1,\dots,x_n\}$, let $T_j(x_i)$ denote the distance of the $j$th nearest neighbor of $x_i$ from $x_i$. For a hyperparameter $k$, the ID estimate is calculated as : 
$$\widehat{d}_{\operatorname{MLE}}(\X,k) := \frac{1}{n} \sum_{i=1}^n \left(\frac{1}{k-1} \sum_{j=1}^{k-1} \log \frac{T_k(x_i)}{T_j(x_i)}\right)^{-1}$$
For each dataset we take $10$ independent samples each  of size $n' = 10,000$ to estimate the ID for a particular value of $k$. This is repeated for $k \in \{5,10,20,50,100\}$ and the average is reported. We plot the estimates $\widehat{\eps}$ of the quantization exponent against $\widehat{d}_{\operatorname{MLE}}$ in Figure~\ref{fig:id-min}.\\\\

\begin{figure}
    \centering
    \begin{subfigure}{0.48\linewidth}
        \centering
        \includegraphics[width=\linewidth]{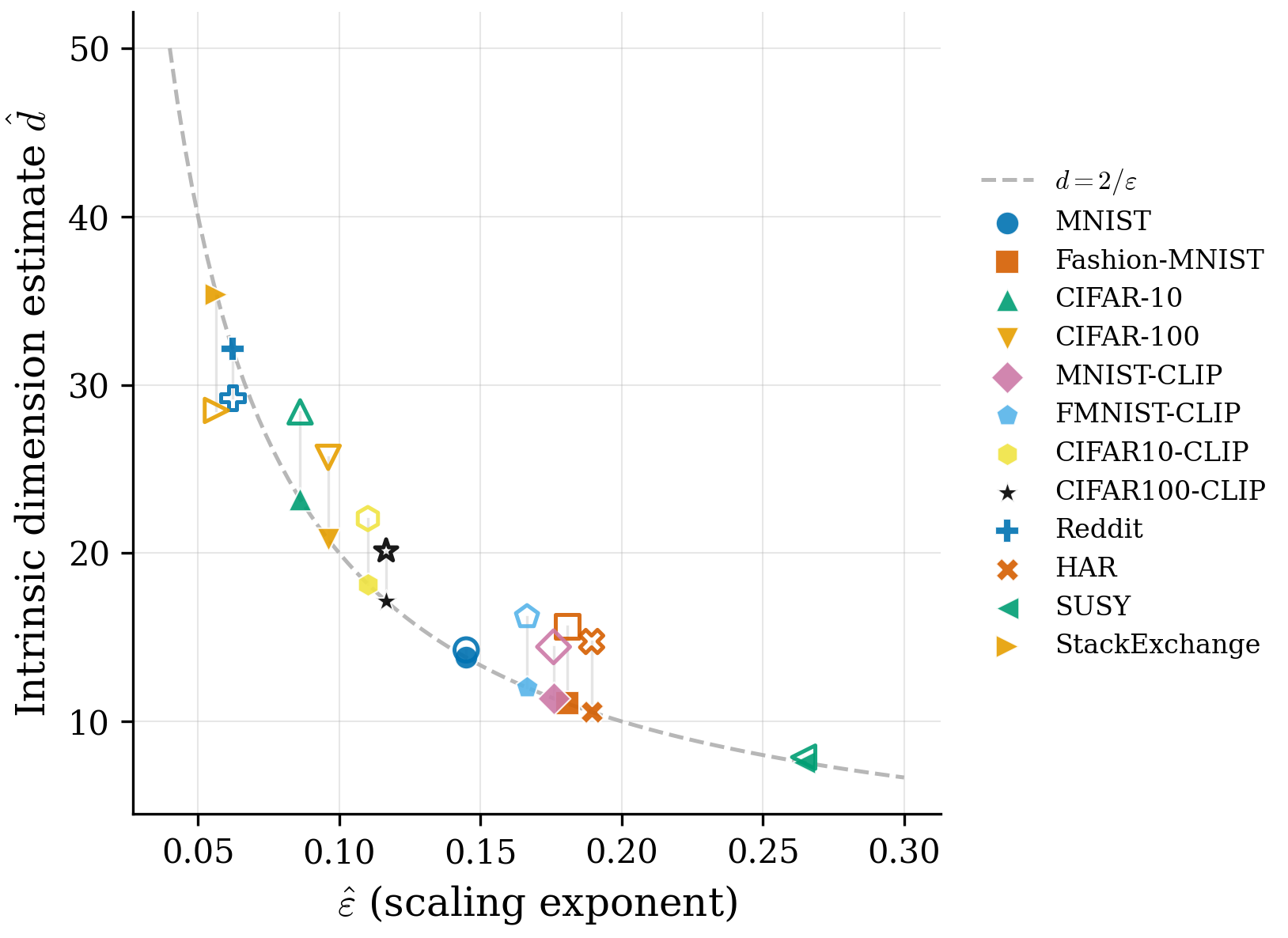}
        \caption{$\eps$ vs $\hat{d}_{\operatorname{MLE}}$}
        \label{fig:eps-deps}
    \end{subfigure}
    \hfill
    \begin{subfigure}{0.48\linewidth}
        \centering
        \includegraphics[width=\linewidth]{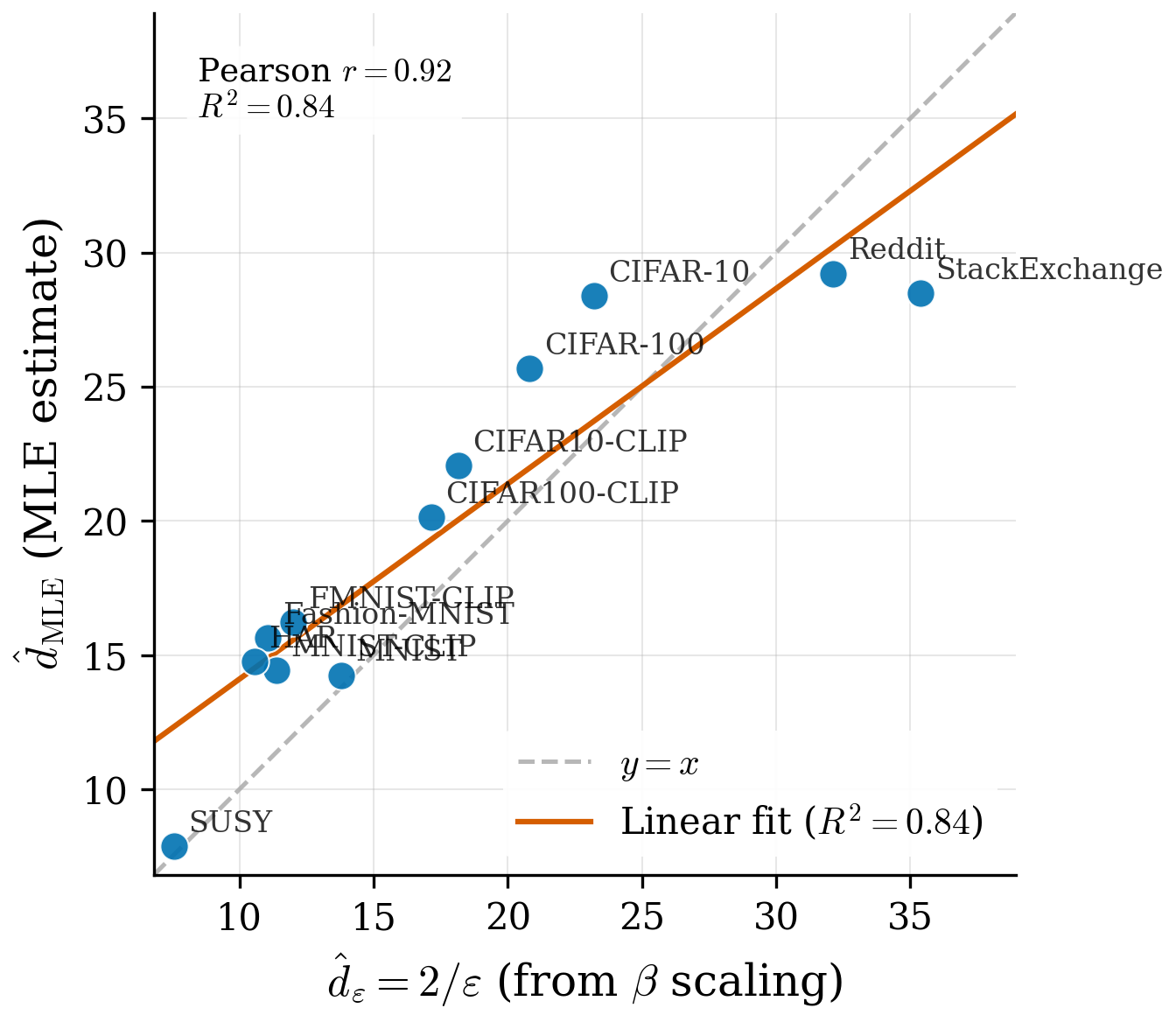}
        \caption{$\hat{d}_\eps$ vs $\hat{d}_{\operatorname{MLE}}$}
        \label{fig:eta-scaling}
    \end{subfigure}
    \caption{Dependence of quantization exponent on ID}
    \label{fig:id-min}
\end{figure}

\vspace{-20pt}

\paragraph{Seeding performance.} We compare our seeding algorithm $\qkmeans$ with $\afkmc$ \cite{bachem2016good}, $\rejsample$ \cite{cohenaddad2020fast}, $\pronecoreset$ \cite{charikar2023simple} and $\fastkmeans$ \cite{draganov2024settle}. All the algorithms were implemented in C++ without using any dimensionality reduction to ensure a fair comparison. Across a wide range of datasets and data modalities, the experiments show that $\qkmeans$ consistently provides the fastest seeding procedure in practically relevant regimes, particularly as the number of clusters $k$ and the ambient dimension $D$ increase. $\rejsample$ runtimes are almost independent of $k$, but it along with $\fastkmeans$ suffers a significant overhead due to constructing the multi-tree embeddings. $\afkmc$
is competitive only in low-dimensional tabular settings, but loses its advantage on image, text, and embedding data. From the detailed benchmarks in Appendix~\ref{app:empirical}, $\qkmeans$ achieves consistent speedups of $5–10\times$ over $\afkmc$ and $\pronecoreset$ and up to two orders of magnitude over $\rejsample$, while maintaining comparable seeding cost. Additional experiments on the effect of noisy manifolds and studying the rejection rate vs $m$ are presented in Appendix~\ref{app:empirical}.

\begin{figure}[h]
    \centering
    \includegraphics[width=\linewidth]{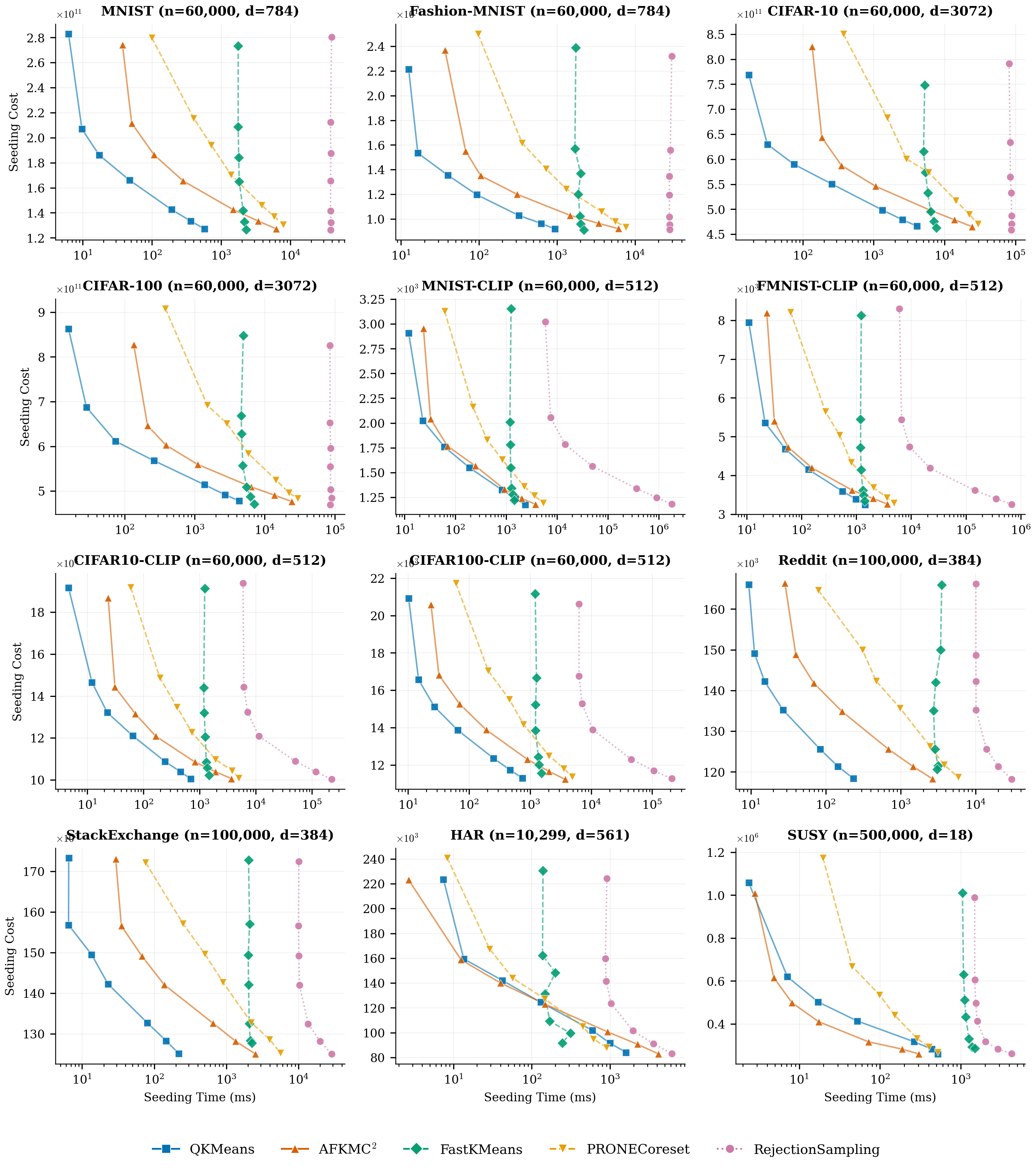}
    \caption{Seeding quality - runtime tradeoff plots}
    \label{fig:placeholder}
\end{figure}
\vspace{-20pt}
\section{Conclusion, Limitations, Future Directions}

In this work, we modelled typical $k$-means clustering instances from a manifold hypothesis perspective. Using techniques from optimum quantization theory, we identified key geometric properties of the dataset that follow predictable scaling laws. We also showed how to exploit these regularities to design a fast rejection sampling-based seeding algorithm with new tradeoffs. Our theoretical results were backed by extensive empirical evaluations. We feel that due to its broad impact  on the machine learning community, studying algorithms under the manifold hypothesis is an appealing and exciting direction for future research.

\newpage
\bibliography{example_paper}
\bibliographystyle{icml2026}

\newpage
\appendix
\onecolumn

\section*{Appendix}

The appendices expand upon the content of the main paper. In Appendix~\ref{app:scaling-laws}, we provide a theoretical proof of the scaling laws presented in Theorem~\ref{thm:scaling-laws}. In Appendix~\ref{app:theoretical-analysis} we provide the technical proofs for our seeding algorithm along with the analysis of its runtime and approximation guarantees. Finally, in Appendix~\ref{app:empirical} we provide the details of aour empirical study along with additional experiments and data. 

\section{Theoretical Analysis of Scaling Laws} \label{app:scaling-laws}

In this section we provide theoretical proofs for the scaling laws presented in  Theorem~\ref{thm:scaling-laws}. We mention that the techniques used here are well developed in optimum quantization theory \cite{zador1982asymptotic,gray1998quantization,graf2000foundations,gruber2004optimum} and we provide an elementary proof for the sake of being self-contained. Since we work under the manifold assumption, we require definitions and techniques from differential geometry for our proofs. We refer the reader to the book by \cite{Lee2019} for obtaining familiarity with such notions. Recall that we are given a dataset $\X = \{x_1,\dots,x_n\}$ where the $x_i$'s are sampled i.i.d. from a probability density function $f : \M \to \R_{+}$ supported on a compact, connected and smooth subset of a manifold $\M \subset \R^D$. The proof can be broken down into the following two steps : 
\begin{itemize}
    \item \emph{Step 1} Showing that the optimum distortion $\Delta_k(f) = \inf_{|C| = k} \int_{\M}\min_{c \in C}\|x-c\|^2 f(x)\dvol_\M(x)$ follows the scaling law $\Delta_k(f) \in \Theta_f(k^{-2/d})$. 
    \item \emph{Step 2} Showing how to translate this in the finite sample case by studying how $\frac{1}{n}\opt_k(\X) \to \Delta_k(f)$ as $n \to \infty$. 
\end{itemize}

\paragraph{Preliminaries.}
A function $\phi : A \to B$ between two metric spaces $A$ and $B$ is called a
\emph{homeomorphism} if $\phi$ is bijective and both $\phi$ and $\phi^{-1}$ are
continuous. A set $\M \subset \R^D$ is called a smooth $d$-dimensional manifold embedded in
$\R^D$ if for each point $p \in \M$ there exists an open neighborhood
$W \subset \R^D$ of $p$ and a smooth diffeomorphism
$\phi : W \to \R^D$ such that $\phi(\M \cap W) = \R^d \times \{0\}^{D-d}.$ Equivalently, $\M$ admits a collection of local coordinate maps as follows.
A \emph{chart} on $\M$ is a pair $(U,\varphi)$, where $U \subset \M$ is open and
$\varphi : U \to V \subset \R^d$ is a homeomorphism onto an open subset of
$\R^d$. A collection of charts
$
\mathcal{A} = \{(U_\alpha,\varphi_\alpha)\}_{\alpha \in I}
$
is called an \emph{atlas} if $\M = \bigcup_{\alpha \in I} U_\alpha$.
The atlas is said to be \emph{smooth} if for every pair of overlapping charts
$(U_\alpha,\varphi_\alpha)$ and $(U_\beta,\varphi_\beta)$, the transition maps
$
\varphi_\beta \circ \varphi_\alpha^{-1} :
\varphi_\alpha(U_\alpha \cap U_\beta) \to
\varphi_\beta(U_\alpha \cap U_\beta)
$
are smooth diffeomorphisms. Throughout, we assume that $\M$ is equipped with the smooth atlas induced by its
embedding in $\R^D$. This atlas endows $\M$ with the structure of a smooth
manifold and allows us to perform local analysis by identifying sufficiently
small neighborhoods of $\M$ with open subsets of $\R^d$. 

\begin{lemma}[Zador bounds for manifolds]\label{lem:zador-manifold}
Let $\M \subset \R^D$ be a compact, $d$-dimensional $C^2$ submanifold and let
$J \subset \M$ be compact, measurable, and satisfy $\vol_{\M}(J)>0$. Let $f : J \to R_+$ be a probability density function on $\M$ such that $f_{min} < f(x) < f_{max}$ and $\int_\M f(x)\dvol_\M(x) = 1$. Define
\[
\Delta_k(f)
:=
\inf_{\substack{S \subset \R^D\\ |S| = k}}
\int_J \min_{p \in S} \|x-p\|^2 \, f(x) \dvol_{\M}(x).
\]
Then there exist constants $0<c<C<\infty$ depending only on $\M$,$J$ and $f$ such that
\[
c\, k^{-2/d} \le \Delta_k(f) \le C\, k^{-2/d}
\quad \text{for all } k \ge 1.
\]
\end{lemma}

\begin{proof}

\textit{Upper bound.}
Since $J$ is compact and $\M$ is a smooth embedded submanifold, there exists a
finite atlas $\{(U_\ell,\varphi_\ell)\}_{\ell=1}^m$ such that
\[
J = \bigcup_{\ell=1}^m J_\ell,
\qquad J_\ell \subset U_\ell,
\]
where each $\varphi_\ell : U_\ell \to V_\ell \subset \R^d$ is a diffeomorphism.
By compactness and smoothness, shrinking charts if necessary, there exists
a constant $C_0 \ge 1$ such that for all $\ell$ and all $x,y \in U_\ell$,
\begin{equation}\label{eq:metric-compare}
C_0^{-1}\|\varphi_\ell(x)-\varphi_\ell(y)\|
\le
\|x-y\|
\le
C_0\|\varphi_\ell(x)-\varphi_\ell(y)\|,
\end{equation}
and the induced volume form satisfies
\begin{equation}\label{eq:volume-compare}
C_0^{-1}\,du
\le
\dvol_{\M}(x)
\le
C_0\,du
\quad \text{under } u=\varphi_\ell(x).
\end{equation}

Fix $k \ge m$ and set $k_\ell := \lfloor k/m \rfloor$. The set
$\varphi_\ell(J_\ell) \subset \R^d$ is compact, hence contained in a cube of side
length $L_\ell < \infty$. Partition this cube into
$\lceil k_\ell^{1/d} \rceil^d$ congruent subcubes, each of side length
at most $2L_\ell k_\ell^{-1/d}$. Selecting one point from each subcube intersecting
$\varphi_\ell(J_\ell)$ yields a set
$\{u_{\ell,i}\}_{i=1}^{k_\ell} \subset \varphi_\ell(J_\ell)$ such that
\[
\varphi_\ell(J_\ell)
\subset
\bigcup_{i=1}^{k_\ell}
B_{\R^d}(u_{\ell,i}, C k^{-1/d}),
\]
for a constant $C$ depending only on $J$ and $d$.

Define $p_{\ell,i} := \varphi_\ell^{-1}(u_{\ell,i}) \in J_\ell$ and let
\[
S := \bigcup_{\ell=1}^m \{p_{\ell,1},\dots,p_{\ell,k_\ell}\},
\]
so that $|S| \le k$. By \eqref{eq:metric-compare}, for every $x \in J$,
\[
\min_{p \in S} \|x-p\|
\le
C' k^{-1/d},
\]
and hence
\[
\min_{p \in S} \|x-p\|^2
\le
C'' k^{-2/d}.
\]
Integrating over $J$ and using the finiteness of $\vol_{\M}(J)$ gives
\[
\int_J \min_{p \in S} \|x-p\|^2 \, f(x) \dvol_{\M}(x)
\le
C k^{-2/d},
\]
which implies $\Delta_k(f)\le C k^{-2/d}$.

\textit{Lower bound.}
Since $\vol_{\M}(J)>0$ and $\M$ is smooth, there exists a point $x_0 \in J$ and
constants $\rho>0$ and $c_0>0$ such that the geodesic ball
$B_{\M}(x_0,r) \subset J$ for all $0<r\le\rho$, and
\begin{equation}\label{eq:volume-growth}
\vol_{\M}(B_{\M}(x_0,r)) \ge c_0 r^d.
\end{equation}

Let $S \subset \R^D$ with $|S|=k$ be arbitrary and let $r>0$. For each $p \in S$,
the intersection $B(p,r) \cap \M$ is contained in a geodesic ball of radius
$C r$, hence by \eqref{eq:volume-growth},
\[
\vol_{\M}(B(p,r) \cap \M) \le C r^d.
\]
Therefore,
\[
\vol_{\M}\!\left(\bigcup_{p \in S} B(p,r) \cap \M\right)
\le
C k r^d.
\]
Choosing $r = c k^{-1/d}$ with $c>0$ sufficiently small ensures that
\[
\vol_{\M}\!\left(\bigcup_{p \in S} B(p,r) \cap \M\right)
<
\vol_{\M}(B_{\M}(x_0,\rho)).
\]
Hence there exists a measurable subset
\[
A \subset B_{\M}(x_0,\rho)
\]
of positive volume such that $\|x-p\| \ge r$ for all $x \in A$ and all $p \in S$.
For such $x$,
\[
\min_{p \in S} \|x-p\|^2 \ge r^2 = c^2 k^{-2/d}.
\]
Integrating over $A$ yields
\[
\int_J \min_{p \in S} \|x-p\|^2 \, f(x)\dvol_{\M}(x)
\ge
c' k^{-2/d},
\]
where $c'>0$ depends only on $\M$ and $J$. Since $S$ was arbitrary, this implies
$\Delta_k(f)\ge c k^{-2/d}$.

Combining the upper and lower bounds completes the proof.
\end{proof}

Next, we show finite sample bounds on the $k$-means cost to quantify how $\opt_k(\X) \to \Delta_k(f)$ as the number of samples $n \to \infty$. By re-scaling, we can assume without loss of generality that $\|x\| \leq 1$ for each $x \in J$. For this, we would require two standard results : 

\begin{lemma}[Hoeffding's Inequality]\label{fact:hoeffding} Let $X_1,\dots,X_m$ be independent random variables such that $a_j \leq X_j \leq b_j$ almost surely for each $j \in [m]$. Let $S = \sum_{j=1}^mX_j$ then : 
$$ \Pr[|S_m - \E[S_m]| \geq t] \leq 2 \exp\left(- \frac{2t^2}{\sum_{j=1}^m (b_j-a_j)^2}\right)$$
    
\end{lemma}  

\begin{lemma}[Covering Lemma]\label{fact:covering}
    Let $A \subset \R^m$ be a bounded set. Then for any $\delta > 0$, there exists a finite set $\mathcal{N}_\delta(A) \subset A$, called a $\delta$-net for $A$, such that: 
    \begin{enumerate}
        \item For every $x \in A$, there exists $x' \in \mathcal{N}_{\delta}(A)$ with $\|x-x'\| \leq \delta$. 
        \item $|\mathcal{N}_\delta(A)| \leq \left(\frac{\diam(A)}{\delta}+1\right)^{m}$. 
    \end{enumerate}
\end{lemma}

Using these, we can show the following : 

\begin{lemma}[Finite sample bound on $k$-means cost] \label{lem:finite-sample} For any $\delta \in (0,1)$, the following holds with probability atleast $1-\delta$ : 

$$\left|\frac{1}{n}\opt_k(\X) - \Delta_k(f)\right| \in O_f\left( \sqrt{\frac{kD \log n + \log (1/\delta)}{n}}\right)$$
\end{lemma}

\begin{proof}
Let us define the set of possible centers $\C_k := \{(c_1,\dots,c_k) \subset \R^D : \|c_j\| \leq 1, 1 \leq j \leq k\}$. For a fixed set of centers $C = (c_1,\dots,c_k) \in \C_k$, let us define the following quantities:

\begin{align*}
     \ell(x|C) &:= \min_{1\leq j\leq k}\|x-c_j\|^2 \\ 
     L_n(C) &:= \frac{1}{n} \sum_{i \leq i \leq n} \ell(x_i|C) \\ 
     L(C) &:= \E_{x \sim f}\ell(x|C)
\end{align*}

It can be easily seen that the function $\ell(x|C)$ is bounded since: 

$$\ell(x|C) = \min_{1 \leq j \leq k}\|x-c_j\|^2 \leq \|x-c_1\|^2 \leq 2(\|x\|^2 + \|c_1\|^2) \leq 4$$

Notice that by definition, $L_n(C)$ is the average of $n$ i.i.d. random variables $\ell(x_i|C)$ each of which is bounded in $[0,4]$. Hence, applying Lemma~\ref{fact:hoeffding} implies that for any $t > 0$,
$$\Pr\bigl(|L_n(C) - L(C)| \ge t\bigr) \le 2 \exp(-nt^2/8)$$

We next control the dependence of the loss on the centers. Let
$C = (c_1,\dots,c_k)$ and $C' = (c_1',\dots,c_k')$ be two elements of $\C_k$ such
that
\[
\max_{j \in [k]} \|c_j - c_j'\| \le \delta.
\]

For any $x \in \R^D$, we have
\begin{align*}
\bigl|\ell(x \mid C) - \ell(x \mid C')\bigr|
&= \left| \min_{1 \le j \le k} \|x-c_j\|^2
      - \min_{1 \le j \le k} \|x-c_j'\|^2 \right| \\
&\le \max_{1 \le j \le k}
    \left| \|x-c_j\|^2 - \|x-c_j'\|^2 \right|.
\end{align*}
For a fixed $j$, we may expand
\begin{align*}
\left| \|x-c_j\|^2 - \|x-c_j'\|^2 \right|
&= \bigl| \langle c_j' - c_j,\, 2x - c_j - c_j' \rangle \bigr| \\
&\le \|c_j - c_j'\| \cdot \|2x - c_j - c_j'\|.
\end{align*}
Since $\|x\| \le 1$ and $\|c_j\|,\|c_j'\| \le 1$, it follows that
\[
\|2x - c_j - c_j'\|
\le 2\|x\| + \|c_j\| + \|c_j'\|
\le 4.
\]
Consequently,
\[
\bigl|\ell(x \mid C) - \ell(x \mid C')\bigr|
\le 4 \max_{j \in [k]} \|c_j - c_j'\|
\le 4\delta.
\]

Averaging over the samples and taking expectation yields
\[
|L_n(C) - L_n(C')| \le 4\delta,
\qquad
|L(C) - L(C')| \le 4\delta.
\]

\medskip

Let $\varepsilon > 0$ be a parameter to be chosen later. By
Lemma~\ref{fact:covering}, there exists an $\varepsilon$-net
$\mathcal N_\varepsilon \subset \C_k$ such that
\[
|\mathcal N_\varepsilon|
\le \left(\frac{2}{\varepsilon}+1\right)^{kD}
\le \left(\frac{3}{\varepsilon}\right)^{kD}.
\]
For any $C \in \C_k$, let $\Pi(C) \in \mathcal N_\varepsilon$ denote a nearest
net point. Then
\[
|L_n(C) - L(C)|
\le |L_n(\Pi(C)) - L(\Pi(C))| + 8\varepsilon.
\]

Applying Hoeffding's inequality and a union bound over
$\mathcal N_\varepsilon$, we obtain for any $t > 0$,
\[
\Pr\!\left(
\sup_{C' \in \mathcal N_\varepsilon}
|L_n(C') - L(C')| \ge t
\right)
\le
2 |\mathcal N_\varepsilon| \exp(-nt^2/8).
\]
Thus, with probability at least $1-\delta$,
\[
\sup_{C' \in \mathcal N_\varepsilon}
|L_n(C') - L(C')|
\lesssim
\sqrt{\frac{kD \log(1/\varepsilon) + \log(1/\delta)}{n}}.
\]

Combining the above bounds, we conclude that with probability at least
$1-\delta$,
\[
\sup_{C \in \C_k} |L_n(C) - L(C)|
\lesssim
\sqrt{\frac{kD \log(1/\varepsilon) + \log(1/\delta)}{n}}
+ \varepsilon.
\]
Choosing $\varepsilon = 1/n$ gives $\log(1/\varepsilon) = \log n$, and hence
\[
\sup_{C \in \C_k} |L_n(C) - L(C)|
\in
O\!\left(
\sqrt{\frac{kD \log n + \log(1/\delta)}{n}}
\right).
\]

Finally, observing that
\[
\frac{1}{n}\opt_k(\X)
= \inf_{C \in \C_k} L_n(C),
\qquad
\Delta_k(f) = \inf_{C \in \C_k} L(C),
\]
we obtain
\[
\left|
\frac{1}{n}\opt_k(\X) - \Delta_k(f)
\right|
\le
\sup_{C \in \C_k} |L_n(C) - L(C)|,
\]
which completes the proof.
\end{proof}

We are now ready to provide a proof for Theorem~\ref{thm:scaling-laws}. 

\begin{theorem}[Scaling laws]\label{thm:scaling-laws-app}
Suppose $\mathcal X = \{x_1,\dots,x_n\} \subset \R^D$ consists of i.i.d.\ samples
drawn from a distribution $f$ supported on a compact, smooth $d$-dimensional
manifold $\M \subset \R^D$, with density bounded above and below on its support.
Let $\varepsilon := 2/d$. Then, with probability at least
$1- \frac{1}{\operatorname{poly}(n)}$, the following hold:
\[
\beta_k(\mathcal X)
=
\left(1 + O_f\!\left(\sqrt{\frac{kD \log n}{n}}\right)\right) k^{\varepsilon},
\qquad
\eta(\mathcal X) \in O_f(1)\, n^{3\varepsilon/2}.
\]
\end{theorem}

\begin{proof}
We prove the two statements separately.

\emph{Analysis of $\beta_k(\mathcal X)$.}
By definition,
\[
\beta_k(\mathcal X)
=
\frac{\opt_1(\mathcal X)}{\opt_k(\mathcal X)},
\qquad
\opt_k(\mathcal X)
=
\min_{C:|C|=k} \sum_{i=1}^n \min_{c \in C} \|x_i-c\|^2.
\]
Define the normalized empirical distortion
\[
\frac{1}{n}\opt_k(\mathcal X).
\]
By Lemma~\ref{lem:finite-sample}, for any fixed $k$, with probability at least
$1-\frac{1}{\operatorname{poly}(n)}$,
\[
\left|
\frac{1}{n}\opt_k(\mathcal X) - \Delta_k(f)
\right|
\le
C_f \sqrt{\frac{kD \log n}{n}}.
\]
In particular, the same bound holds for $k=1$ with the factor $k$ removed.

On the other hand, by Zador’s theorem on manifolds
(Lemma~\ref{lem:zador-manifold} and its matching lower bound),
\[
\Delta_1(f) = \Theta_f(1),
\qquad
\Delta_k(f) = \Theta_f(k^{-2/d}).
\]
Combining these estimates yields, with high probability,
\[
\frac{1}{n}\opt_1(\mathcal X)
=
\Theta_f(1) + O_f\!\left(\sqrt{\frac{D\log n}{n}}\right),
\]
and
\[
\frac{1}{n}\opt_k(\mathcal X)
=
\Theta_f(k^{-2/d})
\left(
1 + O_f\!\left(\sqrt{\frac{k D\log n}{n}}\right)
\right),
\]
where the error term is of lower order provided
$n \gg k^{1+2/d}\log n$.

Taking the ratio, we obtain
\[
\beta_k(\mathcal X)
=
\frac{\opt_1(\mathcal X)}{\opt_k(\mathcal X)}
=
\frac{\frac{1}{n}\opt_1(\mathcal X)}{\frac{1}{n}\opt_k(\mathcal X)}
=
\left(1 + O_f\!\left(\sqrt{\frac{k D\log n}{n}}\right)\right) k^{2/d},
\]
which proves the first claim.

\emph{Analysis of $\eta(\mathcal X)$.}
By definition,
\[
\eta(\mathcal X)
=
\frac{\max_{i \neq j} \|x_i-x_j\|}{\min_{i \neq j} \|x_i-x_j\|}.
\]
Since $f$ is supported on a compact set, we have
\[
\max_{i \neq j} \|x_i-x_j\| \le \diam(\operatorname{supp} f) = O_f(1)
\quad \text{almost surely}.
\]
It remains to lower bound the minimum interpoint distance.

By the volume growth property of a smooth $d$-dimensional manifold and the lower
bound on the density of $f$, there exists a constant $c_f>0$ such that for all
sufficiently small $r$ and all $x \in \M$,
\[
\Pr(\|X-x\| \le r) \le c_f r^d.
\]
Applying a union bound over all unordered pairs $(i,j)$ with $i \neq j$, we obtain
\[
\Pr\!\left( \min_{i \neq j} \|x_i-x_j\| \le r \right)
\le
n^2 c_f r^d.
\]
Choosing $r = n^{-3/d}$ yields
\[
\Pr\!\left( \min_{i \neq j} \|x_i-x_j\| \le r \right)
\le
c_f n^2 \cdot \frac{1}{n^3}
=
\frac{1}{\operatorname{poly}(n)},
\]
for all sufficiently large $n$. Therefore, with probability at least
$1-\frac{1}{\operatorname{poly}(n)}$,
\[
\min_{i \neq j} \|x_i-x_j\| \ge c_f' n^{-3/d}
\]
for some constant $c_f'>0$. Combining this with the diameter bound yields
\[
\eta(\mathcal X)
=
\frac{O_f(1)}{\Omega_f(n^{-3/d})}
=
O_f\!\left(n^{3/d}\right)
=
O_f\!\left(n^{3\varepsilon/2}\right).
\]
\end{proof}

\newpage 
\section{Notation}

\begin{table}[ht] 
\centering
\begin{tabular}{|c|p{3cm}|p{12cm}|}
\hline
\textbf{\#} & \textbf{Symbol} & \textbf{Meaning} \\
\hline
\rownumber & $n$ & number of points in the dataset \\
\rownumber & $D$ & ambient dimension of the dataset \\ 
\rownumber & $k$ & number of clusters into which the dataset is to be partitioned \\ 
\rownumber & $\X = \{x_1,\dots,x_n\}$ & dataset in consideration \\ 
\rownumber &$\mu = \mu(\X)$ & mean of the dataset $\X$ \\
\rownumber & $\|\X\|_\F$ & frobenius norm of the dataset defined as $\|\X\|_\F^2 = \sum_{i=1}^n \|x_i\|^2$ \\ 
\rownumber & $C = \{c_1,\dots,c_k\}$ & a set of $k$ cluster centers \\ 
\rownumber & $\cost(\X,C)$ & $k$-means cost of dataset $\X$ with respect to $C$ \\ 
\rownumber & $\pi(\cdot|C)$& $D^2$ distribution with respect to centers $C$ given by  $\pi(x|C) = \frac{\cost(x,C)}{\cost(\X,C)}$ \\
\rownumber &$\kappa(\cdot|C)$ &  proposal distribution with respect to center $c_1$ given by $\kappa(x|C) = \frac{\|x\|^2+\|c_1\|^2}{\|\X\|_F^2+n \|c_1\|^2}$\\ 
\rownumber & $\opt_k(\X)$ & optimal $k$-means cost for dataset $\X$ \\ 
\rownumber & $\{s_1,\dots,s_k\}$ & an optimal set of cluster centers \\ 
\rownumber & $\X = \bigcup_{j = 1}^k \X_j$ & optimal clustering partition \\
\rownumber & $\L = \L_\rho(C)$& LSH data structure initialized by parameter $\rho$ into which the centers $C$ have been insert\\ 
\rownumber & $\Insert(c,\L)$ & operation to insert a center $c$ in $\L$\\
\rownumber & $\Query(x,\L)$& operation to query an approximate nearest neighbor of $x$ from $\L$\\ 
\rownumber & $\cost^\L(\X,C)$ &  defined as $\sum_{x \in \X}\cost(x,\Query(x,\L))$\\
\rownumber & $\pi^\L(\cdot|C)$&  $\pi^\L(x|C) = \frac{\cost^\L(x,C)}{\cost^\L(\X,C)}$\\ 
\rownumber & $t$ & during the $t$th iteration $t$ centers have been chosen\\ 
\rownumber &$C_t = \{c_1,\dots,c_t\}$ & centers chosen uptill $t$ iterations\\
\rownumber &$H_t,U_t,W_t$ & number of covered clusters, uncovered clusters and wasted iterations uptill iteration $t$\\ 
\rownumber & $\X = \H_t \cup \U_t$& partitioning of $\X$ into covered and uncovered clusters\\
\rownumber & $\cost_t(\mathcal{P})$& defined as $\cost(\mathcal{P},C_t)$\\ 
\rownumber & $d$ & dimension of the underlying manifold $\M$\\
\rownumber & $\eps$ & quantization exponent defined as $\eps := 2/d$\\
\rownumber & $\rho$ & LSH parameter\\
\rownumber & $\gamma$ & overall exponent defined as $\gamma := \eps + \rho$\\
\rownumber & $f$ & probability density function supported on the  underlying manifold $\M$\\

\hline
\end{tabular}
\caption{Index of notation used throughout the appendix}
\end{table}

\newpage 
\section{Theoretical Analysis of $\qkmeans$}\label{app:theoretical-analysis}

The goal of this section is to describe our proof techniques for the theoretical analysis of the $\qkmeans$ seeding algorithm, which we reproduce here for convenience. 
\begin{algorithm}[H]
\caption{\textsc{QKMeans}$(\X, k, m)$}
\label{alg:qkmeans-single}
\begin{enumerate}
    \item \textbf{Preprocessing:}
    \begin{enumerate}
        \item (Optional) For $\eps_{\operatorname{JL}} \in (0,1/4)$, apply the JL-like dimensionality reduction from
        \cite{makarychev2019jl} to $\X$, reducing to dimension
        \[
            D \in O\!\left(\frac{\log(k/\eps_{\operatorname{JL}})}{\eps^2_{\operatorname{JL}}}\right)
        \]
        with distortion $(1+\eps_{\operatorname{JL}})$.
        
        \item Compute the mean $\mu \gets \frac{1}{n}\sum_{x \in \X} x$.
        
        \item Center each point: $x \gets x - \mu$ for all $x \in \X$.
        
        \item Compute $\|\X\|_{\operatorname{F}}^2 \gets \sum_{x \in \X} \|x\|^2$.
        
        \item Compute $\kappa(x) \gets \|x\|^2 / \|\X\|_{\operatorname{F}}^2$ for all
        $x \in \X$.
    \end{enumerate}

    \item \textbf{Initialization:}
    \begin{enumerate}
        \item Initialize the data structure $\textsc{Init}(\L, \rho)$.
        
        \item Sample $c_1 \sim \textsc{Uniform}(\X)$.
        
        \item Set $C_1 \gets \{c_1\}$ and insert $c_1$ into $\L$.
    \end{enumerate}

    \item \textbf{Main loop:}
    \begin{enumerate}
        \item \textbf{for} $t = 2$ \textbf{to} $k$ \textbf{do}
        \begin{enumerate}
            \item Initialize ${iter} \gets 0$ and sample $s \sim \textsc{Uniform}(\X)$.
            
            \item \textbf{while} ${iter} < m \ln k$ \textbf{do}
            \begin{enumerate}
                \item ${iter} \gets {iter} + 1$.
                
                \item Sample $x \sim \kappa(\cdot \mid C_{t-1})$, where
                \[
                    \kappa(x \mid C_{t-1})
                    :=
                    \frac{\|x\|^2 + \|c_1\|^2}
                         {\|\X\|_{\operatorname{F}}^2 + n\|c_1\|^2}.
                \]
                
                \item Sample $R \sim \textsc{Uniform}(0,1)$.
                
                \item Compute
                \[
                    r \gets
                    \frac{\cost(x, \Query(\L, x))}
                         {2 \rho^{-1}\bigl(\|x\|^2 + \|c_1\|^2\bigr)}.
                \]
                
                \item \textbf{if} $R \le r$ \textbf{then}
               
                    Set $s \gets x$.
                    \textbf{break}.
                
            \end{enumerate}
            
            \item Set $c_t \gets s$, insert $c_t$ into $\L$, and update
            \[
                C_t \gets C_{t-1} \cup \{c_t\}.
            \]
        \end{enumerate}
    \end{enumerate}

    \item \textbf{Return} $C_k$.
\end{enumerate}
\end{algorithm}

\subsection{Preprocessing}\label{subsec:data structure}

The Preprocessing step ($\textsc{PreProcess}(\X)$, Algorithm~\ref{proc:preproc}) clearly takes a time of $O(nD)$ as it involves a JL-like linear transformation on the dataset, followed by a couple of for loops iterating over the transformed $D$-dimensional dataset of $n$ points.

Consider the vector $v_{\X} \in \R^{n}$ given by $v_{\X}(x) = \|x\|$. Define $\kappa(x) = \frac{\|x\|^2}{\|\X\|_F^2}$ as a distribution over $\X$. We will use a (complete) binary tree data structure to sample from $\kappa(\cdot)$. The leaves of the binary tree correspond to the entries of $v_\X$ and store weight $v_\X(x)^2$ along with the sign of $v_\X(x)$. The internal nodes also store a weight that is equal to the sum of weights of its children. To sample from $\kappa(\cdot)$, we traverse the tree, choosing either to go left or right at each node with probability proportional to the weight of its two children until reaching the leaves. The binary tree similarly supports querying and updating the entries of $v_\X$.

\begin{figure}[h]
\centering
\begin{tikzpicture}[
  every node/.style = {font=\small},
  level 1/.style = {sibling distance=40mm, level distance = 10mm},
  level 2/.style = {sibling distance=20mm, level distance = 10mm},
  level 3/.style = {sibling distance=10mm, level distance = 10mm},
  edge from parent/.style = {draw, -latex}
]

\node {$\| v \|^2$}
  child {node {$v(1)^2 + v(2)^2$}
    child {node {$v(1)^2$}
      child {node {$\operatorname{sign}(v(1))$}}
    }
    child {node {$v(2)^2$}
      child {node {$\operatorname{sign}(v(2))$}}
    }
  }
  child {node {$v(3)^2 + v(4)^2$}
    child {node {$v(3)^2$}
      child {node {$\operatorname{sign}(v(3))$}}
    }
    child {node {$v(4)^2$}
      child {node {$\operatorname{sign}(v(4))$}}
    }
  };

\end{tikzpicture}

\caption{Data structure for sampling from a vector $v \in \R^4$}

\end{figure}
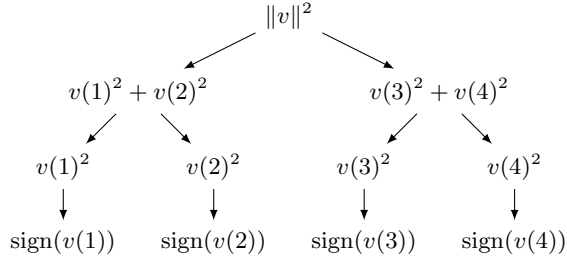

We state this formally following \cite{tang_19}, in which such data structures, called  \textit{sample and query access data structures} were introduced. 

\begin{lemma} \label{lem:bst} (Lemma 3.1 in \cite{tang_19})
There exists a data structure storing a vector $v \in \R^n$ with $\nu$ nonzero entries in $\ O(\nu \log n)$ space which supports the following operations:
\begin{itemize}
    \item Reading and updating an entry of $v$ in $O(\log n)$ time
    \item Finding $\|v\|^2$ in $O(1)$ time
    \item Generating an independent sample $i \in \{1,\dots,n\}$ with probability $\frac{v(i)^2}{\sum_{j = 1}^n v(j)^2}$ in $O(\log n)$ time. 
\end{itemize}

\end{lemma}

Note that if $n$ is not a perfect power of $2$ then we can find a $n' \in \N$ which is a perfect power of 2 such that $n' < n < 2n' $. We can then set the remaining $2n' - n$ data points to have zero norm and use this dataset instead to construct the complete binary tree. Thus the following corollary is immediate.

\begin{corollary}\label{cor:preprocessing time}
    There is a data structure that can be prepared in $O(nD)$ time which enables generating a sample from $\kappa(\cdot)$ in $O(\log n)$ time.
\end{corollary}

Looking ahead, the purpose of preprocessing is to enable fast sampling from the distribution $\kappa(\cdot|C)$ used in $\textsc{Sample}(\L,C,m)$, Algorithm~\ref{proc:sample}.

\subsection{ANNS}\label{subsec:lsh}

In order to employ the rejection sampling idea, the probability of accepting a sampled point $x\in\X$ needs to be proportional to $\cost(x,S)$, where $S$ is the current set of centers. Again, computing this quantity takes $\Theta(kD)$ time, which contributes $\Theta(k^2D)$ time on performing this sampling $k$ times. Thus, we apply another optimization technique; we use an approximate nearest neighbor data structure to approximate the distance between $x$ and the closest center.

In particular, we utilize the following data structure, developed by \citep{cohenaddad2020fast}, which is based on locality-sensitive hash functions \citep{andoni2008near}.

\begin{lemma}{(LSH-based ANNS data structure)}\label{lem:lsh}
For any set of points $P = \{p_1,\dots,p_t\} \subset \R^D$ and parameter $\rho < 1$, there exists a data structure $\L$ with operations $\Insert$ and $\Query$ such that with probability atleast $1 - 1/t$, the following guarantees hold \footnotemark{}
\begin{enumerate}
    \item $\Insert(\L,p)$ inserts the point $p \in P$ into $\L$ in time $O\left(D \log \eta (t \log \eta)^{\rho}\right)$.
    \item $\Query(\L,p)$ returns a point such that $\|p - \Query(\L,p)\| \leq \frac{1}{\sqrt{\rho}} \cdot \min_{q \in \L}\|p-q\|$ and the query time is $O\left(D \log \eta (t \log \eta)^{\rho}\right)$.
    \item $\L$ is monotone under insertions: the distance between $p$ and $\Query(\L,p)$ is non-increasing after inserting more points.
\end{enumerate}
Here, $\eta$ is the aspect ratio of the set $P$ defined as $\eta(P) = \frac{\max_{x \neq y \in P} \|x-y\|}{\min_{x \neq y \in P} \| x-y \|}$.
\end{lemma}

\footnotetext{For the rest of the analysis, we condition on the event that the data structure is \emph{successful}.}

We shall be using this data structure $\L$ to store the sampled centers $S$ and to query the approximate nearest center from $S$. To this end, we define $\cost^\L(P,S) = \sum_{p \in P} \|p-\Query(\L,p)\|^2$ and $\pi^\L(p|S) = \frac{\cost^\L(p,S)}{\cost^\L(P,S)}$ for a set of points $P$ and a set of centers $S \subset \R^D$. Conditioned on the data structure being successful we have the following :

\begin{lemma}[LSH approximation bounds] \label{lem:away}
    Let $X = \{x_1,\dots,x_n\} \subset \R^D$ be a dataset and $S \subset X$ be a set of centers sampled from $X$. Suppose that the data structure $\L$ with parameter $\rho
    < 1$ is successful and each $c \in S$ is inserted into the data structure $\L$. Then for each $x \in X$, the following holds
    \begin{equation*}
        \cost(x,S) \leq \cost^\L(x,S) \leq \rho^{-1} \cdot \cost(x,S)
    \end{equation*}
    from which we also have
    \begin{equation*}
        \rho \cdot \pi(x|S) \leq \pi^\L(x|S) \leq \rho^{-1} \cdot \pi(x|S)
    \end{equation*}
\end{lemma}

\subsection{Analysis of $(\rho,\delta)$-$k$-means++}\label{app:eps-delta-analysis}

The speedup in our algorithm can be attributed to two things; the first being using rejection sampling for sampling for the $D^2$ distribution and second, using the approximate nearest center for fast computation of the rejection sampling probability. Our goal is to bound the expected cost of the output of our algorithm. To do this, we analyze the following \emph{abstract} version of $k$-means++, which we call $(\rho,\delta)$-$k$-means++.  

\begin{algorithm}
\caption{$(\rho,\delta)$-$k$-means++}
\label{alg:rho-delta}
    \textbf{Input:} dataset $\X = \{x_1,\dots,x_n\} \subset \R^D$, parameters $0 < \rho,\delta < 1$ and $k \in \mathbb{N}$
    \begin{algorithmic}[1]
        \STATE $c_1 \gets \operatorname{Uniform}(\X)$
        \STATE $\L \gets$ ANNS data structure with parameter $\rho$
        \STATE $\Insert(c_1,\L)$ ; $C \gets \{c_1\}$
        \FOR{$i = 2$ to $k$}{

            \STATE Sample $y$ with probability $\pi_{(\rho,\delta)}(y|C) := (1-\delta) \cdot \pi^{\L}(y|C) + \delta \cdot \frac{1}{n} $
            \STATE $c_i \gets y$; $\Insert(c_i,\L)$; $C \gets C \cup \{c_i\}$
        
        }
        \ENDFOR
        \STATE \textbf{Output} $C$
    \end{algorithmic}
\end{algorithm}

\subsubsection{The perturbed distribution}

{}
\begin{definition}[\emph{$(\rho,\delta)$ perturbed $D^2$ distribution}]
    Let $\X = \{x_1,\dots,x_n\} \subset \R^d$ be a dataset and $S \subset X$ be a set of centers sampled from $X$. Suppose that the data structure $\L$ with parameter $\rho
    < 1$ is successful and each $c \in S$ is inserted into the data structure $\L$. Then for each $x \in \X$, we define the $(\rho,\delta)$ perturbed $D^2$ distribution as :
    \begin{equation*}
        \pi_{(\rho,\delta)}(x|S) = (1 - \delta)\cdot \pi^\L(x|S) + \delta \cdot \frac{1}{n}
    \end{equation*}
\end{definition}

Now, we provide some useful bounds that a point is chosen according to the perturbed distribution, conditioned on the event that it belongs to a specified set of points beforehand. These bounds will be useful later in several places.  

{}
\begin{lemma}(Bounds on perturbed sampling probabilities) \label{lem:bounds}
Let $\X = \{x_1,\dots,x_n\} \subset \R^d$ be a dataset and $\Y \subseteq \X$ be an arbitrary non-empty subset of $\X$. Let $C \subset \X$ be a set of centers and suppose the data structure $\L$ with parameter $\rho < 1$ is successful and each $c \in C$ is inserted into $\L$. For any point $y \in \Y$ and parameter $\delta \in (0,1)$, the following hold:
\begin{enumerate}
    \item $\Pr_{y \sim \pi_{(\rho,\delta)}(\cdot|C)}[y \mid y \in \Y] \leq \frac{\cost^\L(y,C)}{\cost^\L(\Y,C)} + \frac{\delta}{1 - \delta} \cdot \frac{1}{n} \cdot \frac{\cost^\L(\X,C)}{\cost^\L(\Y,C)}$
    \item $\Pr_{y \sim \pi_{(\rho,\delta)}(\cdot|C)}[y \mid y \in \Y] \geq \frac{\cost^\L(y,C)}{\cost^\L(\Y,C)} - \frac{\delta}{1 - \delta} \cdot \frac{|\Y|}{n} \cdot \frac{\cost^\L(y,C) \cdot \cost^\L(X,C)}{\cost^\L(\Y,C)^2}$ 
\end{enumerate}
\end{lemma}

\begin{proof}
Let us first compute the probability that a point sampled from the perturbed $D^2$ distribution belongs to $\Y$ :  
\begin{equation*}
\Pr_{y \sim \pi_{(\rho,\delta)}(\cdot|C)}[y \in \Y] = \sum_{y' \in \Y} \left( (1 - \delta) \frac{\cost^\L(y',C)}{\cost^\L(\X,C)} + \delta \cdot \frac{1}{n} \right) 
= (1 - \delta) \frac{\cost^\L(\Y,C)}{\cost^\L(\X,C)} + \delta \cdot \frac{|\Y|}{n}
\end{equation*}
Hence the required conditional probability is
\begin{equation*}
\Pr_{y \sim \pi_{(\rho,\delta)}(\cdot|C)}[y\mid y \in \Y] = \frac{(1 - \delta) \frac{\cost^\L(y,C)}{\cost^\L(\X,C)} + \delta \cdot \frac{1}{n}}{(1 - \delta) \frac{\cost^\L(\Y,C)}{\cost^\L(\X,C)} + \delta \cdot \frac{|\Y|}{n}}
\end{equation*}

For part 1, we have:
\begin{align*}
\frac{(1 - \delta) \frac{\cost^\L(y,C)}{\cost^\L(\X,C)} + \delta \cdot \frac{1}{n}}{(1 - \delta) \frac{\cost^\L(\Y,C)}{\cost^\L(\X,C)} + \delta \cdot \frac{|\Y|}{n}} &\leq \frac{(1 - \delta) \frac{\cost^\L(y,C)}{\cost^\L(\X,C)} + \delta \cdot \frac{1}{n}}{(1 - \delta) \frac{\cost^\L(\Y,C)}{\cost^\L(\X,C)}} \\
&= \frac{\cost^\L(y,C)}{\cost^\L(\Y,C)} + \frac{\delta}{1 - \delta} \cdot \frac{1}{n} \cdot \frac{\cost^\L(\X,C)}{\cost^\L(\Y,C)}
\end{align*}

For part 2, we have:
\begin{align*}
\frac{(1 - \delta) \frac{\cost^\L(y,C)}{\cost^\L(\X,C)} + \delta \cdot \frac{1}{n}}{(1 - \delta) \frac{\cost^\L(\Y,C)}{\cost^\L(\X,C)} + \delta \cdot \frac{|\Y|}{n}} &= \frac{\cost^\L(y,C)}{\cost^\L(\Y,C)} \left( \frac{1 + \frac{\delta}{1 - \delta} \cdot \frac{\cost^\L(\X,C)}{n \cdot \cost^\L(y,C)}}{1 + \frac{\delta}{1 - \delta} \cdot \frac{|\Y| \cdot \cost^\L(\X,C)}{n \cdot \cost^\L(\Y,C)}} \right) \\
&\geq \frac{\cost^\L(y,C)}{\cost^\L(\Y,C)} \left( 1 + \frac{\delta}{1 - \delta} \cdot \frac{|\Y| \cdot \cost^\L(\X,C)}{n \cdot \cost^\L(\Y,C)} \right)^{-1} \\
&\geq \frac{\cost^\L(y,C)}{\cost^\L(\Y,C)} \left( 1 - \frac{\delta}{1 - \delta} \cdot \frac{|\Y| \cdot \cost^\L(\X,C)}{n \cdot \cost^\L(\Y,C)} \right)
\end{align*}
where in the last step we used $(1 + x)^{-1} \geq 1 - x$ for any $x \geq 0$.
\end{proof}

\subsubsection{Some useful lemmas}

Suppose $S = \{s_1,\dots,s_k\}$ is am optimal set of centers for the dataset $\X$ so that $\cost(\X,S) = \opt_k(\X)$. Let $\X_j$ denote the set of points from $\X$ whose nearest center is $s_j$. This induces a partition $\X = \bigcup_{j \in [k]} \X_j$ where $\X_i \cap  \X_j = \emptyset$ and $s_j = \mu(\X_j)$. Notice that we can decompose the optimal $k$-means cost as $\opt_k(\X) = \sum_{j \in [k]} \opt_1(\X_j)$  Let us state a few results from \citep{arthur2007kmeanspp} that will be useful.

{}
\begin{lemma}(First center cost, Lemma 3.1 of \citep{arthur2007kmeanspp}) \label{lem:first-center}
    Let $c_1$ be the first center chosen by Algorithm~\ref{alg:rho-delta}. For an optimal choice of cluster centers $S$, let $\X_j$ denote the cluster corresponding to the $j$th center $s_j$ for $j\in[k]$. Then:
    \begin{equation*}
        \E[\cost(\X_j,c_1) | c_1 \in \X_j] = 2\opt_1(\X_j)
    \end{equation*}
\end{lemma}

{}
\begin{lemma}[Standard $k$-means++ bound (Lemma 3.2 of \citep{arthur2007kmeanspp})] \label{lem:kmeanspp-bound}
    For an optimal choice of cluster centers $S$, let $\X_j$ denote the cluster corresponding to the $j$th center for $j\in[k]$. Let $C$ be a set of chosen centers, and let $y$ be a new center chosen according to the (standard) $D^2$ distribution such that $y\in \X_j$. Then:
    \begin{equation*}
        \E[\cost(\X_j, C \cup \{y\}) \mid C,\{y\in \X_j\}]\leq 8\opt_1(\X_j)
    \end{equation*}
\end{lemma}

The following is an analogous version of Lemma 3.2 of \citep{arthur2007kmeanspp}, but we need additional analysis since we use the perturbed distribution instead of the exact $D^2$ distribution:

{}
\begin{lemma}[Perturbed center selection bound] \label{lem:other centers}
    Consider an iteration of Algorithm~\ref{alg:rho-delta} after the first one and let $C$ be the current set of centers. Let $y$ be the center chosen in the current iteration. For an optimal choice of cluster centers $S$, let $\X_j$ denote the cluster corresponding to the $j$th center for $j\in[k]$. Then for any $C \neq \emptyset$:
    \begin{equation*}
        \E[\cost(\X_j,C\cup \{y\}) \mid C, \{y \in \X_j\}] \leq \rho^{-1}\left( 8 \opt_1(\X_j) + \frac{\delta}{1 - \delta} \frac{|\X_j|}{n}\cost(\X,C) \right)
    \end{equation*}
\end{lemma}

\begin{proof}
When the new center $y$ is added, each point $x \in \X_j$ contributes $\cost(x, C \cup\{y\}) = \min(\|x-y\|^2, \cost(x,C))$ to the overall cost, where $\L$ now contains $C \cup \{y\}$. The expected cost of the cluster $\X_j$ can hence be written as:
\begin{DispWithArrows*}[wrap-lines]
&\E[\cost(\X_j,C\cup\{y\}) \mid C, \{y \in \X_j\}] \Arrow{Expanding expectation}\\ \\  
&= \sum_{y \in \X_j} \Pr_{y \sim \pi_{(\rho,\delta)}(\cdot|C)}[y \mid y \in \X_j] \cdot \cost(\X_j,C \cup \{y\}) \Arrow{Expanding $\cost(\X_j,C\cup\{y\})$ }\\ \\
&= \sum_{y \in \X_j} \Pr_{y \sim \pi_{(\rho,\delta)}(\cdot|C)}[y \mid y \in \X_j] \cdot \left( \sum_{x \in \X_j}\cost(x, C \cup\{y\}) \right) \Arrow{Lemma~\ref{lem:bounds} item (1)}\\ \\ 
& \leq \sum_{y \in \X_j} \left( \frac{\cost^\L(y,C)}{\cost^\L(\X_j,C)} + \frac{\delta}{1 - \delta} \frac{1}{n} \frac{\cost^\L(\X,C)}{\cost^\L(\X_j,C)}  \right) \cdot \left( \sum_{x \in \X_j}\cost(x, C \cup\{y\}) \right)\\ \\ 
&\leq \underbrace{\sum_{y \in \X_j}  \frac{\cost^\L(y,C)}{\cost^\L(\X_j,C)} \sum_{x \in \X_j} \min(\|x-y\|^2, \cost(y,C))}_{(I)}\\& + \underbrace{\sum_{y \in \X_j} \left(\frac{\delta}{1 - \delta} \frac{1}{n} \frac{\cost(\X,C)}{\cost^\L(\X_j,C)}  \right) \sum_{x \in \X_j} \min(\|x-y\|^2, \cost(x,C))}_{II}
\end{DispWithArrows*}
where in the third step we used part 1 of Lemma~\ref{lem:bounds}. Using Lemma~\ref{lem:away} and Lemma 3.2 of \citep{arthur2007kmeanspp}, we have for $(I)$ : 
\begin{align*}
 \sum_{y \in \X_j}  \frac{\cost^\L(y,C)}{\cost^\L(\X_j,C)} \sum_{x \in \X_j} \min(\|x-y\|^2, \cost(x,C)) \leq  8 \rho^{-1}\opt_1(\X_j)
\end{align*}

For the second term $(II)$, noting that $\sum_{x \in \X_j} \min(\|x-y\|^2, \cost(x,C)) \leq \sum_{x \in \X_j}\cost(x,C) = \cost(\X_j,C)$, we have:
\begin{align*}
&\sum_{y \in \X_j} \left(\frac{\delta}{1 - \delta} \frac{1}{n} \frac{\cost^\L(\X,C)}{\cost^\L(\X_j,C)}  \right) \sum_{x \in \X_j} \min(\|x-y\|^2, \cost(x,C)) \\
&\leq \rho^{-1} \cdot \sum_{y \in \X_j} \left(\frac{\delta}{1 - \delta} \frac{1}{n} \frac{\cost(\X,C)}{\cost(\X_j,C)}  \right)\cost(\X_j,C) = \frac{\rho^{-1}\delta}{1 - \delta} \frac{|\X_j|}{n} \cost(\X,C)
\end{align*}
Combining both terms completes the proof.
\end{proof}

\subsubsection{Potential Analysis}

Following \cite{dasgupta_03}, we first setup some notation.  Let $t \in \{0,\dots,k\}$ denote the number of centers already chosen by Algorithm~\ref{alg:rho-delta}. Let $C_t := \{c_1, \dots, c_t\}$ be the set of centers after $t$ iterations . We say that an optimal cluster $\X_j$ is \textit{covered} by $C_t$ if at least one of the chosen centers is in $\X_j$. If not, then it is \textit{uncovered}. We denote ${H}_t = \{j \in \{1,\dots,k\} : \X_j \cap C_t \neq \emptyset \}$ and ${U}_t = \{1,\dots,k\} \backslash {H}_t$. Similarly, the dataset $\X$ can be partitioned into two parts: $\H_t \subset \X$ being the points belonging to \textit{covered} clusters and $\U_t = \X \backslash \H_t$ being the points belonging to \textit{uncovered} clusters. Let ${W}_t = t - |{H}_t|$ denote the number of \textit{wasted} iterations so far, i.e., the number of iterations in which no new cluster was covered. Note that we always have $|{H}_t| \leq t$ and hence $|{U}_t| \geq k-t$. For any $\Y \subset \X$, we use the notation $\cost_t(\Y) = \cost(\Y, C_t)$ and $\cost_t^\L(\Y) = \cost^\L(\Y,C_t)$ for brevity. The total cost can be decomposed as:
\begin{equation*}
\cost_t(\X) = \cost_t(\H_t) + \cost_t(\U_t) \leq \cost_t(\H_t) + \cost_t^\L(\U_t)
\end{equation*}

We can bound the first term directly using the previous lemmas.

{}
\begin{lemma}(Cost bound for covered clusters) \label{lem:cost-H}
For each $t \in \{1,\dots,k\}$ the following holds:
\begin{equation*}
\E[\cost_t(\H_t)] \leq \rho^{-1}\left(8 \opt_k(\X) + \frac{2\delta}{1 - \delta}\opt_1(\X)\right)
\end{equation*}
\end{lemma}

\begin{proof}

Consider any cluster $j \in H_t$. This means that at some iteration $t_j$, a point $y \in \X_j$ would have been selected as the new center $c_{t_j}$ so that $j$ becomes covered. This observation allows us to use Lemma~\ref{lem:other centers} in the following manner : 
\begin{DispWithArrows*}[wrap-lines]
\E[\cost_t(\H_t)] &= \E\left[\sum_{j \in {H}_t} \cost_t(\X_j)\right] \Arrow{Using linearity of expectation}\\ \\
&= \sum_{j \in H_t} \E[\cost_t(\X_j)]\Arrow{Using Lemma~\ref{lem:other centers}}\\ \\
&\leq \sum_{j \in H_t} \rho^{-1}\left(8\opt_1(\X_j) + \frac{\delta}{1 - \delta} \frac{|\X_j|}{n} \cost(\X,C_{t_j-1})\right) \Arrow{Using $\E[\cost(\X,C_{t_j-1})] \leq \E[\cost(\X,c_1)] = 2\opt_1(\X)$}\\ \\
&\leq \sum_{j \in H_t} \rho^{-1}\left(8\opt_1(\X_j) + \frac{\delta}{1 - \delta} \frac{|\X_j|}{n} \cdot 2\opt_1(\X) \right) \Arrow{Expanding the sum over $j \in H_t$ to $j \in [k]$}\\\\
&= \rho^{-1}\left(8 \sum_{j\in [k]} \opt_1(\X_j) + \frac{2\delta}{1 - \delta} \cdot  \frac{\sum_{j\in [k]} |\X_j|}{n} \cdot \opt_1(\X)\right) \Arrow{Using $\opt_k(\X) = \sum_{j \in [k]} \opt_1(\X_j)$} \\\\ 
&= \rho^{-1}\left(8\opt_k(\X) + \frac{2\delta}{1 - \delta} \opt_1(\X)\right)
\end{DispWithArrows*}

This completes the proof. 

\end{proof}

\textbf{Potential function.}
To bound the second term, i.e., the cost of the uncovered clusters, we use the potential function method introduced in \citep{dasgupta2013lec3}, with a slightly modified potential function:
\begin{equation*}
\phi_t = \frac{{W}_t}{|{U}_t|} \cost_t^\L(\U_t)
\end{equation*}

Instead of \textit{paying} the complete clustering cost of the uncovered clusters at once, we make sure that at the end of iteration $t$, we have paid an amount of at least $\phi_t$. Observe that when $t = k$, we have ${W}_k = |{U}_k|$, so the potential becomes $\cost^\L_k(\U_k)$, which is atleast $\cost_k(\U_k)$, i.e. the total cost of the uncovered clusters returned by Algorithm~\ref{alg:qkmeans}. We now show how to bound the expected increase in the potential, i.e., $\phi_{t+1} - \phi_t$ at each iteration. The total potential bound will then follow by expressing $\phi_k$ as a telescoping sum of the increments.

 Suppose $t$ centers have been chosen. The next center $c_{t+1}$ is chosen which belongs to some optimal cluster $\X_j$. We consider two cases: the first case is when $j \in {U}_t$, i.e., a new cluster is covered, and the second case is when $j \in {H}_t$, i.e., the center is chosen from an already covered cluster. We shall denote the set of all random variables after the end of iteration $t$ by $\mathtt{F}_t$.

{}
\begin{lemma}(Potential increment for uncovered clusters)\label{lem:uncovered}
For any $t \in \{1,\dots,k-1\}$, the following holds:
\begin{align*}
\E[\phi_{t+1} - \phi_t \mid \mathtt{F}_t, \{j \in {U}_t\}] &\leq \frac{2\delta}{1 - \delta} \frac{t}{\max(1,k-t-1)^2} \opt_1(\X)
\end{align*}
\end{lemma}

\begin{proof}
When $j \in {U}_t$, we have ${W}_{t+1} = {W}_t$, ${H}_{t+1} = {H}_t \cup\{j\}$, and ${U}_{t+1} = U_{t} \backslash\{j\}$. Thus using the monotonicity of the data structure $L$, we have
\begin{align*}
\phi_{t+1} = \frac{{W}_{t+1}}{|{U}_{t+1}|} \cost_{t+1}^\L(\U_{t+1}) \leq \frac{{W}_t}{|{U}_t| - 1} \left( \cost_t^\L(\U_t) - \cost_t^\L(\X_j) \right)
\end{align*}
We can use part 2 of Lemma~\ref{lem:bounds} for getting a lower bound on the second term:
\begin{DispWithArrows*}[wrap-lines]
&\E[\cost_t^\L(\X_j) \mid \mathtt{F}_t , \{j \in {U}_t\}] \\
&\geq \sum_{j \in {U}_t} \left( \frac{\cost_t^\L(\X_j)}{\cost_t^\L(\U_t)} - \frac{\delta}{1 - \delta} \frac{|\X_j|}{n} \frac{\cost_t^\L(\X_j)\cost_t^\L(X)}{\cost_t^\L(\U_t)^2} \right) \cost_t^\L(\X_j) \\
&\geq \left(1 - \frac{\delta}{1- \delta} \frac{\cost_t^\L(\X)}{\cost_t^\L(\U_t)} \right) \sum_{j \in {U}_t} \frac{\cost_t^\L(\X_j)^2}{\cost_t^\L(\U_t)}
\end{DispWithArrows*}
where in the second step we used the fact that $|\X_j| \leq n$ for each $j \in \mathcal{U}_t$. We can use the Cauchy-Schwarz inequality\footnote{For lists of numbers $a_1,\dots,a_m$ and $b_1,\dots, b_m$ we have $\left(\sum_i a_ib_i\right)^2 \leq \left(\sum_i a_i^2\right)\left(\sum_i b_i^2\right)$.} to simplify the last expression as follows:
\begin{equation*}
|{U}_t|^2\sum_{j \in {U}_t} \cost_t^\L(\X_j)^2 \geq |{U}_t| \sum_{j \in {U}_t} \cost_t^\L(\X_j) = |{U}_t| \cost_t^\L(\U_t)
\end{equation*}
This shows that
\begin{equation*}
\E[\cost_t^\L(\X_j) \mid \mathtt{F}_t , \{j \in {U}_t\}] \geq \frac{\cost_t^\L(\U_t)}{|{U}_t|} - \frac{\delta}{1 - \delta} \frac{\cost_t^\L(\X)}{|{U}_t|}
\end{equation*}
Now,
\begin{align*}
\E[\phi_{t+1} \mid \mathtt{F}_t, i \in {U}_t] &\leq \frac{{W}_t}{|{U}_t| - 1} \left(\cost_t^\L(\U_t) - \E[\cost_t^\L(\X_j)| \mathtt{F}_t, i \in {U}_t] \right) \\
&\leq \frac{{W}_t}{|{U}_t| - 1} \left(\cost_t^\L(\U_t) - \frac{\cost_t^\L(\U_t)}{|{U}_t|} + \frac{\delta}{1 - \delta} \frac{\cost_t^\L(X)}{|{U}_t|} \right) \\
&= \phi_t + \frac{\delta}{1 - \delta} \frac{{W}_t}{|{U}_t|\left(|{U}_t|-1\right)} \cost_t^\L(\X)
\end{align*}
Recall that ${W}_t \leq t$ and $|{U}_t| \geq k-t$. So for $t \leq k-2$, the following holds after taking expectation:
\begin{align*}
\E[\phi_{t+1} - \phi_t\mid \mathtt{F}_t, \{j \in {U}_t\}]
&\leq \frac{\delta}{1 - \delta} \frac{t}{(k-t-1)^2} \E[ \cost_t^\L(X) \mid \mathtt{F}_t, j \in {U}_t] \\
&\leq \frac{2\delta}{1 - \delta} \frac{t}{(k-t-1)^2} \opt_1(X)
\end{align*}
Now consider the case when $t= k-1$. We cannot use the above argument directly because it may so happen that $|{U}_{k-1}| = 1$. If this happens, the potential of the uncovered clusters is always $0$. This only happens when a new cluster is covered in each iteration. Let this event be $\mathtt{A}$ (for \textit{All Clusters} being covered). Denoting $\mathtt{E}$ to be the event: $\{\mathtt{F}_{k-1}, \{j \in U_{k-1}\}\}$, we have the following:
\begin{DispWithArrows*}[wrap-lines]
\E[\phi_{k} - \phi_{k-1} \mid \mathtt{E}] &= \E[\phi_{k} - \phi_{k-1} \mid \mathtt{E},\mathtt{A} ] \Pr[\mathtt{A} \mid \mathtt{E}] + \E[\phi_{k} - \phi_{k-1} \mid \mathtt{E}, \neg \mathtt{A} ] \Pr[ \neg \mathtt{A} \mid \mathtt{E}]  \Arrow{Using $\E[\phi_{k} - \phi_{k-1} \mid \mathtt{E},\mathtt{A} ] \leq 0$} \\ \\
&\leq \E[\phi_{k} - \phi_{k-1} \mid \mathtt{E}, \neg \mathtt{A} ] \Pr[ \neg \mathtt{A} \mid \mathtt{E}] \\
&\leq \E[\phi_{k} - \phi_{k-1} \mid \mathtt{E}, \neg \mathtt{A} ] \\
&\leq \frac{2\delta t}{1 - \delta} \opt_1(X)
\end{DispWithArrows*}
where in the last line we used the fact that $|U_{k-1}| > 1$ if all clusters are not covered. We also used the fact that $\E[\phi_k-\phi_{k-1}|\mathtt{E,A}] \leq 0$ since in this case there are no uncovered clusters .Combining both cases completes the proof.
\end{proof}

{}
\begin{lemma}(Potential increment for covered clusters) \label{lem:covered}
For any $t \in \{1,\dots,k-1\}$, the following holds:
\begin{equation*}
\E[\phi_{t+1} - \phi_t \mid \mathtt{F}_t, \{i \in {H}_t\}] \leq \frac{\cost_t^\L(\U_t)}{k-t}
\end{equation*}
\end{lemma}

\begin{proof}
When $i \in {H}_t$, we have ${H}_{t+1} = {H}_t$, ${W}_{t+1} = \mathcal{W}_t + 1$, and ${U}_{t+1} = {U}_t$. Thus,
\begin{align*}
\phi_{t+1} - \phi_t &= \frac{{W}_{t+1}}{|{U}_{t+1}|} \cost_{t+1}^\L(\U_{t+1}) - \frac{W_{t}}{|{U}_{t}|} \cost_t^\L(\U_{t}) \\
&\leq \frac{W_{t}+1}{|{U}_{t}|} \cost_t^\L(\U_{t})-\frac{W_{t}}{|{U}_{t}|} \cost_t^\L(\U_{t}) \\
&= \frac{\cost_t^\L(\U_t)}{|{U}_t|} \leq \frac{\cost_t^\L(\U_t)}{k-t}
\end{align*}
\end{proof}

We can now combine the two cases to get:

{}
\begin{lemma}[Combined potential increment bound] \label{lem:combined}
For any $t \in \{1,\dots,k-1\}$, the following holds:
$
\E[\phi_{t+1} - \phi_t \mid \mathtt{F}_t] \leq (1 - \delta) \frac{\E[\cost_t^\L(\H_t)]}{k-t} \\
+ \delta\left(\frac{2}{k-t} + \frac{2t}{\max(1,k-t-1)^2}\right) \opt_1(\X)
$
\end{lemma}

\begin{proof}
To compute the overall expectation, we have:
\begin{align*}
&\E[\phi_{t+1}-\phi_t \mid \mathtt{F}_t] = \E[\phi_{t+1}-\phi_t \mid \mathtt{F}_t, \{i \in {U}_t\}] \Pr[i \in {U}_t]  + \E[\phi_{t+1}-\phi_t \mid \mathtt{F}_t, \{i \in{H}_t\}] \Pr[i \in {H}_t]
\end{align*}
We can bound the first term using Lemma~\ref{lem:uncovered}:
\begin{align*}
&\E[\phi_{t+1}-\phi_t \mid \mathtt{F}_t, i \in {U}_t] \Pr[i \in {U}_t] \\ &\leq \E[\phi_{t+1}-\phi_t \mid \mathtt{F}_t, i \in {U}_t] \\ &\leq \frac{2\delta}{1 - \delta} \frac{t}{\max(1,k-t-1)^2} \opt_1(\X)
\end{align*}
and the second term using Lemma~\ref{lem:covered}:
\begin{align*}
&\E[\phi_{t+1}-\phi_t \mid \mathtt{F}_t, i \in H_t] \Pr[i \in {H}_t] \\ &\leq \frac{\cost_t^\L(\U_t)}{k-t} \left( (1 - \delta) \frac{\cost_t^\L(\H_t)}{\cost_t^\L(\X)} + \delta \frac{|\H_t|}{n} \right)\\ &\leq (1 - \delta)\frac{\cost_t^\L(\H_t)}{k-t} + \delta \frac{\cost_t^\L(\X)}{k-t}
\end{align*}
where in the last step we used $\cost_t^\L(\U_t) \leq \cost_t^\L(\X)$ and $|\H_t| \leq n$. Combining both terms completes the proof.
\end{proof}

\subsubsection{Final approximation guarantee}
We are now ready to provide a proof for the main theorem regarding $(\rho,\delta)$-$k$-means++:
{}
\begin{theorem}[Main approximation guarantee for $(\rho,\delta)$-$k$-means++] \label{thm:approx guarantee}
Let $\X \subset \R^D$ be any dataset which is to be partitioned into $k$ clusters. Let $C$ be the set of centers returned by Algorithm~\ref{alg:rho-delta} for any $\delta \in (0,0.5)$ and $\rho < 1$. The following approximation guarantee holds:
\begin{equation*}
\E[\cost(\X,C)] \in O(\rho^{-2} \log k)\opt_k(\X) + \delta  \cdot O(k + \rho^{-2} \log k) \opt_1(\X)
\end{equation*}
\end{theorem}

\begin{proof}
At the end of $k$ iterations, we have $\cost(\X,C) \leq \cost_k(\H_k) + \cost_k^\L(\U_k) = \cost_k(\H_k) + \phi_k$. The first term can be bounded using Lemma~\ref{lem:cost-H}. For the second term, we can express $\phi_k$ as a telescopic sum:
\begin{align*}
&\E[\cost(\X,C)] \\&\leq  \E[\cost_k(\H_k)] + \sum_{t = 0}^{k-1} \E [\phi_{t+1} - \phi_t \mid \mathtt{F}_t] \\
&\leq \E[\cost_k(\H_k)] + \sum_{t = 0}^{k-1} (1 - \delta) \frac{\E[\cost_t^\L(\H_t)]}{k-t}  + \sum_{t = 0}^{k-1}\delta \left( \frac{2}{k-t} + \frac{2t}{(1 - \delta) \max(1,k-t-1)^2} \right) \opt_1(\X) \\
&\leq \E[\cost_k(\H_k)] + \sum_{t = 0}^{k-1} \frac{(1 - \delta)}{ \eps}\frac{\E[\cost_t(\H_t)]}{k-t}  + \sum_{t = 0}^{k-1}\delta \left( \frac{2}{k-t} + \frac{2t}{(1 - \delta) \max(1,k-t-1)^2} \right) \opt_1(\X) \\
&\leq \rho^{-1}\left(8\opt_k(\X) + \frac{2\delta}{1 - \delta} \opt_1(\X)\right)\left( 1 + (1 - \delta) \cdot \eps^{-1} \cdot \sum_{t = 0}^{k-1} \frac{1}{k-t}\right) \\
&\quad\quad + \frac{2\delta}{1 - \delta} \opt_1(X) \left(k + (1-\delta)(1+\ln k) + \sum_{t=0}^{k-2} \frac{t}{(k-t-1)^2} \right)
\end{align*}

To simplify this, note that $\sum_{t = 0}^{k-1} \frac{1}{k-t} \leq 1 + \ln k$ and  $\sum_{t = 0}^{k-2} \frac{t}{(k-t-1)^2} \leq k \sum_{t = 1}^{\infty}t^{-2} = k \pi^2 /{6} $.
\begin{align*}
&\E[\cost(\X,C)] \\ &\leq \rho^{-1}\left(8\opt_k(\X) + \frac{2\delta}{1 - \delta}\opt_1(\X)\right)\left(1 + \eps^{-1}(1+ \ln k)\right)+ 2\delta  \opt_1(\X) \left(k + 2 + 2\ln k  + k\pi^2/3\right) \\
\end{align*}

For sufficiently large $k$, we obtain :

\begin{align*}
\E[\cost(\X,C)] \in  O(\rho^{-2} \log k)\opt_k(\X) + \frac{\delta}{1 - \delta} \cdot O(k + \rho^{-2} \log k) \opt_1(\X)
\end{align*}
This completes the proof of the theorem.
\end{proof}

\subsection{Rejection Sampling}

The naive implementation of $k$-means++ seeding involves iteratively selecting each cluster center using the $D^2$ distribution with respect to the current set of centers. Sampling a point from this distribution naively is expensive as it takes $O(nkd)$ time. To obtain a faster implementation, we sample each center from a distribution that \emph{oversamples} the $D^2$ distribution (but is faster to sample from), and then perform rejection sampling to ensure that the overall distribution remains close to $D^2$.

Let us formally define what we mean by \emph{oversampling} a distribution.

{}
\begin{definition}($\tau$-oversampling)
Suppose $\pi$, $\kappa$ define probability distributions over $\X$. The distribution $\kappa$ is said to $\tau$-\textit{oversample} $\pi$ for $\tau > 0$ if  $\pi(x) \leq \tau \cdot \kappa(x)$ for each $x \in \X$.
\end{definition}

Let $\pi, \kappa$ be probability distributions over $\X$ such that $\kappa$ $\tau$-\textit{oversamples} $\pi$. Suppose we have a collection of samples $x_1,x_2,\ldots$ from the distribution $D_2$. Consider the following strategy $\mathsf{RejectionSample}$: iterate through samples $\{x_i\}_i$ and terminate if a sample $x_i$ is accepted with probability $\rho(x_i) = \frac{\pi(x_i)}{\tau \kappa(x_i)}$.

It is not difficult to argue that an accepted sample comes from the distribution $\pi$. Moreover, for any $\varepsilon \in (0,1)$ it takes at most $\tau\ln(1/\varepsilon)$ samples from $\kappa$ to accept a sample with probability at least $1-\varepsilon$. Note that the above strategy may not need to know $\tau$ in advance (indeed, computing $\tau$ may be non-trivial), but only requires the ability to compute the quantity $\rho(x_i)$.

In the current form, our strategy does not have any control over the number of samples from $\kappa$ which it may need to examine. However, a bound on the number of samples to be examined can be obtained if we are content with sampling from a slightly perturbed distribution. Suppose we have another distribution $\nu$ over $\X$. This time we are allowed to use samples coming from $\kappa$ and $\nu$ and instead of a sample from $\pi$, we are content with obtaining a sample generated by a hybrid distribution $\pi_\delta(x) = (1-\delta)\pi(x) + \delta \nu(x)$ for some small enough $\delta \in (0,1)$. For this we can modify the above strategy which we now call $\mathsf{RejectionSample}(m)$: Iterate through $m$ samples $x_1,\ldots,x_m$ from $\kappa$, terminate if a sample $x_i$ is accepted with probability $\rho(x_i)$. If no sample is accepted, terminate with a sample from $\nu$. It can be shown that the \textit{failure} probability $\delta$ diminishes with increasing $m$. Indeed, the following holds:

{}
\begin{lemma}\label{lem:failure probability bound}(Failure probability bound)
The failure probability of $\mathsf{RejectionSample}(m)$ satisfies $\delta \leq e^{-m/\tau}$.
\end{lemma}

In our algorithm, we appropriately choose the distributions $\kappa$ and $\nu$ over $\X$ so that $\kappa$ $\tau$-\textit{oversamples} $\pi$ (for a suitable $\tau$) and for which we obtain samples efficiently.

\subsection{Proof of Correctness for Rejection Sampling}

We now prove that the rejection sampling procedure indeed produces samples from the desired distribution.

{}
\begin{lemma}(Rejection sampling correctness)
Let $\pi, \kappa$ be probability distributions over $\X$ such that $\kappa$ $\tau$-oversamples $\pi$. If we sample $x$ from $\kappa$ and accept it with probability $\rho(x) = \frac{\pi(x)}{\tau \kappa(x)}$, then the accepted sample comes from distribution $\pi$.
\end{lemma}

\begin{proof}
For any $x \in \X$, the probability that $x$ is the accepted sample is:
\begin{align*}
\Pr[\text{accept } x] &= \Pr[\text{sample } x \text{ from } \kappa] \cdot \Pr[\text{accept} \mid \text{sampled } x] \\
&= \kappa(x) \cdot \frac{\pi(x)}{\tau \kappa(x)} = \frac{\pi(x)}{\tau}
\end{align*}

The total probability of accepting any sample is:
\begin{equation*}
\Pr[\text{accept}] = \sum_{x \in \X} \Pr[\text{accept } x] = \sum_{x \in \X} \frac{\pi(x)}{\tau} = \frac{1}{\tau}
\end{equation*}

Therefore, the conditional probability that the accepted sample is $x$ is:
\begin{equation*}
\Pr[x \mid \text{accept}] = \frac{\Pr[\text{accept } x]}{\Pr[\text{accept}]} = \frac{\pi(x)/\tau}{1/\tau} = \pi(x)
\end{equation*}

Thus, accepted samples follow distribution $\pi$.
\end{proof}

{}
\begin{lemma}\label{lem:expected samples}(Expected number of samples)
Using the rejection sampling procedure with $\kappa$ $\tau$-oversampling $\pi$, the expected number of samples from $\kappa$ needed to obtain one accepted sample is exactly $\tau$.
\end{lemma}

\begin{proof}
Each sample from $\kappa$ is accepted with probability $1/\tau$. Therefore, the number of samples follows a geometric distribution with success probability $p = 1/\tau$, which has expected value $\tau$.
\end{proof}

{}
\begin{lemma}(Bounded rejection sampling) \label{lem:bounded rejection sampling}
For $\mathsf{RejectionSample}(m)$ with fallback distribution $\nu$, the output distribution is $(1-\delta)\pi + \delta \nu$ where $\delta \leq e^{-m/\tau}$.
\end{lemma}

\begin{proof}
The probability that no sample is accepted after $m$ trials is:
\begin{equation*}
\delta = \left(1 - \frac{1}{\tau}\right)^m \leq e^{-m/\tau}
\end{equation*}
where we used the inequality $(1-x) \leq e^{-x}$ for $x \geq 0$.

With probability $1-\delta$, a sample is accepted and comes from $\pi$ (by the first lemma). With probability $\delta$, we fallback to sampling from $\nu$. Therefore, the output distribution is $(1-\delta)\pi + \delta \nu$.
\end{proof}

\subsection{Using rejection sampling for $Qk$-means++} 

In the context of our algorithm, we identify the proposal distribution as $\kappa = \kappa(\cdot | S)$, the target distribution as $\pi = \pi(\cdot | S)$ and the fallback distribution as $\nu = \operatorname{uniform}(\X)$. We now show that $\kappa$ $\tau$-oversamples $\pi$ for a suitable $\tau$.

\begin{lemma}(Oversampling factor)\label{lem:oversample} The distribution $\kappa(\cdot|s_1)$ $\tau$-oversamples the distribution $\pi^\L(\cdot|S)$ for 
\begin{equation*}
    \tau = 2 \rho^{-1}\frac{\|\X\|_F^2 + n\|s_1\|^2}{\cost(\X,S)}
\end{equation*}
\end{lemma}

\begin{proof}
    The key to this observation is the following use of Cauchy-Schwartz inequality : 
    \begin{equation*}
        \cost^\L(x,S) \leq \rho^{-1}\cost(x,S) \leq \rho^{-1}\|x-s_1\|^2 \leq 2\rho^{-1}(\|x\|^2 + \|s_1\|^2)
    \end{equation*}
Dividing both sides by $\cost^{L}(\X,S)$, we get 
\begin{equation*}
    \frac{\cost^\L(x,S)}{\cost^\L(\X,S)} \leq \frac{2\rho^{-1}(\|x\|^2 + \|s_1\|^2)}{\cost^\L(\X,S)} \leq \frac{2\rho^{-1}(\|x\|^2 + \|s_1\|^2)}{\cost(\X,S)}
\end{equation*}

from which we can see 
\begin{equation*}
     \frac{\cost^\L(x,S)}{\cost^\L(\X,S)} \leq \left(2 \rho^{-1} \frac{\|\X\|_F^2+n\|s_1\|^2}{\cost(\X,S)}\right) \frac{\|x\|^2+\|s_1\|^2}{\|\X\|_F^2 + n \|s_1\|^2}
\end{equation*} 

This completes the proof. 
\end{proof}

Finally, let us relate this oversampling factor to the quantity $\beta(\X) := \frac{\Delta_1(\X)}{\Delta_k(\X)}$
\begin{lemma}\label{lem:bounding E[tau]}
For the oversampling factor $\tau$ in Lemma~\ref{lem:oversample}, we have $\E[\tau] \leq 4\rho^{-1}\beta(\X)$.    
\end{lemma}

\begin{proof}
    \begin{align*}
        \E[\tau] &\leq 2\rho^{-1}\frac{\|X\|^2 + n\E[\|s_1\|^2]}{\Delta_k(\X)}\\
        &= 2\rho^{-1}\frac{\|X\|^2 + n\cdotp \frac{1}{n}\sum_{x\in\X}\|x\|^2}{\Delta_k(\X)}\\
        &= \frac{4 \rho^{-1}\|\X\|^2}{\Delta_k(\X)}\ = \frac{4\rho^{-1}\Delta_1(\X)}{\Delta_k(\X)}
    \end{align*}
    where the last inequality follows from the fact that $\X$ is translated so that its centroid is at the origin.
\end{proof}

\subsection{Proof of main theorem}
We finally need to combine all the results from the previous sections to prove Theorem~\ref{thm:tradeoffs}, which claims existence of a fast rejection-sampling based $k$-means++ while giving an approximation bound of the clustering from the optimal one. Corollary~\ref{thm:assumption} will then follow as a corollary, together with Assumption~\ref{ass} and Theorem~\ref{thm:scaling-laws}. 

Let us see how $\textsc{QKmeans}(\X,k,m)$ (Algorithm~\ref{alg:qkmeans}) satisfies the claims made by Theorem~\ref{thm:tradeoffs}. As seen in Subsection~\ref{subsec:data structure}, the preprocessing step takes a time of $O(nD)$. Let us now see how it helps in fast sampling from the distribution $\kappa(\cdot|C)$ which is step 4 of $\textsc{Sample}(\L,C,m)$. 

Here is how this step can be performed:

\begin{algorithm}
    \floatname{algorithm}{Procedure}
    \caption{$\mathtt{SampleDistribution}$}
    \label{alg:sample D_2}
    \textbf{Input: } A set of centers $C \subseteq \X$ \\
    \textbf{Output: } A sample according to the distribution $\kappa(\cdot|C)$ defined as $\kappa(x|C) = \frac{\|x\|^2 + \|c_1\|^2}{\|\X\|^2 + n\|c_1\|^2}$ \\
    \vspace{-12pt}
    \begin{algorithmic}[1]
    \STATE Generate $r \sim \textsc{Uniform}([0,1])$
    \IF{$r \leq \frac{\|\X\|_F^2}{\|\X\|_F^2 + n\|c_1\|^2 }$}
        \STATE Generate a sample $x$ with probability $\kappa(x)$
    \ELSE
        \STATE Generate a sample $x \sim \textsc{Uniform}(\X)$
    \ENDIF
    \STATE \textbf{output} $x$
    \end{algorithmic}
\end{algorithm}

\begin{lemma}\label{lem:sample D_2}
    Procedure~\ref{alg:sample D_2} produces a sample from $\X$ according to the distribution $\kappa(\cdot|C)$. Moreover, it takes $O(\log n)$ time after a one time preprocessing time of $O(nD)$.
\end{lemma}
\begin{proof}
    The probability of a sampled point is as follows:
    \begin{align*}
        \Pr[x] &= \Pr\left[r \leq \frac{\|\X\|_F^2}{\|\X\|_F^2 + n\|c_1\|^2 } \right]\Pr[x \sim \kappa(\cdot)]\\
        &\ \ \ \ + \Pr \left[ r > \frac{\|\X\|_F^2}{\|\X\|_F^2 + n\|c_1\|^2} \right]\Pr[x \sim \textsc{Uniform}(\X)] \\
 &= \frac{\|\X\|_F^2}{\|\X\|_F^2 + n\|c_1\|^2 } \frac{\|x\|^2}{\|\X\|_F^2}+ \frac{n\|c_1\|^2}{\|\X\|_F^2 + n\|c_1\|^2 } \frac{1}{n} \\
 &= \frac{(\|x\|^2 + \|c_1\|^2)}{(\|\X\|_F^2 + n\|c_1\|^2)} = \kappa(x|C)
    \end{align*} 
    The time complexity follows from Corollary~\ref{cor:preprocessing time}.
\end{proof}

Next, let us recall from Lemma~\ref{lem:lsh} that the LSH-based ANNS data structure supports (conditioned on success with probability $1-1/t$) a $\Insert(\L,p)$ operation in time $O\left(D \log \eta (t \log \eta)^{\rho}\right)$ and a $\Query(\L,p)$ operation in time $O\left(D \log \eta (t \log \eta)^{\rho}\right)$, where $t$ is the size of the set of points $P$ to be inserted. In our case, the $\textsc{QKmeans}(\X,k,m)$ algorithm only inserts the chosen cluster centers in the dataset, i.e. $t=k$ (Note that the $\Query(\L,p)$ operation can be made for any $p\in\X$). 

Also, recall that here, $D$ is the dimension of the dataset after applying the linear transformation from \cite{makarychev2019jl}, so that $D=O(\log k)$, with a constant distortion in the clustering cost w.r.t. any centers. 

Now, observe that the $\textsc{Insert}$ operation is invoked $k$ times, and the $\textsc{Sample}$ subroutine is invoked $k$ times from the $\textsc{QKmeans}$ algorithm. Within the $\textsc{Sample}$ subroutine, a $\textsc{Query}$ operation and a sampling from $\kappa(\cdot|C)$ operation are invoked atmost $m\ln k$ times. This means that there are a total of $k$ $\textsc{Insert}$ operations, $mk\ln k$ $\textsc{Query}$ operations and $mk\ln k$ sampling operations from $\kappa(\cdot|C)$.

Thus we obtain a time complexity (after the processing step) of 
$$O(kD\log\eta(t\log\eta)^\rho) + O(mk(\ln k)D\log\eta(t\log\eta)^\rho) + O(mk\ln k\log n)$$
which after substituting values of $D,t$ can be simplified to $O(m(\log^2k)(k\log\eta)^{1+\rho} + mk\ln k\log n)$. This completes the time complexity analysis. 

Now let us analyse the approximation guarantee. From Theorem~\ref{thm:approx guarantee}, the $(\rho,\delta)$-$k$-means++ follows an approximation guarantee of: \begin{equation*}
\E[\cost(\X,C)] \in O(\rho^{-2} \log k)\opt_k(\X) + \delta  \cdot O(k + \rho^{-2} \log k) \opt_1(\X)
\end{equation*}
where $\delta$ is the probability of choosing centers via the uniform distribution. From Lemma~\ref{lem:failure probability bound} and Lemma~\ref{lem:oversample}, we get that $\delta\leq e^{-m\ln k/\tau}$ for $\tau = 2 \rho^{-1}\frac{\|\X\|_F^2 + n\|s_1\|^2}{\cost(\X,S)} = 2\rho^{-1}\beta(\X,S)$. Thus $\delta\leq e^{-\frac{\rho m\ln k}{2\beta(\X,S)}}$. Clearly $\delta  \cdot O(k + \rho^{-2} \log k) \leq O(\rho^{-2})e^{-\rho m/\beta}$ and hence we obtain the expected approximation guarantee of Theorem~\ref{thm:tradeoffs}.

Corollary~\ref{thm:rejection-sampling} follows as a corollary from Theorem~\ref{thm:tradeoffs}; if we allow $m=\infty$ in $\textsc{Sample}(\L,C,m)$, we can see that the `break' statement will be issued in expected time $O(\tau)$, since we need $\tau = 2\rho^{-1}\beta(\X,C)$, $\E[\tau] \leq 4\rho^{-1}\beta(\X)$ (from Lemma~\ref{lem:bounding E[tau]}) samples of distribution $\kappa(\cdot|C)$ in expectation to get a desired sample by Lemma~\ref{lem:expected samples}. 

Finally, Corollary~\ref{thm:assumption} can be obtained as a corollary of Corollary~\ref{thm:rejection-sampling} assuming the manifold hypothesis. Theorem~\ref{thm:scaling-laws} gives us $\beta(\X) \leq \beta_0 k^\eps, \eta(\X) \leq \eta_0 n^{3\eps/2}$ for quantization exponent $\eps$. Substituting these in the time complexity bound of Corollary~\ref{thm:rejection-sampling} gives Corollary~\ref{thm:assumption}.

\newpage
\section{Empirical Study} \label{app:empirical}

In this section, we empirically evaluate the extent to which the theoretical predictions developed in this paper are reflected in real-world data. 

\paragraph{Compute environment. }All experiments were conducted on a dual-socket machine equipped with $\operatorname{Intel Xeon Gold 5220R}$ CPUs (48 physical cores total, 2.20 GHz) and 96 hardware threads, running $\operatorname{Linux}$ on $\operatorname{x86-64}$  architecture. Unless otherwise stated, all experiments were executed in a CPU-only setting with a fixed number of threads. Implementations were compiled with aggressive compiler optimizations, and all reported runtimes are wall-clock times averaged over multiple random seeds.

\paragraph{Datasets. } We use datasets across several domains for our empirical evaluation. We include popular image datasets, both as raw pixels as well as image embeddings which are produced through joint vision-language models such as $\operatorname{CLIP-ViT-B/32}$; text datasets in the form of text embeddings generated by the $\operatorname{All-MiniLM-L6-v2}$ model; tabular datasets which record sensor data as well as data generated by scientific experiments. The details of the datasets used are presented in Table~\ref{tab:datasets}. 

\begin{table}[ht]
\centering
\begin{tabular}{p{4cm}lcc}
\toprule
\textsc{Category} & \textsc{Dataset} & $n$ & $d$ \\
\midrule

\multirow{4}{*}{\parbox{4cm}{\centering\textsc{Raw Images}}}
& \textsc{MNIST}     & 60{,}000  & 784    \\
& \textsc{FMNIST}    & 60{,}000  & 784    \\
& \textsc{CIFAR10}   & 50{,}000  & 3{,}072 \\
& \textsc{CIFAR100}  & 50{,}000  & 3{,}072 \\
\midrule

\multirow{4}{*}{\parbox{4cm}{\centering\textsc{CLIP ViT-B/32\\Image Embeddings}}}
& \textsc{MNIST}     & 60{,}000  & 512 \\
& \textsc{FMNIST}    & 60{,}000  & 512 \\
& \textsc{CIFAR10}   & 50{,}000  & 512 \\
& \textsc{CIFAR100}  & 50{,}000  & 512 \\
\midrule

\multirow{2}{*}{\parbox{4cm}{\centering\textsc{all-MiniLM-L6-v2\\Text Embeddings}}}
& \textsc{Reddit}        & 100{,}000 & 384 \\
& \textsc{StackExchange} & 20{,}000  & 384 \\
\midrule

\multirow{2}{*}{\parbox{4cm}{\centering\textsc{Tabular}}}
& \textsc{SUSY} & 5{,}000{,}000 & 18  \\
& \textsc{HAR}  & 10{,}299      & 561 \\
\bottomrule
\end{tabular}
\end{table}


\subsection{Scaling Laws}

We test the hypothesis that the data-dependent parameters $\beta_k = \opt_1/\opt_k$ and $\eta_k$ grow polynomially with the number of clusters $k$, as predicted by optimal quantization theory. For each dataset, we compute $k$-means clusterings for
$k \in \{5,10,50,100,250,500,750,1000\}$ using $10$ independent runs. We fit linear regressions to $\log \beta_k$ and $\log \eta_k$ versus $\log k$ and report slopes, $R^2$, and $95\%$ confidence intervals. Across datasets and modalities, both quantities exhibit reasonable linear behavior in log--log coordinates. The plots are shown in Figures \ref{fig:beta-grid} and \ref{fig:eta-grid}. Table~\ref{tab:scaling-results} provides the estimated quantization exponents and $R^2$ scores for the best fit estimates. The high $R^2$ scores ($ \gtrsim 0.95$) show that the scaling laws fit well on the data. 

\begin{table*}[t]
\centering

\begin{minipage}[t]{0.48\textwidth}
  \centering
  \caption{Scaling law parameters across datasets. $\varepsilon$ is estimated from $\beta_k \sim k^{\varepsilon}$.
  $R^2_\beta$ and $R^2_\eta$ are the coefficients of determination for the $\beta$ and $\eta$ scaling fits respectively.}
  \label{tab:scaling-results}
  \vspace{0.4em}
  \begin{tabular}{lccc}
    \toprule
    Dataset & $\hat{\varepsilon}$ & $\hat{R^2_\beta}$ & $\hat{R^2_\eta}$ \\
    \midrule
    \textsc{MNIST} & 0.1449 & 0.9998 & 0.9944 \\
    \textsc{Fashion-MNIST} & 0.1808 & 0.9847 & 0.9702 \\
    \textsc{CIFAR-10} & 0.0862 & 0.9945 & 0.9661 \\
    \textsc{CIFAR-100} & 0.0961 & 0.9919 & 0.9659 \\
    \textsc{MNIST-CLIP} & 0.1759 & 0.9980 & 0.9981 \\
    \textsc{FMNIST-CLIP} & 0.1665 & 0.9887 & 0.9809 \\
    \textsc{CIFAR-10-CLIP} & 0.1101 & 0.9898 & 0.9737 \\
    \textsc{CIFAR-100-CLIP} & 0.1166 & 0.9988 & 0.9938 \\
    \textsc{Reddit} & 0.0622 & 0.9899 & 0.9814 \\
    \textsc{StackExchange} & 0.0565 & 0.9826 & 0.9898 \\
    \textsc{HAR} & 0.1894 & 0.9984 & 0.8791 \\
    \textsc{SUSY} & 0.2644 & 0.9993 & 0.9693 \\
    \bottomrule
  \end{tabular}
\end{minipage}
\hfill
\begin{minipage}[t]{0.48\textwidth}
  \centering
  \caption{Comparison of intrinsic dimension estimates. $\varepsilon$ is the scaling exponent from $\beta \sim k^{\varepsilon}$.
  $d_{\varepsilon} = 2/\varepsilon$ is the intrinsic dimension estimate from the scaling law. $d_{\text{MLE}}$ is the maximum likelihood
  estimate of intrinsic dimension from \cite{levina2005maximum}.}
  \label{tab:id}
  \vspace{0.4em}
  \begin{tabular}{lccc}
    \toprule
    Dataset & $\varepsilon$ & $\hat{d_{\varepsilon}}$ & $\hat{d}_{\text{MLE}}$ \\
    \midrule
    \textsc{MNIST} & 0.1449 & 13.80 & 14.25 \\
    \textsc{Fashion-MNIST} & 0.1808 & 11.06 & 15.64 \\
    \textsc{CIFAR-10} & 0.0862 & 23.21 & 28.39 \\
    \textsc{CIFAR-100} & 0.0961 & 20.81 & 25.69 \\
    \textsc{MNIST-CLIP} & 0.1759 & 11.37 & 14.43 \\
    \textsc{FMNIST-CLIP} & 0.1665 & 12.01 & 16.22 \\
    \textsc{CIFAR-10-CLIP} & 0.1101 & 18.16 & 22.05 \\
    \textsc{CIFAR-100-CLIP} & 0.1166 & 17.15 & 20.13 \\
    \textsc{Reddit} & 0.0622 & 32.13 & 29.22 \\
    \textsc{HAR} & 0.1894 & 10.56 & 14.75 \\
    \textsc{SUSY} & 0.2644 & 7.56 & 7.88 \\
    \textsc{StackExchange} & 0.0565 & 35.40 & 28.48 \\
    \bottomrule
  \end{tabular}
\end{minipage}

\end{table*}

\subsection{Intrinsic Dimension}

We test the hypothesis that the empirically estimated quantization exponent $\widehat{\varepsilon}$ is inversely proportional to the intrinsic dimension. Independent intrinsic dimension estimates $\widehat{d}_{\operatorname{MLE}}$ are obtained by using the popular maximum likelihood method (MLE) \cite{levina2005maximum}. For a dataset $\X = \{x_1,\dots,x_n\}$, let $T_j(x_i)$ denote the distance of the $j$th nearest neighbor of $x_i$ from $x_i$. For a hyperparameter $k$, the ID estimate is calculated as : 
$$\widehat{d}_{\operatorname{MLE}}(\X,k) := \frac{1}{n} \sum_{i=1}^n \left(\frac{1}{k-1} \sum_{j=1}^{k-1} \log \frac{T_k(x_i)}{T_j(x_i)}\right)^{-1}$$
For each dataset we take $10$ independent samples each  of size $n' = 10,000$ to estimate the ID for a particular value of $k$. This is repeated for $k \in \{5,10,20,50,100\}$ and the average is reported. We plot the estimates $\widehat{\eps}$ of the quantization exponent against $\widehat{d}_{\operatorname{MLE}}$ in Figure~\ref{fig:eps-id}, and the estimated dimension $\hat{d_\eps} = 2/\hat{\eps}$ against $\hat{d}_{\operatorname{MLE}}$. We observe significant agreement with Pearson's $R = 0.92$ and an $R^2$ score of $0.84$. It is to be noted that ID estimation in itself is a challenging problem with several proposed estimators in the literature. We refer the reader to the surveys \cite{camastra201626,binnie2025surveydimensionestimationmethods} for further discussion on ID estimators. We have chosen the MLE estimator of \cite{levina2005maximum} since it is one of the most established and widely studied estimators in the literature. 
\begin{figure*}[ht]
\centering
\begin{minipage}[t]{0.48\textwidth}
  \centering
  \includegraphics[width=0.8\linewidth]{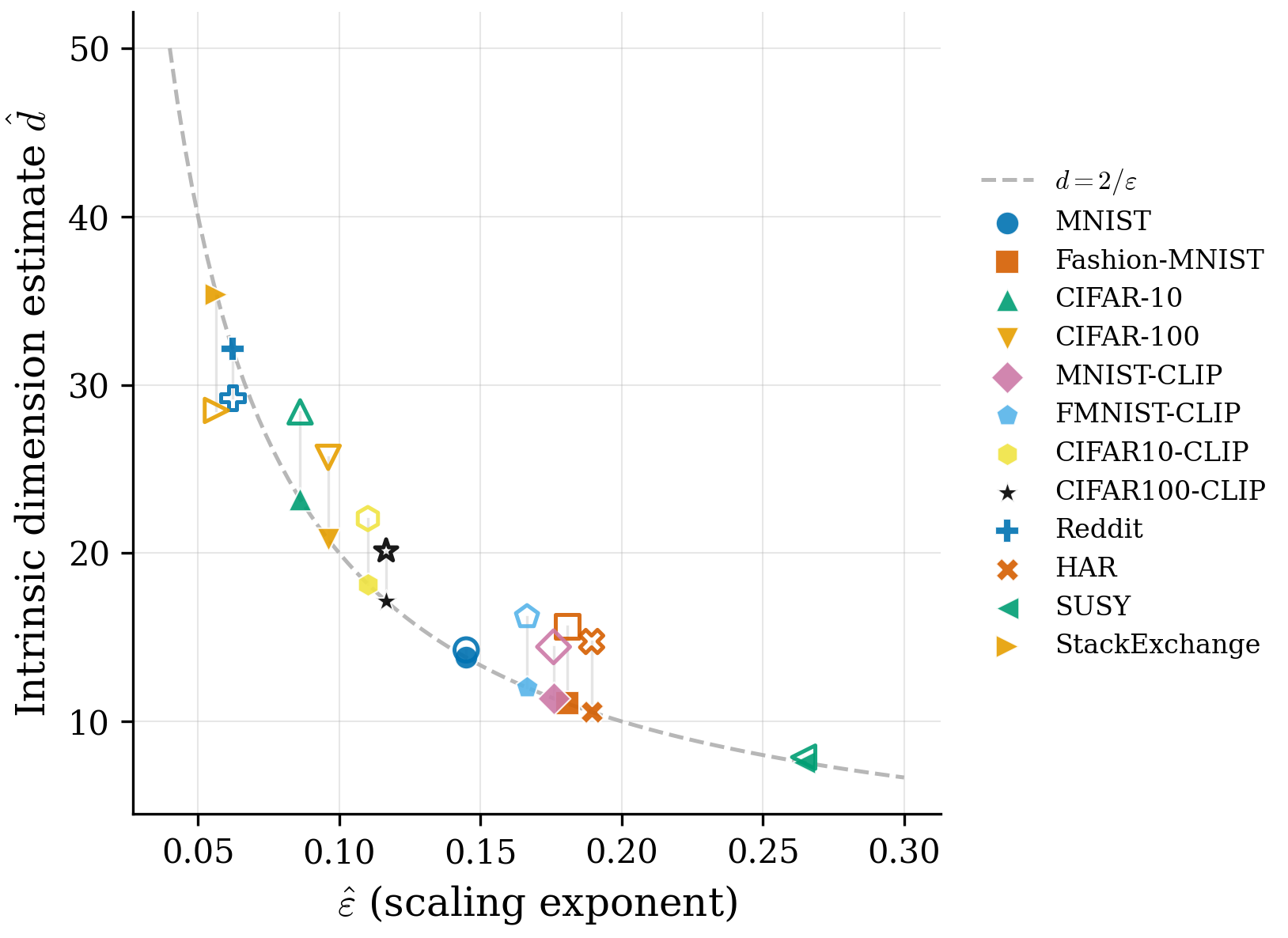}
  \caption{Relationship between $\varepsilon$ (from $\beta$ scaling) and $d_{\text{MLE}}$.}
  \label{fig:eps-id}
\end{minipage}
\hfill
\begin{minipage}[t]{0.48\textwidth}
  \centering
  \includegraphics[width=0.8\linewidth]{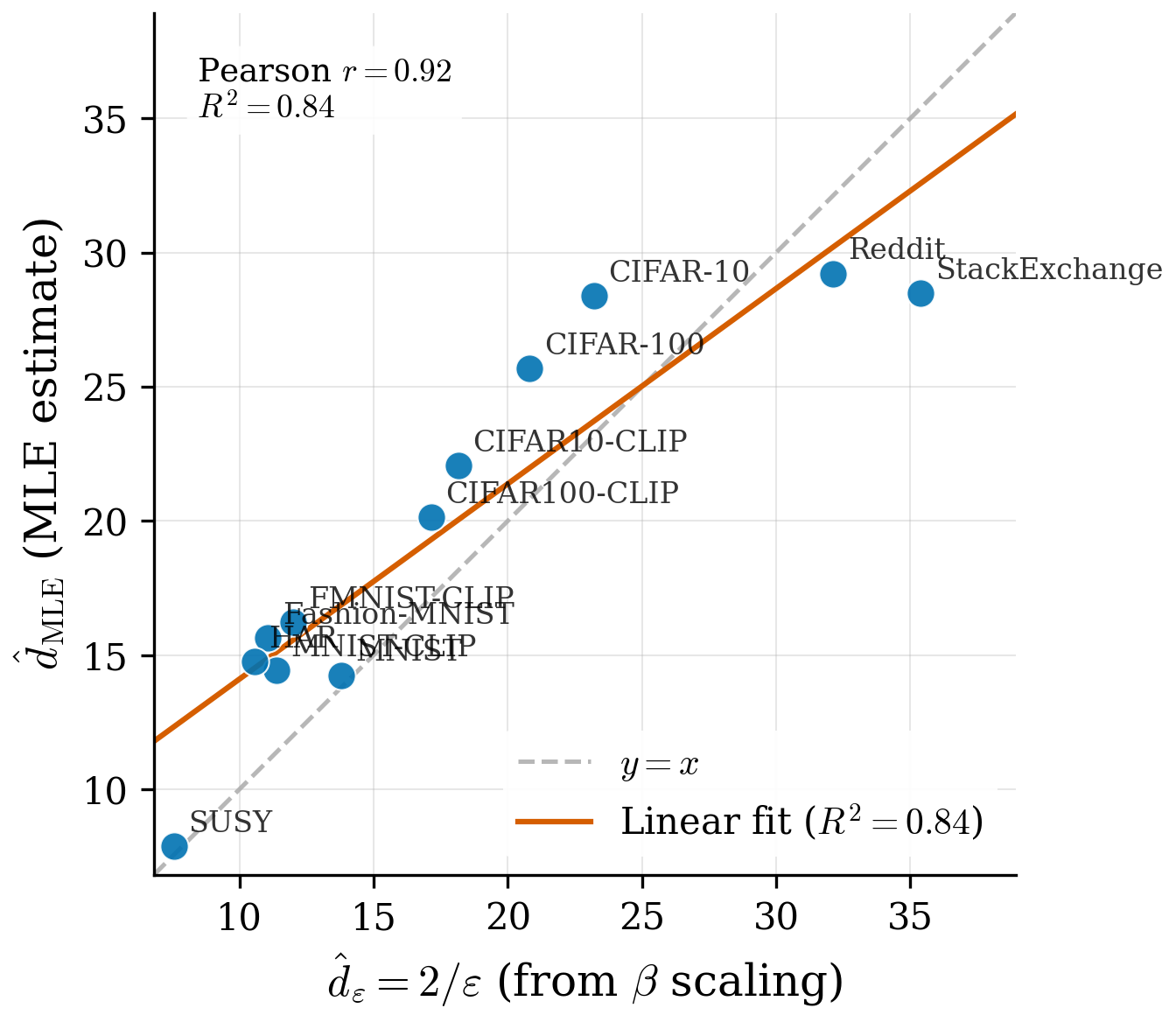}
  \caption{Comparison of $d_{\varepsilon}$ and $d_{\text{MLE}}$.}
  \label{fig:id}
\end{minipage}
\end{figure*}

\subsection{Effect of manifold stucture}

  To assess the robustness of the quantization scaling law $\beta_k \sim k^\varepsilon$ under distributional perturbation, we inject 
  isotropic Gaussian noise at 20 linearly-spaced noise-to-signal ratios (NSR) from 0 to 2 across four datasets (\textsc{MNIST}, \textsc{MNIST-CLIP}, \textsc{HAR}, 
  \textsc{StackExchange}), where $\X_{\operatorname{noisy}} = X + \text{NSR} \cdot \sigma(\X) \cdot \mathcal{N}(0,1)$, and re-fit the power-law exponent 
  $\varepsilon$ at each level. Across all datasets, $\varepsilon$ decays monotonically and convexly toward $2/D$, with the majority of 
  degradation occurring at low noise levels (NSR $< 0.5$), after which the exponent plateaus well above the ambient-dimension limit $2/D$. 
  The rate of decay varies with intrinsic structure: compact, low-dimensional representations such as MNIST-CLIP ($d \approx 11$, $D = 
  512$) lose most of their scaling signal by NSR $\approx 0.5$, while higher-dimensional datasets like StackExchange ($d \approx 39$, $D = 
  384$) degrade more gradually. This confirms that the scaling exponent captures genuine geometric structure rather than statistical 
  artifacts, and suggests that $\varepsilon$ could serve as a noise diagnostic — its sensitivity to perturbation is itself informative 
  about the signal-to-noise regime of the data. The plots are shown in Figure~\ref{fig:noisy}. 

\subsection{Rejection sampling analysis}

  Rejection Rate Experiment. We investigate the rejection sampling behavior of QKMEANS across 12 datasets by varying the chain length $m
  \in \{1, 2, 3, 5, 7, 10, 15, 20, 30, 50\}$ and number of clusters $k \in \{10, 50, 100\}$, averaging over 5 independent runs. For each center
  selection step, QKMEANS draws up to $m \ln(k+1)$ proposals from the distribution $\kappa(x) \propto \|x\|^2 + \|c_1\|^2$ and accepts a proposal
  with probability proportional to its distance to the nearest existing center; if all proposals are rejected, the algorithm falls back to
  uniform random selection. Low-dimensional and well-separated datasets (MNIST, FashionMNIST, SUSY) exhibit low rejection even at $m=1$ and reach
  zero rejection by $m \approx 5$–$10$. High-dimensional raw-pixel datasets (CIFAR-10, CIFAR-100) show moderate rejection at small $m$
  ($\sim$20–30) but converge to zero by $m \approx 60$–$80$. CLIP embeddings present the most challenging setting: FMNIST-CLIP and
  MNIST-CLIP maintain elevated rejection rates ($\sim$75–90 at $m=1$), with MNIST-CLIP retaining $\sim$10–15 rejection even at $m=50$,
  reflecting the tighter concentration of norm-based proposal distributions in these embedding spaces. In contrast, text embedding datasets
   (StackExchange, Reddit) exhibit near-zero rejection across all $m$ values, indicating that their $\kappa$ distribution closely
  approximates the $D^2$ distribution. Across all datasets, larger $k$ generally yields higher rejection fractions at a given $m$, as later
   center selections must distinguish among increasingly fine-grained cluster structure. These results suggest that $m \approx 10$–$20$
  suffices for most practical settings, while highly concentrated embedding datasets may benefit from larger chain lengths. The plots are shown in Figure~\ref{fig:rejection}.

\subsection{Effect of ANNS used }

Although our theoretical analysis uses Locality Sensitive Hashing for implementing nearest neighbor search, there are several other methods for doing the same. Our algorithm uses the ANNS algorithm as a black-box and hence, it can be easily replaced by a different method. We test several ANNS methods on the datasets to study the runtime-solution quality tradeoffs. We include brute-force search, hierarchicaly navigable small worlds \cite{malkov2020hnsw}, locality sensitive hasing \cite{indyk1998lsh}, inverted file index \cite{jegou2011product} and random projection trees \cite{dasgupta2008rptree}.For small $k \approx 50$, $\qkmeans$ is fastest with brute-force nearest-neighbor search, which outperforms LSH and tree-based ANNS by 2–4$\times$ even in very high dimensions, since ANNS overhead dominates at low k.
As k increases, this trend reverses: tree-based methods (e.g., RP-trees, HNSW) become increasingly advantageous by amortizing query costs across many seeding steps. The results are in Table~\ref{tab:anns-all}.  
\begin{table*}[htbp]
\centering
\caption{QKMEANS seeding time (ms) with different ANNS methods across all datasets and $k$ values.}
\label{tab:anns-all}
\small
\begin{tabular}{llrrrrr}
\toprule
\textsc{Dataset} & $k$ & \textsc{HNSW} & \textsc{BruteForce} & \textsc{LSH} & \textsc{IVF} & \textsc{RPTree} \\
\midrule

\multirow{3}{*}{\textsc{MNIST} ($d=784$)}
 & \(50\)   & 12.6  & 9.4   & 16.8   & 10.5  & 12.7 \\
 & \(200\)  & 43.3  & 24.8  & 46.6   & 22.1  & 23.9 \\
 & \(1000\) & 526.6 & 544.0 & 693.1  & 348.8 & 93.6 \\
\midrule

\multirow{3}{*}{\textsc{Fashion-MNIST} ($d=784$)}
 & \(50\)   & 14.8  & 8.8   & 19.4   & 9.1   & 12.5 \\
 & \(200\)  & 73.1  & 48.4  & 99.0   & 46.2  & 43.3 \\
 & \(1000\) & 852.6 & 1506.8 & 1961.1 & 1256.3 & 233.6 \\
\midrule

\multirow{3}{*}{\textsc{CIFAR-10} ($d=3072$)}
 & \(50\)   & 36.1  & 31.4  & 113.1  & 40.6  & 60.6 \\
 & \(200\)  & 341.5 & 301.6 & 681.5  & 327.4 & 231.3 \\
 & \(1000\) & 5242.8 & 8899.1 & 10951.5 & 8596.0 & 1396.1 \\
\midrule

\multirow{3}{*}{\textsc{CIFAR-100} ($d=3072$)}
 & \(50\)   & 31.8  & 29.3  & 105.3  & 33.4  & 52.3 \\
 & \(200\)  & 264.6 & 241.6 & 556.7  & 253.8 & 189.6 \\
 & \(1000\) & 4187.5 & 7065.1 & 8738.9 & 7012.6 & 1156.1 \\
\midrule

\multirow{3}{*}{\textsc{MNIST-CLIP} ($d=512$)}
 & \(50\)   & 33.0  & 15.6  & 78.4   & 23.8  & 34.5 \\
 & \(200\)  & 405.5 & 193.5 & 878.2  & 233.9 & 207.6 \\
 & \(1000\) & 4702.3 & 6664.2 & 26362.4 & 6985.7 & 1052.7 \\
\midrule

\multirow{3}{*}{\textsc{FMNIST-CLIP} ($d=512$)}
 & \(50\)   & 21.9  & 10.6  & 28.4   & 11.9  & 16.8 \\
 & \(200\)  & 155.7 & 76.6  & 299.7  & 123.7 & 72.2 \\
 & \(1000\) & 1828.8 & 2510.2 & 4885.6 & 2689.8 & 387.9 \\
\midrule

\multirow{3}{*}{\textsc{CIFAR10-CLIP} ($d=512$)}
 & \(50\)   & 11.8  & 5.5   & 16.2   & 11.8  & 13.2 \\
 & \(200\)  & 73.0  & 37.9  & 65.5   & 38.9  & 27.9 \\
 & \(1000\) & 735.5 & 825.2 & 960.7  & 856.4 & 141.6 \\
\midrule

\multirow{3}{*}{\textsc{CIFAR100-CLIP} ($d=512$)}
 & \(50\)   & 13.9  & 7.2   & 13.0   & 7.2   & 10.4 \\
 & \(200\)  & 59.6  & 31.3  & 52.7   & 30.4  & 24.8 \\
 & \(1000\) & 695.4 & 707.8 & 816.7  & 742.9 & 120.4 \\
\midrule

\multirow{3}{*}{\textsc{Reddit} ($d=384$)}
 & \(50\)   & 9.5   & 4.1   & 6.4    & 5.0   & 6.1 \\
 & \(200\)  & 21.2  & 12.2  & 14.4   & 9.7   & 9.7 \\
 & \(1000\) & 221.7 & 93.9  & 40.8   & 96.1  & 26.2 \\
\midrule

\multirow{3}{*}{\textsc{StackExchange} ($d=384$)}
 & \(50\)   & 7.6   & 3.9   & 7.5    & 5.9   & 5.0 \\
 & \(200\)  & 24.0  & 10.2  & 13.9   & 12.6  & 9.5 \\
 & \(1000\) & 225.4 & 84.8  & 39.1   & 87.6  & 24.4 \\
\midrule

\multirow{3}{*}{\textsc{HAR} ($d=561$)}
 & \(50\)   & 19.6  & 8.6   & 37.0   & 13.5  & 20.2 \\
 & \(200\)  & 171.8 & 107.8 & 212.9  & 113.1 & 87.5 \\
 & \(1000\) & 2186.7 & 3958.6 & 4748.6 & 3947.5 & 556.0 \\
\midrule

\multirow{3}{*}{\textsc{SUSY} ($d=19$)}
 & \(50\)   & 6.5   & 2.5   & 3.4    & 1.5   & 3.0 \\
 & \(200\)  & 35.8  & 6.1   & 12.1   & 6.7   & 5.3 \\
 & \(1000\) & 370.5 & 68.3  & 85.7   & 81.6  & 40.0 \\
\bottomrule
\end{tabular}
\vspace{1mm}
\end{table*}

\subsection{Comparison with Prior Art}

In our comparative analysis, we consider the following algorithms : 

\begin{itemize}
    \item $\afkmc$ from \cite{bachem2016good}
    \item $\rejsample$ from \cite{cohenaddad2020fast}
    \item $\pronecoreset$ from \cite{charikar2023simple}
    \item $\fastkmeans$ from \cite{draganov2024settle}
\end{itemize}

We followed the reference implementations from \href{https://github.com/adriangoe/afkmc2}{https://github.com/adriangoe/afkmc2} for $\afkmc$, from \href{https://git.informatik.uni-hamburg.de/0moin/google-research/-/tree/master}{https://git.informatik.uni-hamburg.de/0moin/google-research/-/tree/master}for $\rejsample$ , from \href{https://github.com/boredoms/prone}{https://github.com/boredoms/prone} for $\pronecoreset$ and from \href{https://github.com/Andrew-Draganov/Fast-Coreset-Generation}{https://github.com/Andrew-Draganov/Fast-Coreset-Generation} for $\fastkmeans$. We use the default hyperparameters which are specified in these implementations. For implementing $\qkmeans$, we use the HNSW ANNS algorithm. 

Figure~\ref{fig:tradeoff-plots} shows the seeding time and seeding quality tradeoffs for various algorithms across datasets. These results are also tabulated in Tables~\ref{tab:benchmark_mnist}-\ref{tab:benchmark_susy}.

\subsection{Summary of Empirical Findings}

Across all datasets, we observe consistent evidence for quantization-theoretic scaling, a strong inverse relationship between intrinsic dimension and scaling exponents, and stable rejection-sampling behavior. These results collectively support the central premise that geometric structure governs the empirical behavior of fast $k$-means seeding algorithms.

\begin{figure}[H]
    \centering
    \includegraphics[width=0.95\linewidth]{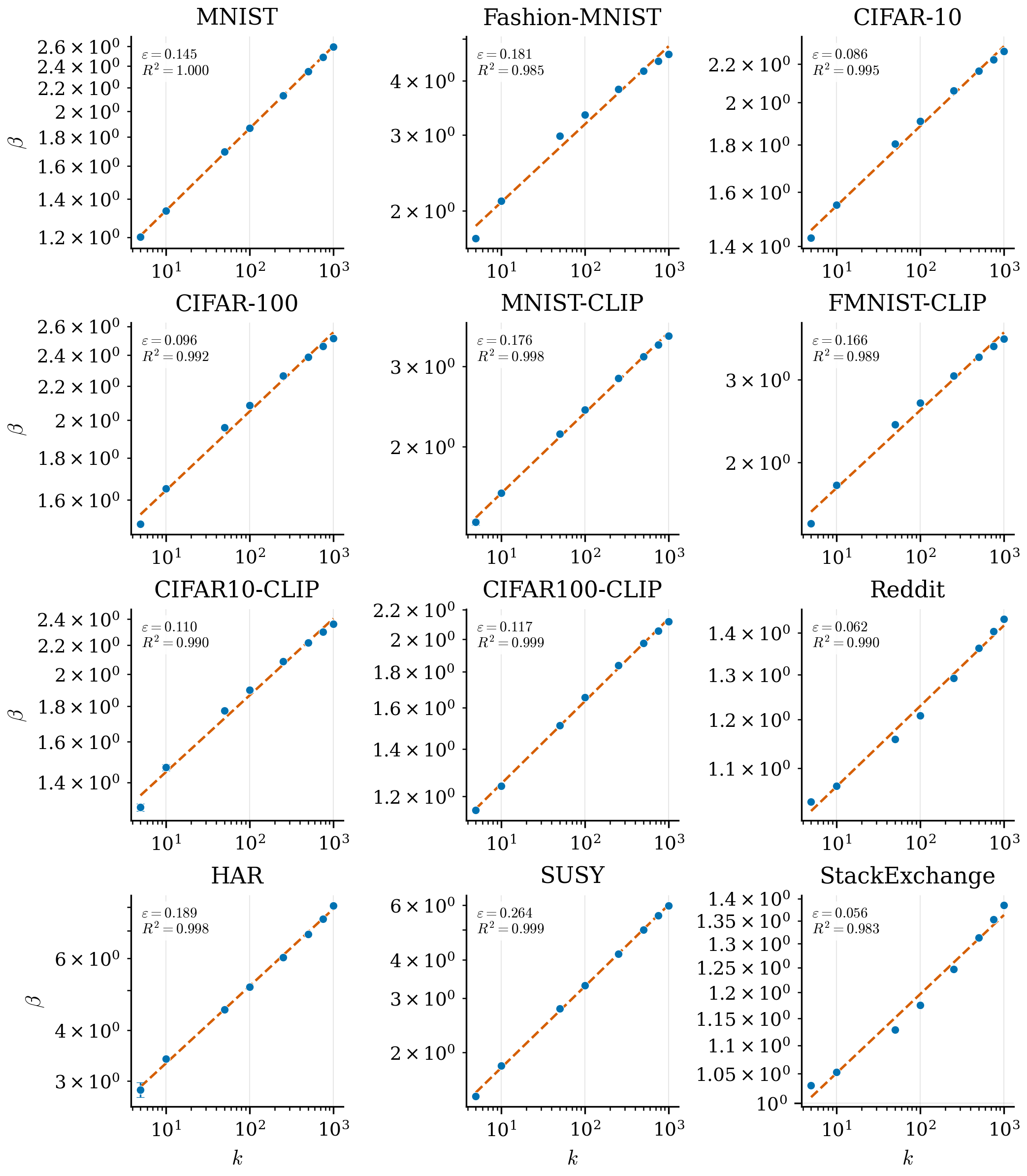}
    \caption{Scaling Laws for $\beta_k$ vs $k$}
    \label{fig:beta-grid}
\end{figure}

\begin{figure}[H]
    \centering
    \includegraphics[width=0.95\linewidth]{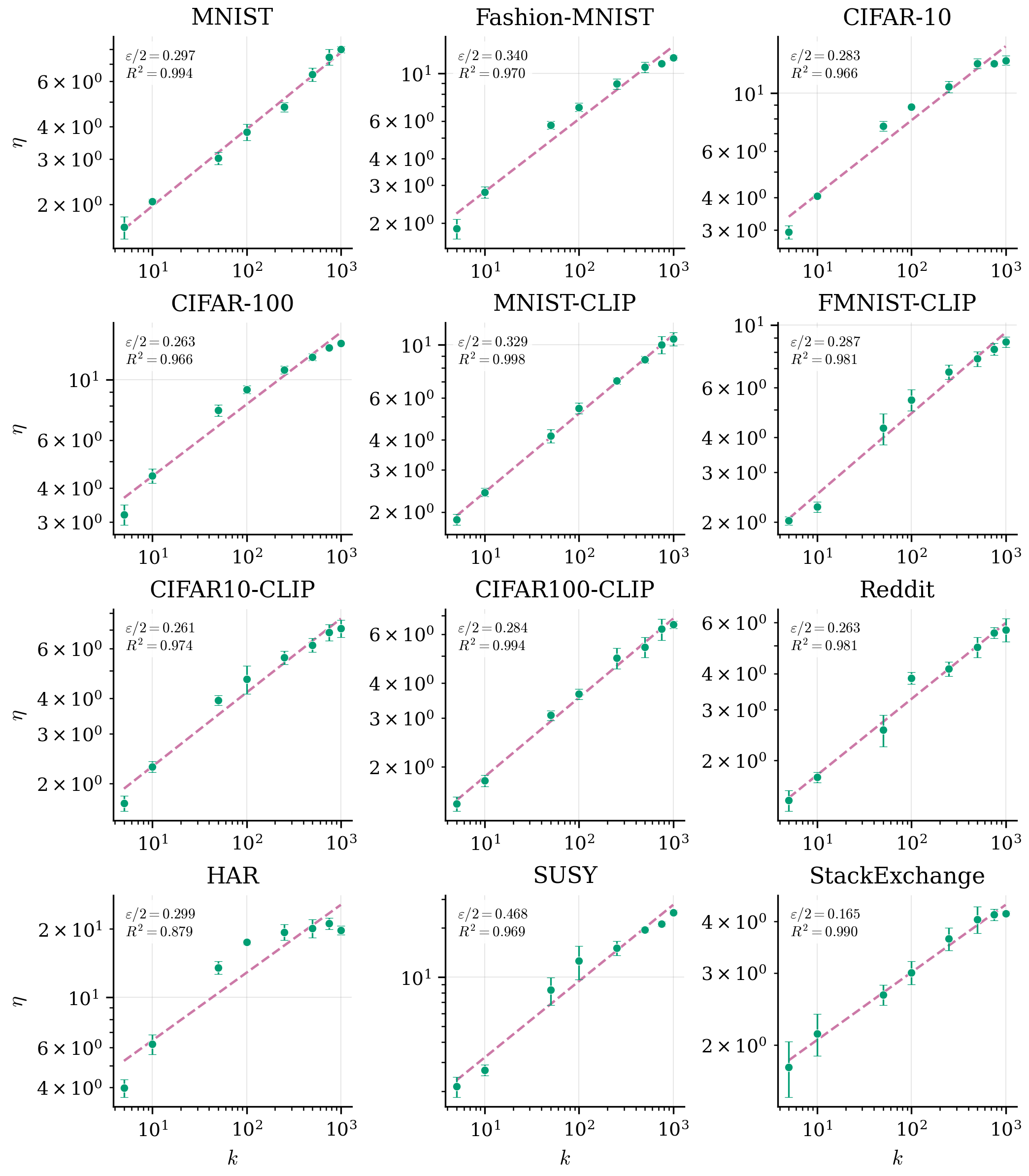}
    \caption{Scaling Laws for $\eta_k$ vs $k$}
    \label{fig:eta-grid}
\end{figure}

\begin{figure}
    \centering
    \includegraphics[width=\linewidth]{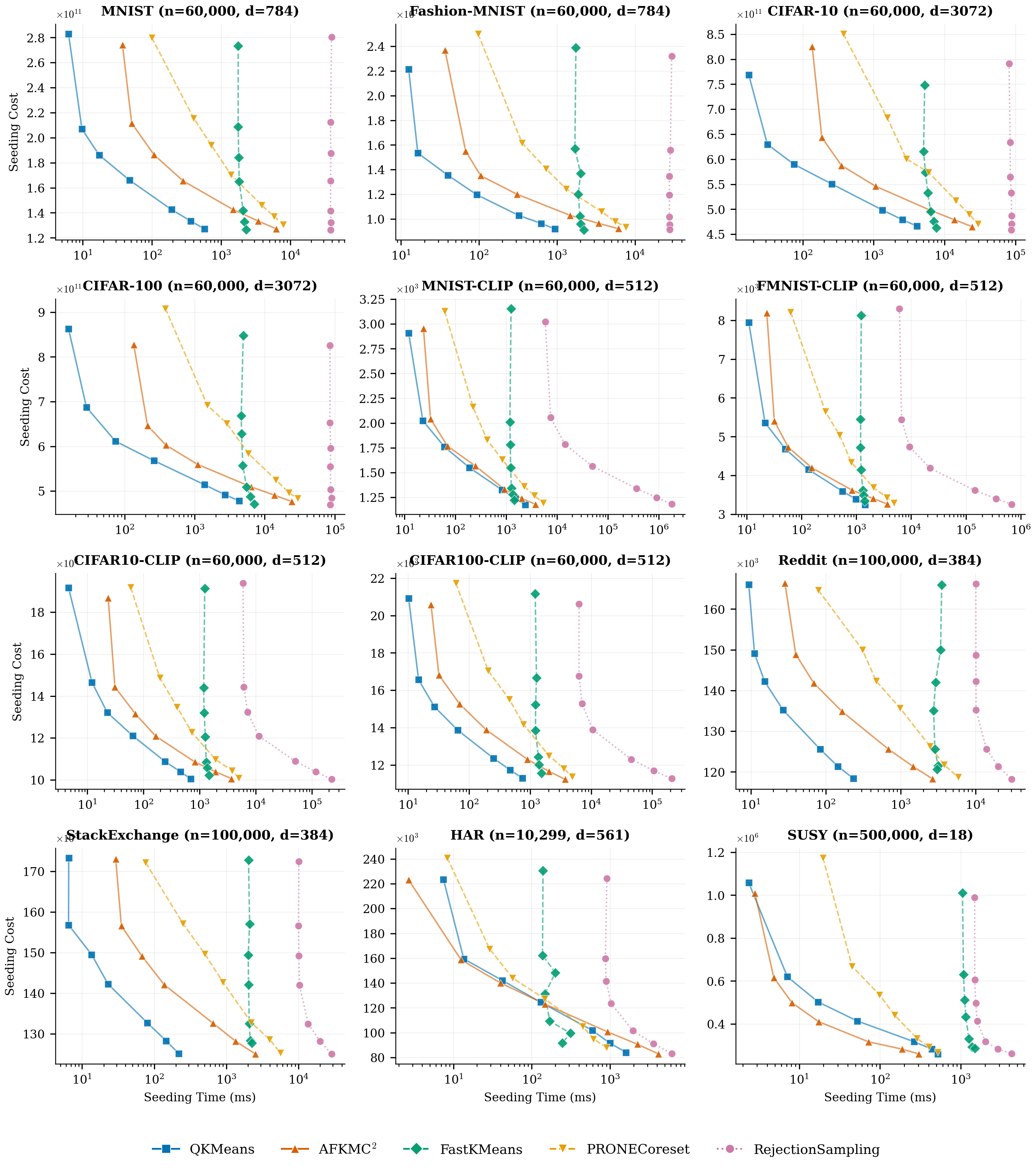}
    \caption{Plots for seeding time and seeding quality tradeoffs for various algorithms}
    \label{fig:tradeoff-plots}
\end{figure}

\begin{figure}
    \centering
    \includegraphics[width=\linewidth]{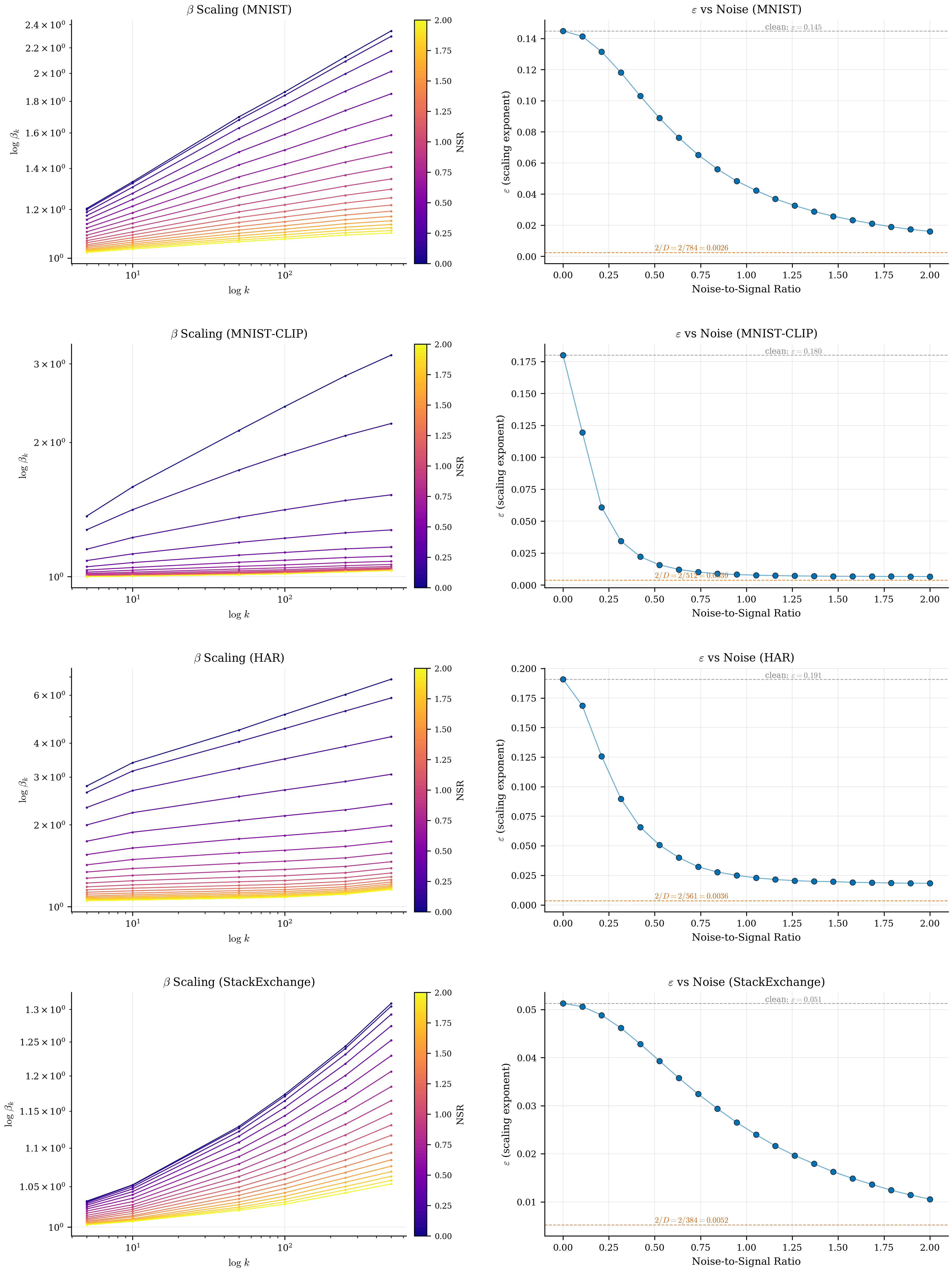}
    \caption{Effect of noise on the scaling laws, demonstrated on 4 datasets}
    \label{fig:noisy}
\end{figure}

\begin{figure}
    \centering
    \includegraphics[width=\linewidth]{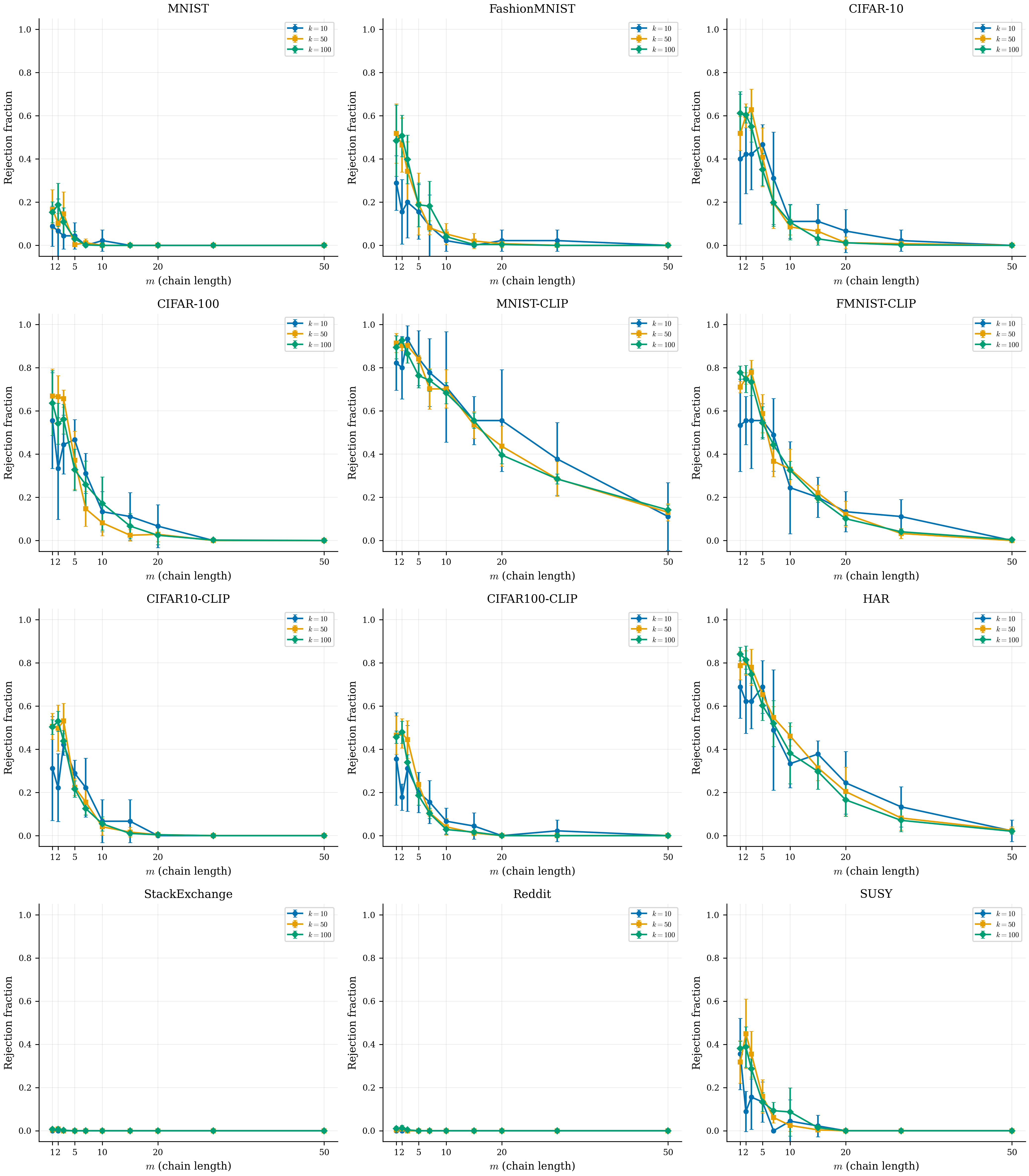}
    \caption{Fraction of samples rejected vs the chain length $m$}
    \label{fig:rejection}
\end{figure}


\begin{table}[htbp]
\centering
\small
\caption{Seeding benchmark on mnist. Time in ms, cost $\times 10^{11}$.}
\label{tab:benchmark_mnist}
\begin{tabular}{lrrrrrrrrrrrrrr}
\toprule
\textbf{Algorithm} & \multicolumn{2}{c}{$k=10$} & \multicolumn{2}{c}{$k=50$} & \multicolumn{2}{c}{$k=100$} & \multicolumn{2}{c}{$k=200$} & \multicolumn{2}{c}{$k=500$} & \multicolumn{2}{c}{$k=750$} & \multicolumn{2}{c}{$k=1000$} \\
\cmidrule(lr){2-3} \cmidrule(lr){4-5} \cmidrule(lr){6-7} \cmidrule(lr){8-9} \cmidrule(lr){10-11} \cmidrule(lr){12-13} \cmidrule(lr){14-15} 
 & Time & Cost & Time & Cost & Time & Cost & Time & Cost & Time & Cost & Time & Cost & Time & Cost \\
\midrule
$\afkmc$ & 38 & 2.74 & 51 & 2.12 & 107 & 1.86 & 281 & 1.65 & 1.5s & 1.43 & 3.4s & 1.33 & 6.2s & 1.27 \\
$\pronecoreset$ & 99 & 2.80 & 400 & 2.16 & 719 & 1.94 & 1.4s & 1.70 & 3.9s & 1.46 & 5.8s & 1.37 & 7.9s & 1.31 \\
$\fastkmeans$ & 1.8s & \textbf{2.73} & 1.8s & 2.09 & 1.8s & \textbf{1.84} & 1.8s & \textbf{1.65} & 2.1s & 1.42 & 2.2s & 1.33 & 2.3s & 1.27 \\
$\rejsample$ & 39.2s & 2.80 & 37.8s & 2.12 & 38.3s & 1.88 & 38.0s & 1.65 & 37.8s & \textbf{1.41} & 38.3s & \textbf{1.32} & 38.0s & \textbf{1.26} \\
$\qkmeans$ & \textbf{6.2} & 2.83 & \textbf{9.7} & \textbf{2.07} & \textbf{17} & 1.86 & 48 & 1.66 & 192 & 1.43 & 365 & 1.33 & 573 & 1.27 \\
\bottomrule
\end{tabular}
\end{table}

\begin{table}[htbp]
\centering
\small
\caption{Seeding benchmark on fmnist. Time in ms, cost $\times 10^{11}$.}
\label{tab:benchmark_fmnist}
\begin{tabular}{lrrrrrrrrrrrrrr}
\toprule
\textbf{Algorithm} & \multicolumn{2}{c}{$k=10$} & \multicolumn{2}{c}{$k=50$} & \multicolumn{2}{c}{$k=100$} & \multicolumn{2}{c}{$k=200$} & \multicolumn{2}{c}{$k=500$} & \multicolumn{2}{c}{$k=750$} & \multicolumn{2}{c}{$k=1000$} \\
\cmidrule(lr){2-3} \cmidrule(lr){4-5} \cmidrule(lr){6-7} \cmidrule(lr){8-9} \cmidrule(lr){10-11} \cmidrule(lr){12-13} \cmidrule(lr){14-15} 
 & Time & Cost & Time & Cost & Time & Cost & Time & Cost & Time & Cost & Time & Cost & Time & Cost \\
\midrule
$\afkmc$ & 37 & 2.37 & 67 & 1.55 & 105 & 1.35 & 308 & 1.20 & 1.5s & 1.03 & 3.4s & 0.96 & 6.1s & 0.92 \\
$\pronecoreset$ & 97 & 2.50 & 356 & 1.62 & 722 & 1.41 & 1.3s & 1.24 & 3.7s & 1.06 & 5.6s & 0.98 & 7.7s & 0.93 \\
$\fastkmeans$ & 1.7s & 2.39 & 1.7s & 1.57 & 2.0s & 1.37 & 1.9s & 1.20 & 2.0s & 1.02 & 2.0s & 0.96 & 2.2s & \textbf{0.91} \\
$\rejsample$ & 29.6s & 2.32 & 28.5s & 1.56 & 27.6s & \textbf{1.35} & 27.6s & \textbf{1.19} & 27.7s & \textbf{1.02} & 27.9s & \textbf{0.95} & 27.9s & 0.91 \\
$\qkmeans$ & \textbf{13} & \textbf{2.21} & \textbf{16} & \textbf{1.53} & 40 & 1.35 & 94 & 1.20 & 325 & 1.03 & 627 & 0.96 & 941 & 0.92 \\
\bottomrule
\end{tabular}
\end{table}

\begin{table}[htbp]
\centering
\small
\caption{Seeding benchmark on cifar10. Time in ms, cost $\times 10^{11}$.}
\label{tab:benchmark_cifar10}
\begin{tabular}{lrrrrrrrrrrrrrr}
\toprule
\textbf{Algorithm} & \multicolumn{2}{c}{$k=10$} & \multicolumn{2}{c}{$k=50$} & \multicolumn{2}{c}{$k=100$} & \multicolumn{2}{c}{$k=200$} & \multicolumn{2}{c}{$k=500$} & \multicolumn{2}{c}{$k=750$} & \multicolumn{2}{c}{$k=1000$} \\
\cmidrule(lr){2-3} \cmidrule(lr){4-5} \cmidrule(lr){6-7} \cmidrule(lr){8-9} \cmidrule(lr){10-11} \cmidrule(lr){12-13} \cmidrule(lr){14-15} 
 & Time & Cost & Time & Cost & Time & Cost & Time & Cost & Time & Cost & Time & Cost & Time & Cost \\
\midrule
$\afkmc$ & 136 & 8.25 & 186 & 6.44 & 352 & 5.87 & 1.1s & 5.46 & 6.2s & 4.99 & 13.8s & 4.79 & 24.6s & 4.65 \\
$\pronecoreset$ & 374 & 8.50 & 1.6s & 6.83 & 2.9s & 6.00 & 6.0s & 5.74 & 14.5s & 5.17 & 22.4s & 4.90 & 29.7s & 4.70 \\
$\fastkmeans$ & 5.3s & \textbf{7.48} & 5.1s & \textbf{6.15} & 5.4s & 5.74 & 5.8s & 5.33 & 6.4s & 4.95 & 7.1s & 4.76 & 7.7s & 4.63 \\
$\rejsample$ & 81.4s & 7.91 & 84.5s & 6.33 & 85.1s & \textbf{5.65} & 87.8s & \textbf{5.32} & 88.5s & \textbf{4.87} & 89.1s & \textbf{4.71} & 87.8s & \textbf{4.58} \\
$\qkmeans$ & \textbf{17} & 7.68 & \textbf{32} & 6.30 & \textbf{75} & 5.90 & \textbf{258} & 5.50 & \textbf{1.3s} & 4.98 & \textbf{2.5s} & 4.79 & \textbf{4.1s} & 4.66 \\
\bottomrule
\end{tabular}
\end{table}

\begin{table}[htbp]
\centering
\small
\caption{Seeding benchmark on cifar100. Time in ms, cost $\times 10^{11}$.}
\label{tab:benchmark_cifar100}
\begin{tabular}{lrrrrrrrrrrrrrr}
\toprule
\textbf{Algorithm} & \multicolumn{2}{c}{$k=10$} & \multicolumn{2}{c}{$k=50$} & \multicolumn{2}{c}{$k=100$} & \multicolumn{2}{c}{$k=200$} & \multicolumn{2}{c}{$k=500$} & \multicolumn{2}{c}{$k=750$} & \multicolumn{2}{c}{$k=1000$} \\
\cmidrule(lr){2-3} \cmidrule(lr){4-5} \cmidrule(lr){6-7} \cmidrule(lr){8-9} \cmidrule(lr){10-11} \cmidrule(lr){12-13} \cmidrule(lr){14-15} 
 & Time & Cost & Time & Cost & Time & Cost & Time & Cost & Time & Cost & Time & Cost & Time & Cost \\
\midrule
$\afkmc$ & 135 & 8.27 & 212 & \textbf{6.46} & 389 & 6.03 & 1.1s & 5.59 & 6.4s & 5.09 & 13.8s & 4.91 & 24.4s & 4.77 \\
$\pronecoreset$ & 382 & 9.08 & 1.5s & 6.92 & 2.9s & 6.51 & 5.9s & 5.84 & 14.5s & 5.25 & 22.2s & 4.96 & 29.7s & 4.83 \\
$\fastkmeans$ & 4.9s & 8.48 & 4.6s & 6.69 & 4.7s & 6.28 & 4.9s & 5.57 & 5.5s & 5.09 & 6.3s & 4.87 & 7.2s & 4.71 \\
$\rejsample$ & 85.9s & \textbf{8.26} & 85.5s & 6.53 & 87.7s & \textbf{5.96} & 86.7s & \textbf{5.54} & 87.9s & \textbf{5.04} & 90.2s & \textbf{4.84} & 86.6s & \textbf{4.70} \\
$\qkmeans$ & \textbf{16} & 8.63 & \textbf{29} & 6.87 & \textbf{74} & 6.11 & \textbf{262} & 5.68 & \textbf{1.4s} & 5.14 & \textbf{2.7s} & 4.92 & \textbf{4.3s} & 4.78 \\
\bottomrule
\end{tabular}
\end{table}

\begin{table}[htbp]
\centering
\small
\caption{Seeding benchmark on mnist\_clip. Time in ms, cost .}
\label{tab:benchmark_mnist_clip}
\begin{tabular}{lrrrrrrrrrrrrrr}
\toprule
\textbf{Algorithm} & \multicolumn{2}{c}{$k=10$} & \multicolumn{2}{c}{$k=50$} & \multicolumn{2}{c}{$k=100$} & \multicolumn{2}{c}{$k=200$} & \multicolumn{2}{c}{$k=500$} & \multicolumn{2}{c}{$k=750$} & \multicolumn{2}{c}{$k=1000$} \\
\cmidrule(lr){2-3} \cmidrule(lr){4-5} \cmidrule(lr){6-7} \cmidrule(lr){8-9} \cmidrule(lr){10-11} \cmidrule(lr){12-13} \cmidrule(lr){14-15} 
 & Time & Cost & Time & Cost & Time & Cost & Time & Cost & Time & Cost & Time & Cost & Time & Cost \\
\midrule
$\afkmc$ & 24 & 2955 & 33 & 2041 & 70 & 1768 & 248 & 1566 & 927 & 1330 & 2.0s & \textbf{1239} & 3.8s & 1178 \\
$\pronecoreset$ & 62 & 3131 & 224 & 2162 & 422 & 1830 & 838 & 1629 & 2.3s & 1361 & 3.6s & 1265 & 5.4s & 1191 \\
$\fastkmeans$ & 1.3s & 3156 & 1.2s & \textbf{2009} & 1.2s & 1781 & 1.2s & \textbf{1550} & 1.3s & 1340 & \textbf{1.4s} & 1279 & \textbf{1.5s} & 1221 \\
$\rejsample$ & 5.9s & 3021 & 7.5s & 2055 & 14.4s & 1783 & 49.9s & 1563 & 363.8s & 1340 & 912.2s & 1246 & 1800.0s & 1180 \\
$\qkmeans$ & \textbf{12} & \textbf{2907} & \textbf{23} & 2024 & \textbf{61} & \textbf{1759} & \textbf{190} & \textbf{1550} & \textbf{828} & \textbf{1325} & 1.5s & \textbf{1239} & 2.4s & \textbf{1175} \\
\bottomrule
\end{tabular}
\end{table}

\begin{table}[htbp]
\centering
\small
\caption{Seeding benchmark on fmnist\_clip. Time in ms, cost .}
\label{tab:benchmark_fmnist_clip}
\begin{tabular}{lrrrrrrrrrrrrrr}
\toprule
\textbf{Algorithm} & \multicolumn{2}{c}{$k=10$} & \multicolumn{2}{c}{$k=50$} & \multicolumn{2}{c}{$k=100$} & \multicolumn{2}{c}{$k=200$} & \multicolumn{2}{c}{$k=500$} & \multicolumn{2}{c}{$k=750$} & \multicolumn{2}{c}{$k=1000$} \\
\cmidrule(lr){2-3} \cmidrule(lr){4-5} \cmidrule(lr){6-7} \cmidrule(lr){8-9} \cmidrule(lr){10-11} \cmidrule(lr){12-13} \cmidrule(lr){14-15} 
 & Time & Cost & Time & Cost & Time & Cost & Time & Cost & Time & Cost & Time & Cost & Time & Cost \\
\midrule
$\afkmc$ & 23 & 8185 & 32 & 5407 & 57 & 4732 & 153 & 4201 & 839 & 3620 & 2.0s & 3407 & 3.7s & 3262 \\
$\pronecoreset$ & 63 & 8210 & 272 & 5646 & 498 & 5038 & 812 & 4339 & 2.0s & 3689 & 3.6s & 3430 & 4.9s & 3294 \\
$\fastkmeans$ & 1.2s & 8128 & 1.2s & 5445 & 1.2s & 4720 & 1.2s & \textbf{4143} & 1.3s & 3622 & 1.4s & 3490 & \textbf{1.4s} & 3348 \\
$\rejsample$ & 6.1s & 8301 & 6.7s & 5437 & 9.3s & 4742 & 22.1s & 4188 & 145.5s & 3616 & 355.1s & 3401 & 685.6s & 3257 \\
$\qkmeans$ & \textbf{11} & \textbf{7944} & \textbf{22} & \textbf{5361} & \textbf{51} & \textbf{4682} & \textbf{134} & 4155 & \textbf{555} & \textbf{3596} & \textbf{976} & \textbf{3394} & 1.4s & \textbf{3249} \\
\bottomrule
\end{tabular}
\end{table}

\begin{table}[htbp]
\centering
{\small
\caption{Seeding benchmark on cifar10\_clip. Time in ms, cost $\times 10^2$.}
\label{tab:benchmark_cifar10_clip}
\begin{tabular}{lrrrrrrrrrrrrrr}
\toprule
\textbf{Algorithm} & \multicolumn{2}{c}{$k=10$} & \multicolumn{2}{c}{$k=50$} & \multicolumn{2}{c}{$k=100$} & \multicolumn{2}{c}{$k=200$} & \multicolumn{2}{c}{$k=500$} & \multicolumn{2}{c}{$k=750$} & \multicolumn{2}{c}{$k=1000$} \\
\cmidrule(lr){2-3} \cmidrule(lr){4-5} \cmidrule(lr){6-7} \cmidrule(lr){8-9} \cmidrule(lr){10-11} \cmidrule(lr){12-13} \cmidrule(lr){14-15} 
 & Time & Cost & Time & Cost & Time & Cost & Time & Cost & Time & Cost & Time & Cost & Time & Cost \\
\midrule
$\afkmc$ & 24 & \textbf{187} & 31 & 144 & 72 & \textbf{131} & 165 & 121 & 823 & 108 & 1.9s & \textbf{104} & 3.7s & 101 \\
$\pronecoreset$ & 59 & 192 & 195 & 149 & 396 & 135 & 720 & 123 & 1.9s & 110 & 3.8s & 104 & 5.0s & 101 \\
$\fastkmeans$ & 1.2s & 191 & 1.2s & \textbf{144} & 1.2s & 132 & 1.3s & \textbf{120} & 1.3s & \textbf{108} & 1.4s & 106 & 1.5s & 102 \\
$\rejsample$ & 5.9s & 194 & 6.1s & 144 & 7.2s & 132 & 11.4s & 121 & 50.5s & 109 & 117.1s & 104 & 224.0s & \textbf{100} \\
$\qkmeans$ & \textbf{4.6} & 192 & \textbf{12} & 146 & \textbf{23} & 132 & \textbf{65} & 121 & \textbf{240} & 109 & \textbf{453} & 104 & \textbf{695} & 100 \\
\bottomrule
\end{tabular}}
\end{table}

\begin{table}[htbp]
\centering
\small
\caption{Seeding benchmark on cifar100\_clip. Time in ms, cost $\times 10^2$.}
\label{tab:benchmark_cifar100_clip}
\begin{tabular}{lrrrrrrrrrrrrrr}
\toprule
\textbf{Algorithm} & \multicolumn{2}{c}{$k=10$} & \multicolumn{2}{c}{$k=50$} & \multicolumn{2}{c}{$k=100$} & \multicolumn{2}{c}{$k=200$} & \multicolumn{2}{c}{$k=500$} & \multicolumn{2}{c}{$k=750$} & \multicolumn{2}{c}{$k=1000$} \\
\cmidrule(lr){2-3} \cmidrule(lr){4-5} \cmidrule(lr){6-7} \cmidrule(lr){8-9} \cmidrule(lr){10-11} \cmidrule(lr){12-13} \cmidrule(lr){14-15} 
 & Time & Cost & Time & Cost & Time & Cost & Time & Cost & Time & Cost & Time & Cost & Time & Cost \\
\midrule
$\afkmc$ & 24 & \textbf{206} & 32 & 168 & 70 & 153 & 191 & 139 & 906 & \textbf{123} & 2.0s & \textbf{117} & 3.7s & \textbf{112} \\
$\pronecoreset$ & 61 & 218 & 206 & 170 & 455 & 155 & 782 & 142 & 2.0s & 125 & 3.6s & 118 & 4.9s & 114 \\
$\fastkmeans$ & 1.2s & 212 & 1.3s & 167 & 1.2s & 152 & 1.2s & \textbf{138} & 1.4s & 124 & 1.4s & 120 & 1.5s & 116 \\
$\rejsample$ & 6.3s & 206 & 6.3s & 168 & 7.1s & 153 & 10.6s & 139 & 45.4s & 123 & 106.6s & 117 & 208.3s & 113 \\
$\qkmeans$ & \textbf{10} & 209 & \textbf{15} & \textbf{166} & \textbf{27} & \textbf{151} & \textbf{65} & 139 & \textbf{252} & 124 & \textbf{470} & 117 & \textbf{745} & 113 \\
\bottomrule
\end{tabular}
\end{table}

\begin{table}[htbp]
\centering
\small
\caption{Seeding benchmark on reddit. Time in ms, cost $\times 10^2$.}
\label{tab:benchmark_reddit}
\begin{tabular}{lrrrrrrrrrrrrrr}
\toprule
\textbf{Algorithm} & \multicolumn{2}{c}{$k=10$} & \multicolumn{2}{c}{$k=50$} & \multicolumn{2}{c}{$k=100$} & \multicolumn{2}{c}{$k=200$} & \multicolumn{2}{c}{$k=500$} & \multicolumn{2}{c}{$k=750$} & \multicolumn{2}{c}{$k=1000$} \\
\cmidrule(lr){2-3} \cmidrule(lr){4-5} \cmidrule(lr){6-7} \cmidrule(lr){8-9} \cmidrule(lr){10-11} \cmidrule(lr){12-13} \cmidrule(lr){14-15} 
 & Time & Cost & Time & Cost & Time & Cost & Time & Cost & Time & Cost & Time & Cost & Time & Cost \\
\midrule
$\afkmc$ & 29 & 1664 & 40 & 1489 & 69 & \textbf{1418} & 165 & \textbf{1348} & 687 & 1256 & 1.5s & 1214 & 2.7s & 1183 \\
$\pronecoreset$ & 80 & \textbf{1647} & 309 & 1500 & 472 & 1423 & 982 & 1357 & 2.5s & 1263 & 3.8s & 1218 & 5.9s & 1187 \\
$\fastkmeans$ & 3.5s & 1659 & 3.4s & 1500 & 2.9s & 1420 & 2.8s & 1350 & 2.9s & 1256 & 3.2s & 1215 & 3.1s & 1207 \\
$\rejsample$ & 10.1s & 1662 & 10.1s & \textbf{1486} & 10.2s & 1423 & 10.2s & 1352 & 14.0s & \textbf{1255} & 20.2s & \textbf{1213} & 30.5s & \textbf{1182} \\
$\qkmeans$ & \textbf{9.4} & 1660 & \textbf{11} & 1491 & \textbf{15} & 1422 & \textbf{27} & 1352 & \textbf{85} & 1256 & \textbf{146} & 1213 & \textbf{235} & 1183 \\
\bottomrule
\end{tabular}
\end{table}

\begin{table}[htbp]
\centering
\small
\caption{Seeding benchmark on stackexchange. Time in ms, cost $\times 10^2$.}
\label{tab:benchmark_stackexchange}
\begin{tabular}{lrrrrrrrrrrrrrr}
\toprule
\textbf{Algorithm} & \multicolumn{2}{c}{$k=10$} & \multicolumn{2}{c}{$k=50$} & \multicolumn{2}{c}{$k=100$} & \multicolumn{2}{c}{$k=200$} & \multicolumn{2}{c}{$k=500$} & \multicolumn{2}{c}{$k=750$} & \multicolumn{2}{c}{$k=1000$} \\
\cmidrule(lr){2-3} \cmidrule(lr){4-5} \cmidrule(lr){6-7} \cmidrule(lr){8-9} \cmidrule(lr){10-11} \cmidrule(lr){12-13} \cmidrule(lr){14-15} 
 & Time & Cost & Time & Cost & Time & Cost & Time & Cost & Time & Cost & Time & Cost & Time & Cost \\
\midrule
$\afkmc$ & 29 & 1730 & 35 & \textbf{1566} & 68 & \textbf{1492} & 137 & 1421 & 655 & 1326 & 1.3s & 1282 & 2.5s & \textbf{1251} \\
$\pronecoreset$ & 76 & \textbf{1722} & 250 & 1572 & 504 & 1497 & 900 & 1427 & 2.2s & 1328 & 4.0s & 1286 & 5.6s & 1253 \\
$\fastkmeans$ & 2.0s & 1728 & 2.1s & 1571 & 2.0s & 1494 & 2.0s & 1421 & 2.1s & 1326 & 2.1s & 1284 & 2.3s & 1277 \\
$\rejsample$ & 10.1s & 1724 & 10.0s & 1566 & 10.1s & 1492 & 10.3s & \textbf{1420} & 13.6s & \textbf{1325} & 19.9s & \textbf{1282} & 28.7s & 1251 \\
$\qkmeans$ & \textbf{6.5} & 1733 & \textbf{6.5} & 1568 & \textbf{14} & 1495 & \textbf{23} & 1423 & \textbf{80} & 1328 & \textbf{146} & 1283 & \textbf{220} & 1251 \\
\bottomrule
\end{tabular}
\end{table}

\begin{table}[htbp]
\centering
\small
\caption{Seeding benchmark on har. Time in ms, cost $\times 10^2$.}
\label{tab:benchmark_har}
\begin{tabular}{lrrrrrrrrrrrrrr}
\toprule
\textbf{Algorithm} & \multicolumn{2}{c}{$k=10$} & \multicolumn{2}{c}{$k=50$} & \multicolumn{2}{c}{$k=100$} & \multicolumn{2}{c}{$k=200$} & \multicolumn{2}{c}{$k=500$} & \multicolumn{2}{c}{$k=750$} & \multicolumn{2}{c}{$k=1000$} \\
\cmidrule(lr){2-3} \cmidrule(lr){4-5} \cmidrule(lr){6-7} \cmidrule(lr){8-9} \cmidrule(lr){10-11} \cmidrule(lr){12-13} \cmidrule(lr){14-15} 
 & Time & Cost & Time & Cost & Time & Cost & Time & Cost & Time & Cost & Time & Cost & Time & Cost \\
\midrule
$\afkmc$ & \textbf{2.7} & \textbf{2230} & \textbf{13} & \textbf{1587} & \textbf{40} & \textbf{1401} & 148 & \textbf{1230} & 931 & \textbf{1007} & 2.2s & \textbf{907} & 4.2s & \textbf{829} \\
$\pronecoreset$ & 8.3 & 2408 & 29 & 1674 & 57 & 1440 & 147 & 1273 & 449 & 1049 & 617 & 949 & 905 & 880 \\
$\fastkmeans$ & 140 & 2305 & 138 & 1623 & 201 & 1483 & 147 & 1312 & \textbf{170} & 1094 & \textbf{313} & 996 & \textbf{247} & 917 \\
$\rejsample$ & 913 & 2244 & 873 & 1596 & 895 & 1415 & 1.0s & 1236 & 2.0s & 1016 & 3.6s & 911 & 6.1s & 832 \\
$\qkmeans$ & 7.5 & 2235 & 14 & 1594 & 42 & 1421 & \textbf{130} & 1245 & 595 & 1018 & 1.0s & 916 & 1.6s & 841 \\
\bottomrule
\end{tabular}
\end{table}

\begin{table}[htbp]
\centering
\small
\caption{Seeding benchmark on susy. Time in ms, cost $\times 10^{6}$.}
\label{tab:benchmark_susy}
\begin{tabular}{lrrrrrrrrrrrrrr}
\toprule
\textbf{Algorithm} & \multicolumn{2}{c}{$k=10$} & \multicolumn{2}{c}{$k=50$} & \multicolumn{2}{c}{$k=100$} & \multicolumn{2}{c}{$k=200$} & \multicolumn{2}{c}{$k=500$} & \multicolumn{2}{c}{$k=750$} & \multicolumn{2}{c}{$k=1000$} \\
\cmidrule(lr){2-3} \cmidrule(lr){4-5} \cmidrule(lr){6-7} \cmidrule(lr){8-9} \cmidrule(lr){10-11} \cmidrule(lr){12-13} \cmidrule(lr){14-15} 
 & Time & Cost & Time & Cost & Time & Cost & Time & Cost & Time & Cost & Time & Cost & Time & Cost \\
\midrule
$\afkmc$ & 2.8 & 1.01 & \textbf{4.8} & 0.62 & \textbf{8.1} & 0.50 & \textbf{17} & \textbf{0.41} & \textbf{72} & \textbf{0.32} & \textbf{187} & \textbf{0.28} & \textbf{302} & 0.26 \\
$\pronecoreset$ & 20 & 1.17 & 45 & 0.67 & 99 & 0.54 & 152 & 0.44 & 288 & 0.33 & 404 & 0.29 & 521 & 0.27 \\
$\fastkmeans$ & 1.1s & 1.01 & 1.1s & 0.63 & 1.1s & 0.51 & 1.2s & 0.43 & 1.3s & 0.33 & 1.4s & 0.29 & 1.5s & 0.29 \\
$\rejsample$ & 1.5s & \textbf{0.99} & 1.5s & \textbf{0.61} & 1.5s & \textbf{0.50} & 1.6s & 0.41 & 2.0s & 0.32 & 2.9s & 0.28 & 4.3s & 0.26 \\
$\qkmeans$ & \textbf{2.4} & 1.06 & 7.1 & 0.62 & 17 & 0.50 & 52 & 0.41 & 264 & 0.32 & 437 & 0.28 & 518 & \textbf{0.26} \\
\bottomrule
\end{tabular}
\end{table}


\end{document}